\documentclass[a4paper,UKenglish,cleveref, autoref, thm-restate]{lipics-v2021}

\nolinenumbers

\usepackage{amssymb}
\usepackage{amsthm}
\usepackage{amsmath}
\usepackage{graphicx}
\usepackage[utf8]{inputenc}
\usepackage{algorithm}
\usepackage[noend]{algpseudocode}
\usepackage{enumitem}
\usepackage[nolist]{acronym}

\usepackage{thmtools}
\usepackage{thm-restate}
\usepackage[usenames,dvipsnames,svgnames,table]{xcolor}

\usepackage{todonotes}
\setdescription{leftmargin=0em,labelindent=0\parindent,labelsep = 1em, itemsep=0.1\baselineskip,topsep=0.1\baselineskip,itemsep =\topsep,  listparindent=\parindent, font=\normalfont\normalsize\itshape}

\newcounter{steplistcount}
\newcounter{steplistcounti}
\newcounter{caselistcount}
\newcounter{caselistcounti}
\newcounter{caselistcountii}
\renewcommand*\thesteplistcount{\arabic{steplistcount}}
\renewcommand*\thesteplistcounti{\arabic{steplistcounti}}
\renewcommand*\thecaselistcount{\arabic{caselistcount}}
\renewcommand*\thecaselistcounti{\arabic{caselistcounti}}
\renewcommand*\thecaselistcountii{\arabic{caselistcountii}}
\newlist{stepList}{description}{2}
\setlist[stepList,1]{%
	before={\setcounter{steplistcount}{0}},
	leftmargin=0em,labelindent=0\parindent,labelsep = \parindent, topsep=0.2\baselineskip, itemsep=\topsep, listparindent=\parindent,
	,font=\normalfont\normalsize\itshape Step~\stepcounter{steplistcount}\thesteplistcount:~ 
}
\setlist[stepList,2]{%
	before={\setcounter{steplistcounti}{0}},%
	leftmargin=0em,labelindent=0\parindent,labelsep = \parindent, itemsep=0.0\baselineskip,topsep=0.0\baselineskip,itemsep =\topsep,  listparindent=\parindent,%
	,font=\normalfont\normalsize\itshape Step~ \stepcounter{steplistcounti}\thesteplistcount.\thesteplistcounti:~ 
}

\newlist{caseList}{description}{3}
\setlist[caseList,1]{%
	before={\setcounter{caselistcount}{0}},%
	leftmargin=0em,labelindent=0\parindent,labelsep = 1em, itemsep=0.1\baselineskip,topsep=0.1\baselineskip,itemsep =\topsep, listparindent=\parindent,%
	,font=\normalfont\normalsize\itshape Case~\stepcounter{caselistcount}\thecaselistcount:~
}

\setlist[caseList,2]{%
	before={\setcounter{caselistcounti}{0}},%
	leftmargin=0em,labelindent=0\parindent,labelsep = 1em, itemsep=0.1\baselineskip,topsep=0.1\baselineskip,itemsep =\topsep,  listparindent=\parindent,%
	,font=\normalfont\normalsize\itshape Case~ \stepcounter{caselistcounti}\thecaselistcount.\thecaselistcounti:~ %
}

\setlist[caseList,3]{%
	before={\setcounter{caselistcountii}{0}},%
	leftmargin=0em,labelindent=0\parindent,labelsep = 1em, itemsep=0.1\baselineskip,topsep=0.1\baselineskip,itemsep =\topsep,  listparindent=\parindent,%
	,font=\normalfont\normalsize\itshape Case~ \stepcounter{caselistcountii}\thecaselistcount.\thecaselistcounti.\thecaselistcountii:~ %
}

\usepackage{tikz}
\usetikzlibrary{patterns}
\usetikzlibrary{decorations,arrows}
\usetikzlibrary{decorations.pathmorphing}
\usepgflibrary{decorations.pathreplacing} 

\newboolean{BlackAndWhite}
\setboolean{BlackAndWhite}{true}
\newcommand{\drawVerticalItem}[5][$ $]{%
	\ifthenelse{\boolean{BlackAndWhite}}{%
		\draw[fill = white!85!black, fill opacity = 0.7] (#2,#3) rectangle node[midway, opacity = 1]{#1} (#4,#5)}{%
		\draw[fill = white!50!vi, fill opacity = 0.7] (#2,#3) rectangle node[midway, opacity = 1]{#1} (#4,#5)}}
\newcommand{\drawTallItem}[5][$ $]{%
	\ifthenelse{\boolean{BlackAndWhite}}{%
		\draw[fill = white!65!black, fill opacity = 0.7] (#2,#3) rectangle node[midway]{#1} (#4,#5)}{%
		\draw[fill = white!50!ti, fill opacity = 0.7] (#2,#3) rectangle node[midway, opacity = 1]{#1} (#4,#5)}}
\newcommand{\drawLargeItem}[5][$ $]{	\ifthenelse{\boolean{BlackAndWhite}}{%
		\draw[fill = white!55!black, fill opacity = 0.7] (#2,#3) rectangle node[midway]{#1} (#4,#5)}{%
		\draw[fill = white!50!li, fill opacity = 0.7] (#2,#3) rectangle node[midway, opacity = 1]{#1} (#4,#5)}}
\newcommand{\drawSmallItem}[5][$ $]{	\ifthenelse{\boolean{BlackAndWhite}}{%
		\draw[fill = white!95!black, fill opacity = 0.7] (#2,#3) rectangle node[midway]{#1} (#4,#5)}{%
		\draw[fill = white!50!si, fill opacity = 0.7] (#2,#3) rectangle node[midway, opacity = 1]{#1} (#4,#5)}}
\newcommand{\drawHorizontalItem}[5][$ $]{	\ifthenelse{\boolean{BlackAndWhite}}{%
		\draw[fill = white!75!black, fill opacity = 0.7] (#2,#3) rectangle node[midway]{#1} (#4,#5)}{%
		\draw[fill = white!50!hi, fill opacity = 0.7] (#2,#3) rectangle node[midway, opacity = 1]{#1} (#4,#5)}}

\newcommand{\drawBlueItem}[5][$ $]{%
	\ifthenelse{\boolean{BlackAndWhite}}{%
		\draw[fill = white!85!black, fill opacity = 0.7] (#2,#3) rectangle node[midway]{#1} (#4,#5);%
		\draw[pattern = north west lines] (#2,#3) rectangle node[midway]{#1} (#4,#5);
	}{%
		\draw[fill = white!50!ti, fill opacity = 0.7] (#2,#3) rectangle node[midway, opacity = 1]{#1} (#4,#5)}}

\newcommand\xqed[1]{%
	\leavevmode\unskip\penalty9999 \hbox{}\nobreak\hfill
	\quad\hbox{#1}}
\newcommand\demo{\xqed{$\lhd$}}

\newenvironment{proofClaim}{\noindent\textit{Proof.~}}{\demo \par%
	\addvspace{\baselineskip}
}

\newcommand{\eps}{\varepsilon}
\newcommand{\OPT}{\mathrm{OPT}}
\newcommand{\jobs}{\mathcal{J}}

\newcommand{\conf}{\mathcal{C}}

\newcommand{\sched}{\sigma}

\newcommand{\NN}{\mathbb{N}}
\newcommand{\QQ}{\mathbb{Q}}
\newcommand{\RR}{\mathbb{R}}
\newcommand{\Oh}{\mathcal{O}}

\newcommand{\NP}{\mathrm{NP}}
\newcommand{\Poly}{\mathrm{P}}

\newcommand{\pT}[1]{p(#1)}

\newcommand{\pTs}[1]{p_s(#1)}

\newcommand{\areaI}[1]{\mathrm{area}(#1)}

\newcommand{\epss}{\varepsilon'}

\newcommand{\wind}{\mathcal{W}}

\newcommand{\sset}[2]{\{#1\,|\,#2\}}
\newcommand{\LP}{\mathrm{LP}}

\newcommand{\pre}{\mathrm{split}}
\newcommand{\OPTpre}{\OPT_{\pre}}
\newcommand{\pmax}{p_{\mathrm{max}}}

\newcommand{\res}[1]{R(#1)}

\newcommand{\greedyMakespan}{\makespan}

\newcommand{\jR}[1]{r(#1)}

\newcommand{\makespan}{T}

\newcommand{\post}{\mathrm{post}}
\newcommand{\addTop}{
\geq
}
\newcommand{\intersectingJobs}{
>
}
\newcommand{\jobsT}[1]{\jobs_{#1,\addTop,\mathrm{pre}}} 
\newcommand{\jobsTInter}[1]{\jobs_{#1,\intersectingJobs,\mathrm{pre}}} 
\newcommand{\jobsLTNS}[1]{\jobs_{L,#1,\intersectingJobs,\mathrm{pre}}} 
\newcommand{\jobsLTS}[1]{\jobs_{L,#1,\intersectingJobs,\mathrm{post}}} 

\newcommand{\LLS}{\top} 
\newcommand{\addT}{\lceil T'/2 \rceil_{\eps\delta \greedyMakespan}}

\newcommand{\startPoints}{\mathcal{S}}
\newcommand{\REntry}[2]{R_{#1}^{(#2)}}
\newcommand{\mEntry}[2]{m_{#1}^{(#2)}}

\newcommand{\REntryS}[1]{\REntry{#1}{S}}
\newcommand{\mEntryS}[1]{\mEntry{#1}{S}}
\newcommand{\REntryL}[1]{\REntry{#1}{L}}
\newcommand{\mEntryL}[1]{\mEntry{#1}{L}}

\newcommand{\mRvectorLarge}{(\mEntryL{s},\REntryL{s})_{s \in \mathcal{S}}}

\def\DEBUG{true} 

\ifdefined\DEBUG
 
  \def\rem#1{{\marginpar{\raggedright\scriptsize #1}}}
  \newcommand{\malr}[1]{\rem{\textcolor{OliveGreen}{$\bullet$ #1}}}
\else
  
  \newcommand{\malr}[1]{}
\fi

\title{Closing the gap for single resource constraint scheduling} 

\titlerunning{Closing the gap for single resource constraint scheduling} 

\author{Klaus Jansen}{Kiel University, Germany}{kj@informatik.uni-kiel.de}{https://orcid.org/0000-0001-8358-6796}{}
\author{Malin Rau}{Hamburg University, Germany}{malin.rau@uni-hamburg.de}{https://orcid.org/0000-0002-5710-560X}{}

\authorrunning{K. Jansen and M. Rau} 

\Copyright{Klaus Jansen and Malin Rau} 



\category{} 

\relatedversion{} 

\funding{Research was supported by German Research Foundation (DFG) project JA 612 / 20-1}

\hideLIPIcs  

\EventEditors{}
\EventNoEds{0}
\EventLongTitle{}
\EventShortTitle{ARXIV}
\EventAcronym{ARXIV}
\EventYear{2021}
\EventDate{}
\EventLocation{}
\EventLogo{}
\SeriesVolume{0}
\ArticleNo{0}
\begin{document}

\maketitle

\begin{abstract}
In the problem \acl{resource}, we are given $m$ identical machines and a set of jobs, each needing one machine to be processed as well as a share of a limited renewable resource $R$.
A schedule of these jobs is feasible if, at each point in the schedule, the number of machines and resources required by jobs processed at this time is not exceeded.
It is $\NP$-hard to approximate this problem with a ratio better than $3/2$. 
On the other hand, the best algorithm so far has an absolute approximation ratio of $2+\eps$.
In this paper, we present an algorithm with absolute approximation ratio~$(3/2+\eps)$, which closes the gap between inapproximability and best algorithm with exception of a negligible small~$\eps$.
\end{abstract}

\begin{acronym}
	\acro{BP}[BP]{bin packing}
	\acro{parallelTSP}[PTS]{parallel task scheduling}
	\acro{moldablePTS}[MTS]{moldable task scheduling}
	\acro{StripPacking}[SP]{strip packing}
	\acro{SPRotations}[SPR]{strip packing with rotations}
	\acro{countiguousPTS}[CTS]{contiguous task scheduling}
	\acro{countiguousMPTS}[CMTS]{contiguous moldable task scheduling}
	\acro{RCS}[RCS]{resource constraint scheduling}
	\acro{resource}[SRCS]{single resource constraint scheduling}
	\acro{MSRCS}[MSRCS]{moldable single resource constraint scheduling}
	\acro{BPCardinality}[BPCC]{bin packing with cardinality constraint}
	\acro{SPCC}[SPCC]{strip packing with cardinality constraint}
	\acro{KKP}[kKP]{knapsack with cardinality constraint}
	\acro{UKKP}[UkKP]{unbounded knapsack with cardinality constraint}
	\acro{MMRS}[MMRS]{max-min-resource-sharing}
\end{acronym}

\newpage

\section{Introduction}
In the \acl{resource} problem, we are given $m$ identical machines, a discrete renewable resource with a fixed size $R \in \NN$ and a set of $n$ jobs $\jobs$. 
Each job has a processing time $\pT{j} \in \QQ$. 
We define the total processing time of a set of jobs $\jobs' \subseteq \jobs$ as $\pT{\jobs'} := \sum_{j \in \jobs'}\pT{j}$.
To be scheduled, each job $j \in \jobs$ needs one of the machines as well as a fix amount $\jR{j} \in \NN$ of the resource which it will allocate during the complete processing time $\pT{j}$ and which it deallocates as soon as it has finished its processing. 
Neither machine nor any part of the resource can be allocated by two different jobs at the same time. 
We define the area of a job as $\areaI{j} := \jR{j} \cdot \pT{j}$ and the area of a set of jobs $\jobs' \subseteq \jobs$ as  $\areaI{\jobs'} := \sum_{j \in \jobs'} \jR{j} \cdot \pT{j}$.

A schedule $\sigma: \jobs \rightarrow \NN$ maps each job $j \in \jobs$ to a starting point $\sigma(j) \in \QQ$. 
We say a schedule is feasible if
\begin{align*}
\forall t\in\QQ: & \sum_{j:t\in [\sigma(j),\sigma(j)+\pT{j})}\jR{j}\leq R \text{ and }&  \hfill \textit{(machine condition)}\\
\forall t\in\QQ: & \sum_{j:t\in [\sigma(j),\sigma(j)+\pT{j})}1\leq m \text{.} &  \hfill \textit{(resource condition)}
\end{align*}
Given a schedule $\sched$ where these two conditions hold, we can generate an assignment of resources and machines to the jobs such that each machine and each resource part is allocated by at most one job at a time, see \cite{Johannes06}. 
The objective is to find a feasible schedule, which minimizes the total length of the schedule called makespan, i.e., we have to minimize $\max_{j \in \jobs} \sigma(j) +p(j)$.

This problem arises naturally, e.g., in parallel computing, where jobs that are scheduled in parallel share common memory or in production logistics where different jobs need a different number of people working on it. 
From a theoretical perspective these problem is a sensible generalization of problems like scheduling on identical machines, parallel task scheduling and \acl{BPCardinality}. 

The algorithm with the best absolute ratio so far is a $(2+\eps)$-approximation by Niemeier and Wiese~\cite{NiemeierW15}.
In this paper, we close this gap between approximation and lower bound by presenting an algorithm with approximation ratio $(3/2+\eps)$.

\begin{theorem}
	\label{thm:SRCS:absoluteApprox}
	There is an algorithm for \acl{resource} with approximation ratio $(3/2 +\eps)$ and running time $\mathcal{O}(n\log(1/\eps)) + \Oh(n)(m\log(R)/\eps)^{\mathcal{O}_{\eps}(1)}$,
	where $\mathcal{O}_{\eps}$ dismisses all factors solely dependent on $1/\eps$.
\end{theorem}

As a by-product of this algorithm, we also present an algorithm, which has an approximation guarantee of $(1+\eps)\OPT + \pmax$. 
Note that we can scale each instance such that $\pmax = 1$ and hence we can see $\pmax$ as a constant independent of the instance. 
Algorithms with an approximation guarantee of the form $(1+\eps)\OPT + c$ fore some constant $c$ are called asymptotic polynomial time approximation schemes (APTAS).
Note that this algorithm is a $(2+\eps)$-approximation as well, but improves this ratio, in the case that $\pmax$ is strictly smaller than $\OPT$ and for $\pmax < \OPT/2$ improves the approximation ratio of the algorithm from Theorem \ref{thm:SRCS:absoluteApprox}.
\begin{theorem}
	\label{thm:SRCS:APTAS}
	There is an APTAS for \acl{resource} with an additive term $\pmax$ and running time 
	$\mathcal{O}(n\log(1/\eps+n)) + \Oh(n)(m\log(R)/\eps)^{\mathcal{O}_{\eps}(1)}$.
\end{theorem}

In the schedule generated by this APTAS, almost all jobs are completed before $(1+\Oh(\eps))\OPT$, except for a small set $\jobs'$ of jobs that all start simultaneously at $(1+\Oh(\eps))\OPT$, after the processing of all other jobs is finished.
The processing of this set $\jobs'$ causes the additive term $\pmax$.


\paragraph*{Methodology and Organization of this Paper}

In Section \ref{sec:SRCS:APTAS}, we will present the APTAS from Theorem \ref{thm:SRCS:APTAS}.
The general structure of the algorithm can be summarized as follows.  
First, we simplify the instance by rounding the processing time of the jobs and partitioning them into large, medium, and small corresponding to their processing time. 
Afterward, we use some linear programming approaches to find a placement of these jobs inside the optimal schedule.
The few jobs that are placed fractional with this linear program will be placed on top of the schedule contributing to the set $\jobs'$ which was mentioned before.

As usual for this kind of algorithms, we divide the jobs into large and small jobs. 
While for the placement of small jobs, we use techniques already known, see e.g. \cite{journals/talg/JansenMR19}, we used a new technique to place the large jobs.
For the following $(3/2+\eps)$ approximation, it is important to guarantee that only $\Oh_{\eps}(1)$, i.e. constant in $1/\eps$, jobs are not placed inside the optimal scheduling area.
To find such a schedule, we divide the schedule into $\Oh_{\eps}(1)$ time slots and guess the machine requirement and a rounded resource requirement of large jobs during this slot. 
Afterward, the large jobs are placed inside this profile using a linear program, which schedules only $\Oh_{\eps}(1)$ of these jobs fractionally.
We have to remove these $\Oh_{\eps}(1)$ fractionally scheduled jobs, as well as only one extra job per time slot due to the rounded guess, and assign them to the set $\jobs'$. 

Afterward, in Section \ref{sec:SRCS:absoluteRatio}, we present the $(3/2 +\eps)$-approximation and prove Theorem \ref{thm:SRCS:absoluteApprox}.
Instead of placing the fractional jobs on top of the schedule, we stretch the schedule by $(1/2 +\Oh(\eps))\OPT$, and place the fractional scheduled large jobs inside a gap in this stretched schedule.
This stretching allows us to define a common finishing point of (almost) all the jobs, which have a processing time larger than $\OPT/2$ and, thus, we avoid scheduling them fractionally with the linear program.
While the general idea of creating such a gap was used before, e.g., in \cite{Jansen12}, the novelty of this approach is to search for this gap at multiple points in time while considering two and not only one constraint.
This obstacle of considering two instead of one constraint while searching for a gap requires a more careful analysis of the shifted schedule, as was needed in other gap constructions.

\paragraph*{Related Work}
The problem \acl{RCS} is one of the classical problems in scheduling.
It was first studied in 1975 by Garey and Graham \cite{GareyG75}.
Given $m$ identical machines and capacities $R_1, \dots, R_s$ of $s$ distinct resources such that each job requires a share of each of them, they proved that the greedy list algorithm produces a schedule of length at most $(s +2 -(2s+1)/m)\OPT$. 
This corresponds to an approximation ratio of $(3 -3/m)$ for the case of $s = 1$ i.e., the problem studied in this paper.
In the same year Garey and Johnson \cite{garey1975complexity} showed that this general scheduling problem is NP-complete even if just one resource is given, i.e., $s = 1$. 
Lately, Niemeier and Wiese \cite{NiemeierW15} presented a $(2+\eps)$-approximation for \acl{resource}, and this is the best known ratio so far.

Note that the problem \acl{resource} contains multiple problems as subproblems.
When all the processing times are equal to one, this problem corresponds to \acl{BPCardinality}. 
Hence there is no algorithm with an approximation guarantee better than $3/2$ for this problem unless $\Poly = \NP$. 
For \acl{BPCardinality} Epstein and Levin \cite{EpsteinLevin} presented an AFPTAS.
This AFPTAS was improved and extended to work for \acl{resource} by Jansen et al. \cite{journals/talg/JansenMR19}. 
It has an additive term of $\mathcal{O}(\pmax\log(1/\eps)/\eps)$.

On the other hand, if the number of machines $m$ is larger than $n$, the constraint that only $m$ jobs can be processed at the same time is no longer a restriction.
The resulting problem is known as the parallel task scheduling problem. 
This problem is strongly NP-complete for $R\geq 4$ \cite{HenningJRS19} and there exists a pseudo polynomial algorithm for $R \leq 3$ \cite{DuL89a}.
Furthermore, for $R$ constant and $R \in n^{\Oh(1)}$ there exists polynomial time approximation schemes by Jansen and Porkolab \cite{JansenP02} as well as Jansen and Th\"ole \cite{JansenT10} respectively.
For an arbitrary large $R$ there exists no algorithm with approximation ratio smaller than $3/2$ unless $P = NP$.
The best algorithm for this scenario is a $(3/2+\eps)$ approximation by Jansen \cite{Jansen12}.

Finally, if each job requires at most $R/m$ from the resource, the used resources are no longer a restriction.
The corresponding problem is known as the NP-hard problem scheduling on identical machines.
For this problem several algorithms are known, see, e.g., \cite{Graham1969, Hochbaum,alon1998approximation, JansenKV16}.

An interesting extension of the considered problem is the consideration of resource dependent processing times. 
In this scenario, instead of a fixed resource requirement for the jobs, each jobs processing time depends on the amount of allocated resources.
The first result for this extension was achieved by Grigoriev et al. \cite{GUS07}.
They studied a variant where the processing time of a job depends on the machine it is processed on as well as the number of assigned resources and described a 3.75-approximation.
For the case of identical machines this result was improved by Kellerer \cite{Kel08} to a $(3.5+ \eps)$-approximation.
Finally, Jansen at al \cite{journals/talg/JansenMR19} presented an AFPTAS for this problem with additive term $\Oh(\pi_{\max}\log(1/\eps)/\eps)$, where $\pi_{\max}$ is the largest occurring processing time considering all possible resource assignments.

Closely related to \acl{resource} is the strip packing problem. 
Here we are given a set of rectangular items that have to be placed overlapping free into a strip with bounded width and infinite height. 
We can interpret \acl{resource} as a Strip Packing problem by setting the jobs processing time to the height of a rectangular item and the resource requirement to its width. The difference to strip packing is now, that when placing an item, we are allowed to slice it vertically as long as the lower border of all slices are placed at the same \emph{vertical level}. Furthermore, we have an \acl{resource} carnality condition that allows only $m$ items to intersect each horizontal line through the strip.
The strip packing problem has been widely studied~\cite{BakerBK81, BakerCR80, BougeretDJRT2011, CoffmanGJT80, Golan, HarrenS09, JansenS09, Kenyon00, Schiermeyer94, Sleator80, Steinberg97, Sviridenko12}.
The algorithm with approximation ratio $(5/3+\eps)$, which is the smallest so far, was presented by Harren, Jansen, Prädel, and van Stee~\cite{HarrenJPS14}.
On the other hand, using a reduction from the partition problem, we know that there is no polynomial-time algorithm with an approximation ratio smaller than $3/2$. Closing this gap between the best approximation ratio and lower bound represents an open question.
Strip packing has also been studied with respect to asymptotic approximation ratio~\cite{BakerBK81, CoffmanGJT80, Golan,Kenyon00}.  
The best algorithms in this respect are an $AFPTAS$ with additive term $\Oh(1/\eps\log(1/\eps))h_{\max}$~\cite{Sviridenko12,BougeretDJRT2011} and an APTAS with additive term $h_{\max}$~\cite{JansenS09}, where $h_{\max}$ is the tallest height in the set of given items.
Finally, this problem has been studied with respect to pseudo-polynomial processing time~\cite{TurekWY92,JansenT10, NadiradzeW16, GalvezGIK16, tcs/JansenR19,AdamaszekKPP17}, where the width of the strip is allowed to occur polynomial in the running time of the algorithm. 
There is no pseudo-polynomial algorithm with an approximation ratio smaller than $5/4$~\cite{HenningJRS19} and a ratio of $(5/4+\eps)$ is achieved by the algorithm in~\cite{JansenR19}.

\section{APTAS with additive term \texorpdfstring{$p_{\max}$}{}}
\label{sec:SRCS:APTAS}
In this section, we present an asymptotic PTAS for the single resource constraint scheduling problem which has an approximation guarantee of $(1+\eps)\OPT + p_{\max}$ and a running time of $\mathcal{O}(n\log(n)) + (m\log(R))^{\mathcal{O}_{\eps}(1)}$, i.e., we prove Theorem \ref{thm:SRCS:APTAS} in this section. 
This algorithm represents the base for the $(3/2 +\eps)$-approximation, see Section \ref{sec:SRCS:absoluteRatio}.
For simplicity of notation, we assume that $1/\eps \in \NN$ in the following sections.
If it is not, the algorithm will use $1/\lceil 1/\eps\rceil$ as value for $\eps$.
Remember that we can assume $n > m$ since otherwise the machine constraint can never be violated and thus the problem corresponds to \acl{parallelTSP} (for which such an algorithm is already known). 

The general structure of the algorithm can be summarized as follows.
First, we find a lower bound $T$ on the optimal makespan of the schedule and use it, to simplify the instance.
The simplification contains the follwoing steps. First partition the jobs into large, medium, and small jobs dependent on their processing times compared to $T$. 
Depending on this partition the processing times and resource requirement of the jobs are rounded, see Section \ref{sec:SRCS:APTAS:Simplify}.
Afterward, we use a binary search framework and a dual approximation approach:
Given a target makespan $T'$, we guess the structure of the optimal solution and try to place the jobs according to this guess.
The placement of the three groups of jobs (i.e., large, medium, and small) happens independently and is described in Sections \ref{sec:SRCS:APTAS:largeJobs} and \ref{sec:SRCS:APTAS:smallJobs} respectively.
A more detailed summary of the algorithm can be found in Section \ref{sec:SRCS:APTAS:Summary}.

\subsection{Simplifying the input instance}
\label{sec:SRCS:APTAS:Simplify}
In the first step of the algorithm, we simplify the given instance such that it has a simple structure and a reduced set of processing times. 
Consider the lower bound on the optimal makespan $T:= \min\{\pmax, \areaI{\jobs}/R, \pT{\jobs}/m\}$. 
By the analysis of the greedy list schedule, as described in \cite{NiemeierW15}, we know that the optimal schedule has a size of at most $\frac{1}{m}\pT{\jobs} + \frac{2}{R}\areaI{\jobs} + \pmax \leq  4 T$.


In the next step, we will create a gap between jobs with a large processing time and jobs with a small processing time, by removing a set of medium sized jobs.
We want to schedule this set of medium sized jobs in the beginning or end of the schedule using the greedy algorithm. 
However this schedule of the medium sized items should add at most $\Oh(\eps) \OPT$ to the makespan.
The schedule generated by the greedy list schedule algorithm has a makespan of at most $\frac{1}{m}\pT{\jobs} + \frac{2}{R}\areaI{\jobs} + \pmax \leq 4 \OPT$.
Hence, we choose the set of medium jobs $\jobs_M$ such that $\pmax(\jobs_M) \leq \eps \OPT$, i.e. the maximal processing time appearing in the set of medium jobs is bounded by $\eps \OPT$.
On the other hand, the total area and the total processing time of the medium jobs should be small enough.
\begin{lemma}
	\label{lma:gapMediumJobsResource}
	Consider the sequence $\gamma_0 = \eps$, $\gamma_{i+1}= \gamma_i\eps^{4}$. 
	There exists an $i \in \{1, \dots, 1/\eps\}$ such that 
	\begin{align}
	\frac{1}{m}\pT{\jobs_{\gamma_i}} + \frac{2}{R} \areaI{\jobs_{\gamma_i}} \leq \eps\left(\frac{1}{m} \pT{\jobs} + \frac{2}{R} \areaI{\jobs}\right)\mathrm{,}
	\label{eq:mediumJobs}
	\end{align}
	where $\jobs_{\gamma_i} := \sset{j \in \jobs}{\pT{j} \in [\gamma_iT,\gamma_{i-1}T)}$ and we can find this $i$ in $\Oh(n +1/\eps)$.
\end{lemma}

\begin{proof}
	Note that the $1/\eps$ sets $\jobs_{\gamma_i}$ are disjoint subsets of $\jobs$.
	Hence, by the pigeon principle there has to be one set, which fulfills condition (\ref{eq:mediumJobs}).
	Obviously by sorting the jobs, we can find the sets $\jobs_{\gamma_i}$ in $\Oh(nlog(n))$ for each $i \in \{1, \dots, 1/\eps\}$.
\end{proof}

Let $i \in \{1, \dots, 1/\eps\}$ be the smallest value such that $\jobs_{\gamma_i}$ has the property from Lemma \ref{lma:gapMediumJobsResource} and define $\mu := \gamma_i$ and $\delta := \gamma_{i-1}$. 
Note that $\gamma_i = \eps^{1+4i}$ and hence $\delta \geq \eps^{4/\eps+1}$.
Using these values for $\delta$ and $\mu$, we partition set of jobs into large $\jobs_L := \{j \in \jobs | \pT{j} \geq \delta \greedyMakespan\}$, small $\jobs_S := \{j \in \jobs| \pT{j} < \mu \greedyMakespan\}$ and medium $\jobs_M := \{j \in \jobs | \mu \greedyMakespan \leq \pT{j} < \delta \greedyMakespan\}$. 

\begin{lemma}
	\label{lma:mediumJobs}
The medium jobs can be scheduled in $\Oh(n \log (n))$ operations with makespan $\Oh(\eps)T$
\end{lemma}
\begin{proof}
	We use the greedy algorithm described in \cite{NiemeierW15} to schedule the medium jobs in $\Oh(n \log (n))$ operations.
	The resulting schedule has a makespan of at most $\frac{1}{m}\pT{\jobs_M} + \frac{2}{R}\areaI{\jobs_M} + \pmax \leq 4\eps T$, by the choice of the medium jobs and Lemma \ref{lma:gapMediumJobsResource}.
\end{proof}

The final simplification step is to round the processing times of the large jobs using the rounding in Lemma \ref{lma:rounding} to multiples of $\eps\delta T$. 
\begin{lemma}[See \cite{tcs/JansenR19}]
	\label{lma:rounding}
	Let be $\delta \geq \eps^k$ for some value $k \in \mathbb{N}$. 
	At loss of a factor $(1+2\eps)$ in the approximation ratio, we can round the processing time of each job $j$ with processing time $\eps^{l-1} \greedyMakespan\geq p(j) \geq \eps^l \greedyMakespan$ for some $l\in \mathbb{N}\leq k$ such that it has a processing time $k_j \eps^{l+1} \greedyMakespan$ for a value $k_j \in \{1/\eps,\dots 1/\eps^2 -1\}$. 
	Furthermore, the jobs can be started at a multiple of $\eps^{l+1}\greedyMakespan$.
\end{lemma}
In this step, we reduce the number of different processing times of large jobs to $\Oh(\log_{\eps}(\delta)/\eps^2) = \Oh(1/\eps^3)$.
However, we lengthen the schedule at most by the factor $(1+2\eps)$.
Furthermore, this rounding reduces the starting times of the jobs to at most $\Oh(1/(\eps\delta))$ possibilities since all the large jobs start and end at multiples of $\eps \delta T$ and the optimal makespan of the rounded instance is bounded by $(1+2\eps)\cdot 4T$.

Let $I_{rounded}$ be the rounded instance and $\OPT_{rounded}$ its optimal makespan.
In the following, we assume that the considered optimal schedule that is rounded has a makespan $\OPT_{rounded}$ of the form $l \cdot \eps T$, while making an extra error of at most $\eps T$. 
As a result, the makespan of this optimal schedule is bounded by $T \leq \OPT_{rounded} \leq (1+2\eps)\cdot 4T + \eps T$ and hence $l \in [1/\eps, 13/\eps] \cap \NN$, for $\eps \leq 1$. 

In the following, we will prove that given a value $T' := l \eps T$ for some $l \in [1/\eps, 16/\eps] \cap \NN$, we either can find a schedule with makespan $(1+ \Oh(\eps))T' + \pmax$, or prove that there is no schedule that schedules the large and small jobs with makespan at most $T'$.
In the algorithm, we will find the smallest value for $l$ such that we find a schedule of makespan at most $(1+ \Oh(\eps))l \eps T+\pmax$.

\subsection{Scheduling Large Jobs}
\label{sec:SRCS:APTAS:largeJobs}
In this section, we describe how to schedule the large jobs when given the size of the makespan $T' := l \eps T$ for some $l \in [1/\eps, 16/\eps] \cap \NN$.
Remember, that in the rounded schedule these jobs only start at the multiples of $\eps\delta T$.
Let $\startPoints$ be the set of all these points in time up to $T'$ and let $\mathcal{P}$ be the set of all rounded processing times for large jobs. 
The processing time between two consecutive starting times $s_i,s_{i+1} \in \startPoints$ is called layer $l_i$. 
Notice that, during the processing of a layer in a rounded optimal schedule, the resource requirement and number of machines used by large jobs stays unchanged since the large jobs only start and end at the starting points in $\startPoints$. 

\begin{lemma}
	\label{lma:schedulingLargeJobs}
	Let $\gamma \in (0,1]$ and $T' = l \eps T \geq \OPT$ for some $l \in \NN$, and $\bar{\jobs}$ be a set of jobs for which an optimal schedule exists such that all jobs in $\bar{\jobs}$ have their starting and endpoints in $\startPoints$.
	There exists an algorithm that finds in $\Oh((m\log(R))^{\Oh_{\eps}(1)|\startPoints|/\gamma})$ operations $\Oh((m\log(R))^{\Oh_{\eps}(1)|\startPoints|/\gamma})$ schedules with the following properties
	\begin{enumerate}
		\item In each of the schedules, all large jobs are scheduled except for a set $\jobs' \subseteq \bar{\jobs}$ of at most $|\jobs'| \in 3|\startPoints|$ jobs and a total resource requirement of at most $\res{\jobs'} \leq \gamma R$.
		\item In at least one of the schedules, in each layer given by $\startPoints$, the total number of machines and resources not used by jobs in $\bar{\jobs}$ is as large as in a rounded optimal schedule.
	\end{enumerate}
\end{lemma}
In the following, we describe the algorithm to find these schedules.
If  $m \leq 3|\startPoints| /\gamma$, there can be at most $m/\delta \in \Oh(\eps\gamma\delta^2)$ large jobs and we can enumerate all possible schedules of these large jobs in $(1/\eps\delta)^{1/(\eps\gamma\delta^2)} \leq \Oh((m\log(R)/\delta)^{\Oh_{\eps}(1)/\gamma})$. 
Hence consider the case that $m \geq 3|\startPoints| /\gamma$.

In the first step, we partition the set of large jobs into wide and narrow jobs. 
Let $\alpha \in O(\gamma\eps\delta)$ for a small enough constant.
The required constraints on $\alpha$ can be found at the end of this section.
Wide jobs $j \in \jobs_{LW}$ have a resource requirement of $\jR{j} \geq \alpha R$ and narrow jobs $j \in \jobs_{LN}$ have a resource requirement of $\jR{j} < \alpha R$. 
There can be at most $\Oh(1/\alpha\delta)$ wide large jobs and at most $\Oh(m/\delta)$ narrow large jobs since $\areaI{\jobs}/R \leq (1+\eps) T$ and $\pT{\jobs}/m \leq (1+\eps) T$. 
Hence there are at most  $\mathcal{O}((1/\eps\delta)^{1/\alpha\delta})$ possibilities to place the wide large jobs.

To schedule the narrow large jobs with a given Makespan $T'$, we guess for each possible starting point $s \in \startPoints$ the required number of machines $\mEntryL{s}$ and the resource requirement $\REntryL{s}$ of the large jobs that are processed between this starting point and the next in a rounded optimal solution. 
Since we have at most $1/\eps\delta$ starting points, there are at most $(mR)^{1/\eps\delta}$ possible guesses. 

However, since $R$ can be exponential in the input size, this number is too large.
Instead of guessing the exact resource requirement, we guess the interval $(R/(1 +1/m)^{t+1}, R/(1+1/m)^t]$ in which this resource requirement lies, i.e., if the resource requirement lies in the interval $(R/(1 +1/m)^{t+1}, R/(1+1/m)^t]$, we define $\REntryL{s} := R/(1+1/m)^t$.
There are at most $\lceil\log_{(1+1/m)}(R)\rceil \in \mathcal{O}(\log(R)m)$ such intervals intersecting $[1,R]$ and, therefore, at most $\mathcal{O}((\log(R)m)^{\Oh(|\startPoints|)})$ possible guesses for the resource requirement.

For a given guess $\mRvectorLarge$, we solve the following linear program called $\LP_{large}$.
\begin{align}
\sum_{j \in \jobs_{LN}} \sum_{s'=s-\pT{j}+\eps\delta\greedyMakespan}^s \jR{j}x_{j,s'} &\leq \REntryL{s} &\forall s \in \startPoints\\
\sum_{j \in \jobs_{LN}} \sum_{s'=s-\pT{j}+\eps\delta\greedyMakespan}^s x_{j,s'} &\leq \mEntryL{s} &\forall s \in \startPoints\\
\sum_{s = 0}^{\makespan' - p(j)} x_{j,s} &=1 & j \in \jobs_{LN}\\
x_{j,s} & \geq 0 & \forall  s \in \startPoints, j \in \jobs_{LN}
\end{align}
The variable $x_{j,s}$ represents which fraction of job $j$ that starts at time $s$.
The first two conditions ensure that in each layer, the positioned jobs do not exceed the guessed number of resources or machines. 
The third condition ensures that each job is scheduled. 
If there is an optimal schedule with resource requirement and machine number such that the values $\mRvectorLarge$ are an upper bound, we can transform it to a solution of this linear program by setting $x_{j,s} = 1$ and $x_{j,s'} = 0$ for $s = \sched{(j)}$ and $s' \in \startPoints \setminus\{s\}$.
hence, if the guess was correct, the linear program has a solution. 
This linear program has $2|\startPoints| + |\jobs_{LN}|$ conditions and $|\startPoints| |\jobs_{LN}|$ variables. 
Hence, a basic solution has at most $2|\startPoints| + |\jobs_{LN}|$ non zero components. 
Since the factors on the right hand side are bounded by $R$, we can find a basic solution to this linear program in $(|\startPoints| |\jobs_{LN}|\log(R))^{\Oh(1)} = (m\log(R)/\delta)^{\Oh(1)}$ using the Ellipsoid-Method or any other polynomial time algorithms for linear programs.

In the end of the algorithm, after the binary search part, we transform the solution $x$ to $\LP_{large}$ to an integral solution, where no job is scheduled fractionally.
Further, since we rounded the resource requirement of the narrow large jobs per layer, we have to remove some of the jobs from their starting points.
However, we have to be careful to remove as few as possible since we are only allowed to add one $\pmax$ to the makespan of the schedule and we have a machine constraint of $m$.  

\begin{lemma}
	\label{lma:SRCS:RemoveLargeJobs}
	There exists an $\alpha \in  O(\gamma\eps\delta)$ such that
	given a solution $x$ to $\LP_{large}$ that has only $2|\startPoints| + |\jobs_{NL}|$ non-zero components, we have to remove at most $\Oh(1/(\eps\delta))$ jobs, to fulfill the second property of Lemma \ref{lma:schedulingLargeJobs}.
	The removed jobs have a total resource requirement of at most $\gamma R$ and
	can be found in $\Oh(|\startPoints| \cdot |\jobs_{NL}|)$.
\end{lemma}

\begin{proof}
	There are two reasons why we have to remove jobs from the solution $x$.
	The first is the removal of fractionally scheduled jobs and the second is to fulfill the second property of Lemma \ref{lma:schedulingLargeJobs}.  
	
	\begin{claim}
		A basic solution for the above linear program has at most $2|\startPoints|$ fractional scheduled jobs and their total resource requirement is bounded by $2\alpha R/\eps\delta$.
	\end{claim}
	
	\begin{proofClaim}
		A job $j \in \jobs_{LN}$ is scheduled fractional if there are at least two points $s,s' \in S$ such that $x_{j,s}>0$ and $x_{j,s'}>0$.
		Each job needs at least one non zero components to be scheduled and hence $|\jobs_{LN}|$ non zero components are used by different jobs. 
		Therefore, each fractional scheduled job needs one of the $2|\startPoints|$ residual non zero components, i.e., there are at most $2|\startPoints|$ fractional scheduled jobs.
		Hence, we have to remove at most $2|\startPoints|$ jobs, to schedule all jobs integral. 
		Since each job has a resource requirement of at most $\alpha R$ the total resource requirement of removed jobs is bounded by $2\alpha R/\eps\delta$.
	\end{proofClaim}
	
	Not that in the linear program, the values $\REntryL{s}$ are upper bounds on the resource requirements of the large jobs. 
	Hence, the current schedule may not fulfill the second property of Lemma \ref{lma:schedulingLargeJobs}.
	To fulfill this property we have to remove jobs such that each layer uses at most $\REntryL{s}/(1+1/m)$ resources, which is a lower bound on the number resources used by narrow large jobs in the optimal schedule if the guess was correct.
	\begin{claim}
		We have to remove at most one job per layer to use at most $\REntryL{s}/(1+1/m)$ resources for each layer $s \in S$.
	\end{claim}
	\begin{proofClaim}
		Let be $s \in S$ and $\REntryL{s} = R/(1+1/m)^{t}$, i.e., we had guessed that the resource requirement of the narrow large jobs scheduled in this layer is contained in the interval $(R/(1+1/m)^{t+1},R/(1+1/m)^{t}]$. 
		
		If the total resource requirement of all narrow large jobs scheduled in layer $s$ is at most $R/(1+1/m)^{t+1}$, we do not need to remove any job and hence the claim is trivially true for this layer.	
		Otherwise, the widest job has a resource requirement of at least $(R/(1+1/m)^{t+1})/m$ since there are at most $m$ jobs scheduled in this layer. 
		When we remove the widest job, the total resource requirement of the residual jobs is bounded by 
		\begin{align*}
		&R/(1+1/m)^{t} - R/m(1+1/m)^{t+1} \\
		= &m(1+1/m)R/m(1+1/m)^{t+1} - R/m(1+1/m)^{t+1} \\
		= &R/(1+1/m)^{t+1} \mathrm{.}
		\end{align*}
	\end{proofClaim}
	
	Therefore, the total number of narrow jobs to be removed is bounded by $3|S| \in \Oh(1/(\eps\delta))$ and their total resource requirement is bounded by $3|S|\alpha R \leq \gamma R$ for $\alpha \in \Oh(\gamma\eps\delta)$ small enough.

\end{proof}

\subsection{Scheduling Small Jobs}
\label{sec:SRCS:APTAS:smallJobs}
In this section, we describe how to schedule the small jobs.
We will schedule these jobs inside the layers using the residual free resources and machines given by the guess for the large jobs.
We define $\mEntryS{s}$ as the number of machines in layer $s$ not used by jobs with processing times larger than $\delta T$ and analogously define $\REntry{s}{S}$ as the number of resources not used by jobs with processing times larger than $\delta T$ during the processing of this layer in an optimal schedule.

\begin{restatable}{lemma}{schedulingSmallJobs}
	\label{lma:schedulingSmallJobs}
	Define for each layer $s \in \startPoints$ a box with processing time $(1+\eps)\eps\delta T$, $\mEntryS{s}$ machines and $\REntry{s}{S}$ resources, where
	the values $\mEntryS{s}$ and $\REntry{s}{S}$ are at least as large as in a rounded optimal schedule.
	There exists an algorithm with time complexity $\Oh(n) \cdot \Oh_{\eps,|\startPoints|}(1)$, that schedules the jobs inside the boxes and an additional horizontal box with $m$ machines, $R$ resources, and processing time $\Oh(\eps)T$.
\end{restatable}

This algorithm uses the same techniques as the AFPTAS designed by Jansen et al.~\cite{journals/talg/JansenMR19}.
For the sake of completeness the ideas can be found in Section \ref{sec:smallJobs}.


\subsection{Summary of the Algorithm}
\label{sec:SRCS:APTAS:Summary}
In this section, we summarize the steps of the APTAS. Given an instance $I = (\jobs, m, R)$ of \acl{resource} and an $\eps' \in (0,1)$ the algorithm performs the following steps:
\begin{stepList}
	\item[Initialization] First, we define an $\eps \in (0,1)$ such that $1/\eps \in \NN$ and $\eps \leq \eps'/c$ for a $c \in \RR$ large enough such that the following $(1+\Oh(\eps))$ approximation is an $(1+\eps')$ approximation.
	Further, we define $T:= \min\{\pmax, \areaI{\jobs}/R, \pT{\jobs}/m\}$ and know that $T \leq \OPT \leq 4 T$.
	\item[Simplification] 
	In the next step, we simplify the given instance $I$, by partitioning and rounding the jobs. 
	First, we find the values $\delta$ and $\mu$ in $\Oh(n \log(n))$ operations, as described in Lemma \ref{lma:gapMediumJobsResource} and partition the set of jobs $\jobs$ into large, medium and small jobs in $\Oh(n)$.
	The large jobs are rounded to at most $\Oh(1/\eps^4)$ different processing times using Lemma \ref{lma:rounding}. 
	We call this rounded instance $I_{r,1}$ and it holds that $\OPT(I) \leq \OPT(I_{r,1}) \leq (1+2\eps)\OPT(I)$.
	In the last simplification step, we  partition the set of large jobs into narrow and wide ones.
	\item[Binary Search Framework] For the rounded instance $I_{r,1}$, we will find the value $T' = (1+i\eps)T$, for $i \in [0,4/\eps +8]\cap \NN$ such that $T' -\eps T \leq \OPT(I_{r,1}) \leq T'$ and give a schedule with makespan at most $(1+\Oh(\eps))T' + \pmax = (1+\Oh(\eps))\OPT(I)+ \pmax$. 
	We try the possible values for $T'$ in binary search fashion in Step 4.
	\item[Scheduling Step\label{test}] Given a value $T' := i \eps T$ we determine the corresponding set $\startPoints$. 
	Afterward, we use the algorithm in Lemma \ref{lma:schedulingLargeJobs} with $\gamma = 1$ to find the at most $\Oh((m\log(R))^{\Oh_{\eps}(1)}$ possible schedules for all large jobs except for the set $\jobs'$.
	For each of the possible schedules, we determine the values $\mEntryS{s}$ and $\REntry{s}{S}$ and try to place the small jobs in the corresponding boxes using the algorithm from Lemma \ref{lma:schedulingSmallJobs}.
	The schedules form the boxes can be inserted in the layers of the schedule stretched by the factor $(1+\eps)$, while the extra box is placed on to, adding at most $\Oh(\eps)T$ to the makespan. 
	Note that by Lemma \ref{lma:schedulingLargeJobs} for one of the schedules of large jobs the values $\mEntryS{s}$ and $\REntry{s}{S}$ are as large as in a rounded optimal schedule, or the value $T'$ was to small.
	If we find such a solution, we place the set $\jobs'$ on top, save the corresponding schedule, and try the next smaller value for $T'$ in binary search fashion.
	Otherwise, the algorithm tires the next larger value for $T'$.
	\item[Final step] 
	After the binary search, we take the smallest computed schedule and add the medium jobs on top, adding at most $\Oh(\eps)T$ to the makespan by lemma \ref{lma:mediumJobs}.
	Finally, we replace the rounded large jobs, with the original large jobs.
\end{stepList}

The time complexity of the above algorithm can be bounded by $\Oh(n\log(n)) + \Oh(\log(1/\eps) \cdot \Oh((m\log(R))^{\Oh_{\eps}(1)}) \cdot \Oh(n) \cdot \Oh_{\eps,|\startPoints|}(1)\leq \Oh(n\log(n)) + \Oh(n) \cdot \Oh((m\log(R)/\eps)^{\Oh_{\eps}(1)})$. 
Hence, we have proven the Theorem \ref{thm:SRCS:APTAS}.

\section{A \texorpdfstring{$(3/2 +\eps)$}{}-Approximation}
\label{sec:SRCS:absoluteRatio} 
In this section, we design an algorithm with approximation ratio $(3/2 +\eps)$.
The algorithm uses the techniques described for the APTAS. 
However, we have to be more careful which large jobs can be scheduled fractional and shifted to the end of the schedule and which cannot.

We aim to find a schedule with makespan $(3/2+\mathcal{O}(\eps))T'$, where $T'$ is the assumed optimal makespan given by a binary search framework.
If we discard a job larger than $ T'/2+ \mathcal{O}(\eps)\greedyMakespan$ and schedule it after $T'$, we exceed the aspired approximation ratio of $(3/2+\mathcal{O}(\eps))T$.
We call this set of critical jobs huge jobs, i.e., $\jobs_H := \{j \in \jobs | p(j) > T'/2\}$.
As a consequence, we redefine the set of large jobs as $\jobs_L := \{j \in \jobs | \delta\greedyMakespan \leq p(j) \leq  T'/2 \}$ respectively.

Notice that the processing of all huge jobs has to intersect the time $T'/2$ in each schedule with makespan at most $T'$.
Therefore, each machine can contain at most one of these jobs.
If we could guess the starting positions of these huge jobs, the discarded jobs in the linear program $\LP_{large}$ would have a processing time of at most $\OPT / 2$ and could be placed on top of the schedule. 
Sadly this guessing step is not possible in polynomial time since there are up to $m$ of these jobs and iterating all combinations of their starting position needs $\Omega((1/\eps\delta)^m)$ operations.
Our idea is to let almost all the huge jobs end at a common point in time, e.g. $T'$, and thus avoid the guessing step.
To solve the violation of the resource or machine condition,
we shift up all the jobs which start after $\lceil T'/2\rceil_{\eps\delta\greedyMakespan}$ by $\lceil T'/2\rceil_{\eps\delta\greedyMakespan}$ such that they now start after $T'$, where we denote by $\lceil T'/2\rceil_{\eps\delta\greedyMakespan}$ the integer multiple of $\eps\delta\greedyMakespan$ that is the first which has a size of at least $T'/2$.

While this shift fixes the start positions of the huge jobs, the large jobs are again placed with the techniques described in Section \ref{sec:SRCS:APTAS:largeJobs}.
Since Lemma \ref{lma:schedulingLargeJobs} states that each of the generated schedules may not schedule a subset $\jobs'$ of the large jobs, we need to find a gap in the shifted schedule where we can place them. 
In the following, we will consider optimal schedules and the possibilities to rearrange the jobs. 
Depending on this arrangement, we can find a gap for the fractional scheduled large jobs of processing time $\addT$.

Let us assume that we have to schedule $k := |\jobs'|\in \mathcal{O}_{\eps}(1) \leq m/4$ jobs with total resource requirement at most $\gamma R \leq R/(3|\startPoints|)$. 
We consider an optimal schedule, after applying the simplification steps, i.e., we consider the rounded instance $I_{r,1}$ and the corresponding transformed optimal schedule, where each large job starts at a multiple of $\eps\delta\greedyMakespan$ and each huge job starts at a multiple of $\eps^2\greedyMakespan$. 
Furthermore, we will assume that there are more than $4k = \mathcal{O}_{\eps}(1)$ huge jobs. 
Otherwise, we can guess their starting positions in $\mathcal{O}((1/\eps\delta)^{4k})$ and place the fractional scheduled large jobs on top of the schedule. 

In the following, we will prove that by extending it by $\lceil T'/2\rceil_{\eps\delta\greedyMakespan}$, we can transform the rounded optimal schedule $\OPT_{\mathrm{rounded}}$ such that all the huge jobs, except for $\Oh(k)$ of them, end at a common point in time and we can place $k$ further narrow large jobs without violating the machine or the resource constraint.

\begin{lemma}
	\label{lma:hugeJobs}
	Let $k := |\jobs'| \leq m/4$ and $\gamma \leq 1/(3|\startPoints|)$.
	Furthermore, let a rounded optimal schedule $\OPT_{\mathrm{rounded}}$ with makespan at most $T'$ and at most $|\startPoints|$ starting positions for large jobs be given.
	
	Without removing any job from the schedule, we can find a transformed schedule $\OPT_{\mathrm{shift}}$ with makespan at most $\OPT_{\mathrm{rounded}} + \addT$, with the following properties: 
	\begin{enumerate}
	\item We can guess the end positions of all huge jobs in polynomial time. 
	\item There is a gap of processing time $\addT$ with $k$ empty machines and $\gamma R$ free resources where we can schedule the jobs in $\jobs'$.
	\item There is an injection which maps each layer $s$ in $\OPT_{\mathrm{rounded}}$ with $m_{s,S}$ machines and $R_{s,S}$ resources not used by huge and large jobs to a layer in $\OPT_{\mathrm{shift}}$ where there are at least as many machines and resources not used by these jobs. 
	\end{enumerate}
\end{lemma}

\begin{proof}
	We will prove this lemma by a careful analysis of the structure of the schedule $\OPT_{\mathrm{rounded}}$. 
	First, however, we introduce some notations.
	Let $s \in \startPoints$, with $s > T'/2$ be any starting point of large jobs.
	We say a job $j \in \jobs$ intersects $s$ or is intersected by $s$ if $\sigma(j) < s < \sigma(j) + \pT{j}$.
	We will differentiate sets of jobs that start before $T'/2$ and those that start at or after $T'/2$ by adding the attribute $_{\mathrm{pre}}$ to sets of jobs that contain only jobs starting before $T'/2$, and the attribute $_{\mathrm{\post}}$ to those that contain only jobs that start at or after $T'/2$.
	Furthermore, we will identify the sets of jobs that intersect certain points of time.
	We will add the attribute $_{s,\addTop}$ to denote a set of jobs that is processed at least until the point in time $s \in \startPoints$, i.e., we denote by $\jobsT{s} := \{j \in \jobs| \pT{j} \geq \delta\greedyMakespan, \sigma(j) < T'/2, \sigma(j) +\pT{j} \geq s\}$ the set of large and huge jobs starting before $T'/2$ and ending at or after $s$.
	On the other hand, if we are only interested in the jobs that intersect the time $s$, we add the attribute $_{s,\intersectingJobs}$ to the set and mean $\jobs_{s,\intersectingJobs,\mathrm{pre}} := \{j \in \jobs| \pT{j} \geq \delta\greedyMakespan, \sigma(j) < T'/2, \sigma(j) +\pT{j} > s\}$.
	Finally, we will indicate if the set contains only huge or only large jobs, by adding the attribute $_H$ or $_L$.

	\begin{figure}[ht]
		\centering
		\resizebox{\textwidth}{!}{
		\begin{tikzpicture}
		
		\pgfmathsetmacro{\h}{2}
		\pgfmathsetmacro{\w}{4}
		\pgfmathsetmacro{\t}{\h + 0.3*\h}
		\pgfmathsetmacro{\lm}{0.3 *\w}
		\pgfmathsetmacro{\rm}{0.7*\w}
		\pgfmathsetmacro{\rrm}{0.6*\w}
		
		\pgfmathsetmacro{\rt}{\t - 0.2 *\t +0.2* \h}
		\pgfmathsetmacro{\rtt}{\t - 0.5 *\t +0.5* \h}
		\pgfmathsetmacro{\rttt}{\t - 0.7 *\t +0.7* \h}
		
		\drawVerticalItem{0}{0}{\w}{2*\h};
		
		\draw[dashed] (-0.5,\h) node[left] {$T'/2$} -- (\w +0.5,\h);
		\draw[dashed] (-0.5,2*\h) node[left] {$T'$} -- (\w +0.5,2*\h);
		
		\draw[dotted] (-0.5,\t) node[left] {$\tau$} -- (\w+0.5,\t) node[right] {};
		
		\draw (\rrm,\h)  --
		(\rrm,\t- 0.7*\t + 0.7*\h) -- 
		(\rm,\t- 0.7*\t + 0.7*\h)-- 
		(\rm,\t+ 0.7*\t - 0.7*\h)-- (
		0.8*\w,\t+ 0.7*\t - 0.7*\h) -- 
		(0.8*\w,\t + 1.6 *\h - 0.8*\t)-- 
		(\w,\t + 1.6 *\h - 0.8*\t);
		
		\drawTallItem{0.0}{0.1*\h}{0.1 *\lm}{2*\h};
		\drawTallItem{0.1*\lm}{0.2*\h-0.1*\h}{0.2*\lm}{2*\h-0.1*\h};
		\drawTallItem{0.2*\lm}{0.3*\h-0.2*\h}{0.4*\lm}{2*\h-0.2*\h};
		\drawTallItem{0.4*\lm}{0.5* \h-0.3*\h}{0.6*\lm}{2*\h-0.3*\h};
		\drawTallItem{0.6*\lm}{0.6* \h-0.4*\h}{0.8*\lm}{2*\h-0.4*\h};
		\drawTallItem{0.8*\lm}{0.8*\h-0.5*\h}{0.9*\lm}{2*\h-0.5*\h};
		\drawTallItem{0.9*\lm}{0.85*\h-0.6*\h}{\lm}{2*\h-0.6*\h};

		\drawTallItem{\lm}{\rt}{\lm+0.1*\rrm -0.1*\lm}{0.2};
		\drawTallItem{\lm+0.1*\rrm -0.1*\lm}{\rtt}{\lm+0.15*\rrm -0.15*\lm}{0.3};
		\drawTallItem{\lm+0.15*\rrm -0.15*\lm}{\rttt}{\lm+0.3*\rrm -0.3*\lm}{0};
		
		\begin{scope}[yshift= -\h cm]
		\drawBlueItem{\lm}{\rt + \h}{\lm+0.1*\rrm -0.1*\lm}{\t+\h+0.6*2*\h -0.6*\t};		
		\drawBlueItem{\lm +0.1*\rrm -0.1*\lm}{\rtt +\h}{\lm +0.4*\rrm -0.4*\lm}{\t+\h+0.4*2*\h -0.4*\t};		
		\drawBlueItem{\lm +0.4*\rrm -0.4*\lm}{\t+\h - 0.3*\t + 0.3*\h}{\lm +0.9*\rrm -0.9*\lm}{\t+\h+0.7*2*\h -0.7*\t};		
		\drawBlueItem{\lm +0.9*\rrm -0.9*\lm}{2*\h}{\rrm}{\t+\h+0.65*2*\h -0.65*\t};		
		\drawBlueItem{\rrm}{\t+\h- 0.7*\t + 0.7*\h}{\rm}{\t+\h+0.5*2*\h -0.5*\t};
		\end{scope}

		\draw[fill = gray, opacity=0.4](\rrm,\h)  --(\rrm,\t- 0.7*\t + 0.7*\h) -- (\rm,\t- 0.7*\t + 0.7*\h)-- (\rm,\t+ 0.7*\t - 0.7*\h)-- (0.8*\w,\t+ 0.7*\t - 0.7*\h) -- (0.8*\w,\t + 1.6 *\h - 0.8*\t)-- (\w,\t + 1.6 *\h - 0.8*\t) -- (\w,\h);
		
		\draw (3.4,\t) ellipse (17pt and 5pt);
		\draw (3.4,\t) -- (\w+1,\t+0.5) node[right] {$\jobsLTNS{\tau}$};
		
		\draw (0.6,\t) ellipse (17pt and 5pt);
		\draw (0.6,\t) -- (0.5*\w,2*\h +1) node[right] {$\jobsT{\tau}$};
		\draw (3.4,\t) -- (0.5*\w,2*\h +1);
		
		\begin{scope}[xshift = 1.8*\w cm]
		\pgfmathsetmacro{\h}{2}
		\pgfmathsetmacro{\w}{4}
		\pgfmathsetmacro{\t}{\h + 0.3*\h}
		\pgfmathsetmacro{\lm}{0.3 *\w}
		\pgfmathsetmacro{\rm}{0.7*\w}
		\pgfmathsetmacro{\rrm}{0.6*\w}
		
		\pgfmathsetmacro{\rt}{\t - 0.2 *\t +0.2* \h}
		\pgfmathsetmacro{\rtt}{\t - 0.5 *\t +0.5* \h}
		\pgfmathsetmacro{\rttt}{\t - 0.7 *\t +0.7* \h}
		
		\def\midd[#1]{\lm +{#1}*\rrm -{#1}*\lm}
		
		\draw   (0,0) rectangle (\w,3*\h);
		
		\draw[dashed] (-0.5,\h) -- (\w +0.5,\h) node[right] {$T'/2$};
		\draw[dashed] (-0.5,2*\h) node[left] {$T'$} -- (\w +0.5,2*\h);
		\draw[dashed] (-0.5,3*\h) -- (\w +0.5,3*\h) node[right] {$T' +\addT$};
		
		\draw[dashed] (-0.5,\t) -- (\w+0.5,\t) node[right] {$\tau$};
		\draw[dashed] (-0.5,\t+\h) -- (\w+0.5,\t+\h) node[right] {$\tau + \addT$};
		
		\draw (\rrm,\h)  --(\rrm,\t- 0.7*\t + 0.7*\h) -- (\rm,\t- 0.7*\t + 0.7*\h)-- (\rm,\t+ 0.7*\t - 0.7*\h)-- (0.8*\w,\t+ 0.7*\t - 0.7*\h) -- (0.8*\w,\t + 1.6 *\h - 0.8*\t)-- (\w,\t + 1.6 *\h - 0.8*\t);

		\foreach \x/\xx/\y in {0.0/0.1/0.1,0.1/0.2/0.2,0.2/0.35/0.3,0.35/0.4/0.45,0.4/0.6/0.5,0.6/0.65/0.6,
			0.65/0.75/0.75,0.75/0.8/0.8,0.8/0.9/0.85,0.9/1.0/0.9}
		{
			\drawTallItem{\xx*\lm}{\y*\h}{\x *\lm}{2*\h};	
		}

		\drawTallItem{\lm}{\rt}{\lm+0.1*\rm -0.1*\lm}{0.2};
		\drawTallItem{\lm+0.1*\rm -0.1*\lm}{\rtt}{\lm+0.15*\rm -0.15*\lm}{0.3};
		\drawTallItem{\lm+0.15*\rm -0.15*\lm}{\rttt}{\lm+0.3*\rm -0.3*\lm}{0};
		
		\drawVerticalItem{\lm}{\rt + \h}{\lm+0.1*\rrm -0.1*\lm}{\t+\h+0.6*2*\h -0.6*\t};
		\drawVerticalItem{\lm +0.1*\rrm -0.1*\lm}{\rtt +\h}{\lm +0.4*\rrm -0.4*\lm}{\t+\h+0.4*2*\h -0.4*\t};
		\drawVerticalItem{\lm +0.4*\rrm -0.4*\lm}{\t+\h - 0.3*\t + 0.3*\h}{\lm +0.9*\rrm -0.9*\lm}{\t+\h+0.7*2*\h -0.7*\t};
		\drawVerticalItem{\lm +0.9*\rrm -0.9*\lm}{2*\h}{\rrm}{\t+\h+0.65*2*\h -0.65*\t};
		\drawVerticalItem{\rrm}{\t+\h- 0.7*\t + 0.7*\h}{\rm}{\t+\h+0.5*2*\h -0.5*\t};
		
		\draw[fill = white!85!black, fill opacity = 0.7](\lm+0.3*\rm -0.3*\lm,0)-- (\lm+0.3*\rm -0.3*\lm,\h) --(\rrm,\h)  --(\rrm,\t- 0.7*\t + 0.7*\h) -- (\rm,\t- 0.7*\t + 0.7*\h)-- (\rm,\t+ 0.7*\t - 0.7*\h)-- (0.8*\w,\t+ 0.7*\t - 0.7*\h) -- (0.8*\w,\t + 1.6 *\h - 0.8*\t)-- (\w,\t + 1.6 *\h - 0.8*\t) -- (\w,0);
		
		\draw[fill = white!85!black, fill opacity = 0.7] (0.0,0.0) -- (0.0,0.1*\h) -- (0.1*\lm,0.1*\h)--(0.1*\lm,0.2*\h) -- (0.2*\lm,0.2*\h)--(0.2*\lm,0.3*\h)--(0.35*\lm,0.3* \h)--(0.35*\lm,0.45* \h)--(0.4*\lm,0.45* \h)--(0.4*\lm,0.5* \h) -- (0.6*\lm,0.5* \h)-- (0.6*\lm,0.6* \h) --(0.65*\lm,0.6*\h)--(0.65*\lm,0.75*\h)--(0.75*\lm,0.75*\h)--(0.75*\lm,0.8*\h) --(0.8*\lm,0.8*\h)--(0.8*\lm,0.85*\h)--(0.9*\lm,0.85* \h)--(0.9*\lm,0.9* \h)--(\lm,0.9* \h) -- (\lm,0.2) --(\lm+0.1*\rm -0.1*\lm,0.2) --(\lm+0.1*\rm -0.1*\lm,0.3)--(\lm+0.15*\rm -0.15*\lm,0.3) --(\lm+0.15*\rm -0.15*\lm,0.0);

		\draw[fill = white!85!black, fill opacity = 0.7](0.0,3*\h)--(0.0,\t +\h)--(\lm,\t+\h)--(\lm,\t+\h+0.6*2*\h -0.6*\t)--(\lm+0.1*\rrm -0.1*\lm,\t+\h+0.6*2*\h -0.6*\t) -- (\lm+0.1*\rrm -0.1*\lm,\t+\h+0.4*2*\h -0.4*\t)--(\lm +0.4*\rrm -0.4*\lm,\t+\h+0.4*2*\h -0.4*\t)--(\lm +0.4*\rrm -0.4*\lm,\t+\h+0.7*2*\h -0.7*\t)--(\lm +0.9*\rrm -0.9*\lm,\t+\h+0.7*2*\h -0.7*\t)--(\lm +0.9*\rrm -0.9*\lm,\t+\h+0.65*2*\h -0.65*\t)--(\rrm,\t+\h+0.65*2*\h -0.65*\t)--(\rrm,\t+\h+0.5*2*\h -0.5*\t)--(\rm,\t+\h+0.5*2*\h -0.5*\t)--(\rm,\t- 0.7*\t + 0.7*\h +\h)-- (\rm,\t+ 0.7*\t - 0.7*\h+\h)-- (0.8*\w,\t+ 0.7*\t - 0.7*\h+\h) -- (0.8*\w,\t + 1.6 *\h - 0.8*\t+\h)-- (\w,\t + 1.6 *\h - 0.8*\t+\h)--(\w,\t+\h)--(\w,3*\h);
		
		\draw [thick,decorate,decoration={brace,amplitude=5pt}] (\w,\t+\h) -- (\w,\t) node [midway,right,xshift=+12pt](B){};
		\node[right,text width=0.7*\w cm] at (B) { $k$ machines\\ unused};.
		
		\draw[fill = white!85!black, fill opacity = 0.7] (\lm+0.15*\rm -0.15*\lm,\rttt) -- (\lm+0.3*\rm -0.3*\lm,\rttt) -- (\lm+0.3*\rm -0.3*\lm,\h)--(\lm +0.9*\rrm -0.9*\lm,\h)--(\lm +0.9*\rrm -0.9*\lm,\t - 0.3*\t + 0.3*\h) --(\lm +0.4*\rrm -0.4*\lm,\t - 0.3*\t + 0.3*\h)--(\lm +0.4*\rrm -0.4*\lm,\rtt)--(\lm+0.15*\rm -0.15*\lm,\rtt);
		\end{scope}
		\end{tikzpicture}  
	}
		\caption{\textbf{On the left}: a rounded optimal schedule. The hatched rectangles are the jobs that start after $T'/2$ and intersect $\tau$, 
			the dark gray area corresponds to large jobs, which start before $T'/2$ and end after $\tau$ and the dark gray rectangles on the left are huge jobs.
			\textbf{On the right}: the corresponding shifted schedule
		}
		\label{fig:SRCS:IntroductionOfJobSets}
	\end{figure}
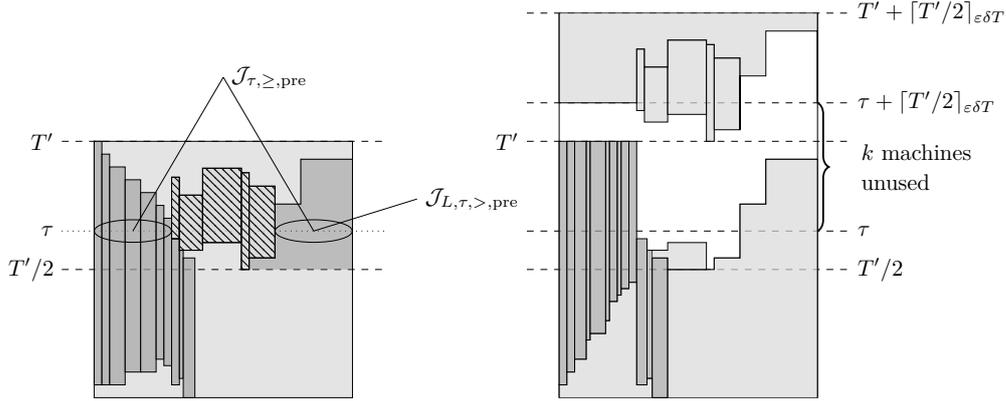

	Let $\tau \in \{s|s \in S,T'/2 \leq s \leq T'\}$ be the smallest value such that there are at most $m-k$ jobs (huge or large) that start before $T'/2$ and intersect $\tau$, i.e., end after $\tau$. 
	We partition the set of large jobs intersected by $\tau$ into two sets.
	Let $\jobsLTNS{\tau} := \{j \in \jobs_L| \sigma(j) < T'/2, \sigma(j) +\pT{j} > \tau\}$ be the set of large jobs which start before $T'/2$ and end at or after $\tau$. 
	Further let $\jobsLTS{\tau} := \{j \in \jobs_L| T'/2 \leq \sigma(j) < \tau, \sigma(j) + \pT{j}> \tau\}$ be the set of large jobs, which are started at or after $T/2$ but before $\tau$ and end after $\tau$, see Figure \ref{fig:SRCS:IntroductionOfJobSets}.

	Note that by the choice of $\tau$ in each point between $\tau$ and $T'/2$ in the schedule there are more than $m -k$ machines used by $\jobsT{\tau}$. 
	As a result there are at most $k-1$ machines used by jobs starting after $T'/2$ at each point between $T'/2$ and $\tau$, implying $|\jobsLTS{\tau}| < k$.
		
	We now construct a shifted schedule. 
	Starting times in this schedule will be denoted by $\sigma'$. 
	We shift each job $j\in \jobs$ with $\sigma(j) \geq T'/2$ and $\sigma(j) + p_j \geq \tau$ exactly $\addT$ upwards, i.e., we define $\sigma'(j) := \sigma(j) + \addT$ for these jobs. 
	Furthermore, each huge job $j \in \jobs_H$ intersecting $\tau$ is shifted upwards such that it ends at $T'$, i.e., we define $\sigma'(j) := T' - p_j$ for these jobs $j$, see Figure \ref{fig:SRCS:IntroductionOfJobSets}. 
	Note that there are at most $k$ huge jobs ending strictly before $\tau$.
	If the total number of huge jobs ending before or at $\tau$ is larger than $k$, we choose arbitrarily from the set of jobs ending at $\tau$ and shift them until there are exactly $k$ huge jobs ending before or at $\tau$.
	
	\begin{claim}
		After this shift there are at least $k$ machines at each point between $\tau$ and $\tau + \addT$ that are not used by any other job.
	\end{claim}
	\begin{proofClaim}
	Up to $T'$, there are $k$ free machines, because there is no new job starting between $\tau$ and $T'$ since we shifted all of them up such that they start after $\addT$. 
	On the other hand, only jobs from the set $\jobsLTS{\tau}$ are processed between $T'$ and $\tau + \addT$. 
	Since $|\jobsLTS{\tau}| < k$ and $m-k \geq k$ this leaves $k$ free machines which proves the claim. 
	\end{proofClaim}
	
	The idea is to place the gap at $\tau$ since there are enough free machines. 
	However, it can happen that at a point between $\tau$ and $\tau + \addT$ there is not enough free resource for the gap. 
	In the following, we carefully analyze where we can place the $k$ jobs, dependent on the structure of the optimal schedule $\OPT_{\mathrm{rounded}}$

	\begin{caseList}
		\item[$\jR{\jobsT{\tau}} \leq R-\gamma R$] 
		In this case there are at least $\gamma R$ free resources at each point in the shifted schedule between $\tau$ and $T'$ since there are no jobs starting between these points of time. 
		
		To place the $k$ fractional scheduled jobs, we have to generate a gap of processing time $\lceil T'/2 \rceil_{\eps\delta T}$.
		In this gap there have to be $k$ unused machines and $\gamma R$ unused resources.
		For the time between $\tau$ and $T'$, we have this guarantee, while for the time between $T'$ and $\tau + \lceil T'/2 \rceil_{\eps\delta \greedyMakespan}$, we have $k$ free machines, but might have less than $\gamma R$ free resource.
		The only jobs overlapping in this time window are the jobs from the set $\jobsLTS{\tau}$, see Figure \ref{fig:SRCS:IntroductionOfJobSets}.
		If these jobs have a small resource requirement, we have found our gap, see Case 1.1. and, otherwise, we have to look more careful at the schedule.

		\begin{caseList}
			\item[$\jR{\jobsLTS{\tau}} \leq R- \gamma R$.]	
			In this case, the required gap is positioned between $\tau$ and $\tau + \addT$, see Figure \ref{fig:SRCS:ShiftCase11}. 
			In this shifted optimal schedule there are at most $k$ huge jobs ending before $\tau$. 
			
			In the algorithm, we will guess $\tau$ dependent on a given solution for the large jobs and guess these $k$ huge jobs and their start points in $\mathcal{O}(m^kS^k)$, which is polynomial in the input size, see Section \ref{sec:SRCS:Absolute:Summary} for an overview.
			
			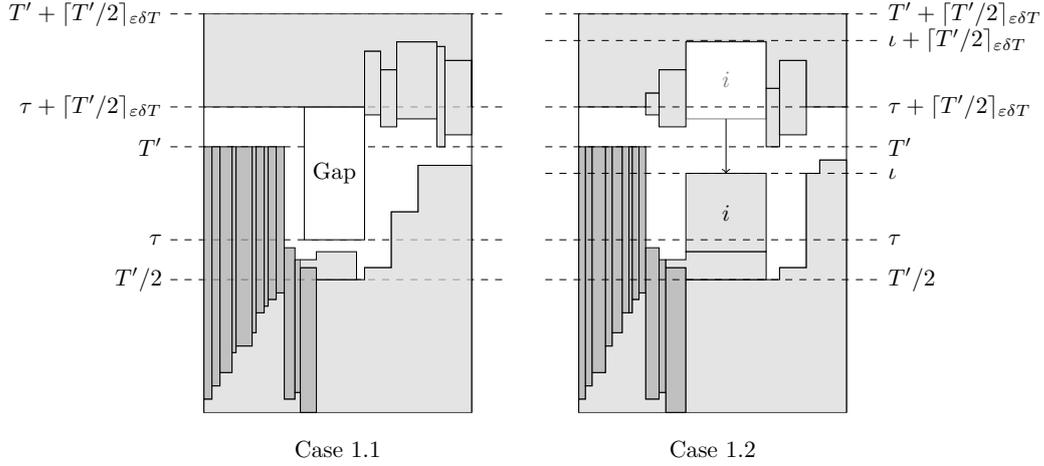
\begin{figure}[ht]
				\centering
				\resizebox{\textwidth}{!}{
				\begin{tikzpicture}
				
				\pgfmathsetmacro{\h}{2}
				\pgfmathsetmacro{\w}{4}
				\pgfmathsetmacro{\t}{\h + 0.3*\h}
				\pgfmathsetmacro{\lm}{0.3 *\w}
				\pgfmathsetmacro{\rm}{0.7*\w}
				\pgfmathsetmacro{\rrm}{0.6*\w}
				
				\pgfmathsetmacro{\rt}{\t - 0.2 *\t +0.2* \h}
				\pgfmathsetmacro{\rtt}{\t - 0.5 *\t +0.5* \h}
				\pgfmathsetmacro{\rttt}{\t - 0.7 *\t +0.7* \h}
				
				\draw   (0,0) rectangle (\w,3*\h);
				\node[below] at (\w/2,-2ex) {Case 1.1};
				\draw[dashed] (-0.5,\h) node[left] {$T'/2$} -- (\w +0.5,\h);
				\draw[dashed] (-0.5,2*\h) node[left] {$T'$} -- (\w +0.5,2*\h);
				\draw[dashed] (-0.5,3*\h) node[left] {$T' + \addT$} -- (\w +0.5,3*\h);
				
				\draw[dashed] (-0.5,\t) node[left] {$\tau$} -- (\w+0.5,\t);
				\draw[dashed] (-0.5,\t+\h) node[left] {$\tau + \addT$} -- (\w+0.5,\t+\h);
				
				\draw (\rrm,\h)  --(\rrm,\t- 0.7*\t + 0.7*\h) -- (\rm,\t- 0.7*\t + 0.7*\h)-- (\rm,\t+ 0.7*\t - 0.7*\h)-- (0.8*\w,\t+ 0.7*\t - 0.7*\h) -- (0.8*\w,\t + 1.6 *\h - 0.8*\t)-- (\w,\t + 1.6 *\h - 0.8*\t);
				
				
				\foreach \x/\xx/\y in {
					0.0/0.1/0.1,
					0.1/0.2/0.2,
					0.2/0.35/0.3,
					0.35/0.4/0.45,
					0.4/0.6/0.5,
					0.6/0.65/0.6,
					0.65/0.75/0.75,
					0.75/0.8/0.8,
					0.8/0.9/0.85,
					0.9/1.0/0.9
				}
				{
					\drawTallItem{\xx*\lm}{\y*\h}{\x *\lm}{2*\h};	
				}
				
				\drawTallItem{\lm}{\rt}{\lm+0.1*\rm -0.1*\lm}{0.2};
				\drawTallItem{\lm+0.1*\rm -0.1*\lm}{\rtt}{\lm+0.15*\rm -0.15*\lm}{0.3};
				\drawTallItem{\lm+0.15*\rm -0.15*\lm}{\rttt}{\lm+0.3*\rm -0.3*\lm}{0};
				
				\begin{scope}[xshift= -\rm cm +\w cm]
				\drawVerticalItem{\lm}{\rt + \h}{\lm+0.2*\rrm -0.2*\lm}{\t+\h+0.6*2*\h -0.6*\t};
				\drawVerticalItem{\lm +0.2*\rrm -0.2*\lm}{\rtt +\h}{\lm +0.4*\rrm -0.4*\lm}{\t+\h+0.4*2*\h -0.4*\t};
				\drawVerticalItem{\lm +0.4*\rrm -0.4*\lm}{\t+\h - 0.3*\t + 0.3*\h}{\lm +0.9*\rrm -0.9*\lm}{\t+\h+0.7*2*\h -0.7*\t};
				\drawVerticalItem{\lm +0.9*\rrm -0.9*\lm}{2*\h}{\rrm}{\t+\h+0.65*2*\h -0.65*\t};
				\drawVerticalItem{\rrm}{\t+\h- 0.7*\t + 0.7*\h}{\rm}{\t+\h+0.5*2*\h -0.5*\t};
				\end{scope}
				
				\draw[fill = white] (\lm +0.3,\t) rectangle node[midway]{Gap} (\rm - 0.4,\t+\h);
				
				\draw[fill = white!85!black, fill opacity = 0.7](\lm+0.3*\rm -0.3*\lm,0)-- (\lm+0.3*\rm -0.3*\lm,\h) --(\rrm,\h)  --(\rrm,\t- 0.7*\t + 0.7*\h) -- (\rm,\t- 0.7*\t + 0.7*\h)-- (\rm,\t+ 0.7*\t - 0.7*\h)-- (0.8*\w,\t+ 0.7*\t - 0.7*\h) -- (0.8*\w,\t + 1.6 *\h - 0.8*\t)-- (\w,\t + 1.6 *\h - 0.8*\t) -- (\w,0);
				
				\draw[fill = white!85!black, fill opacity = 0.7] (0.0,0.0) -- (0.0,0.1*\h) -- (0.1*\lm,0.1*\h)--(0.1*\lm,0.2*\h) -- (0.2*\lm,0.2*\h)--(0.2*\lm,0.3*\h)--(0.35*\lm,0.3* \h)--(0.35*\lm,0.45* \h)--(0.4*\lm,0.45* \h)--(0.4*\lm,0.5* \h) -- (0.6*\lm,0.5* \h)-- (0.6*\lm,0.6* \h) --(0.65*\lm,0.6*\h)--(0.65*\lm,0.75*\h)--(0.75*\lm,0.75*\h)--(0.75*\lm,0.8*\h) --(0.8*\lm,0.8*\h)--(0.8*\lm,0.85*\h)--(0.9*\lm,0.85* \h)--(0.9*\lm,0.9* \h)--(\lm,0.9* \h) -- (\lm,0.2) --(\lm+0.1*\rm -0.1*\lm,0.2) --(\lm+0.1*\rm -0.1*\lm,0.3)--(\lm+0.15*\rm -0.15*\lm,0.3) --(\lm+0.15*\rm -0.15*\lm,0.0);
				
				\draw[fill = white!85!black, fill opacity = 0.7](0.0,3*\h)--(0.0,\t +\h)--(\lm-\rm+\w,\t+\h)--(\lm-\rm+\w,\t+\h+0.6*2*\h -0.6*\t)--(\lm+0.2*\rrm -0.2*\lm-\rm+\w,\t+\h+0.6*2*\h -0.6*\t) -- (\lm+0.2*\rrm -0.2*\lm-\rm+\w,\t+\h+0.4*2*\h -0.4*\t)--(\lm +0.4*\rrm -0.4*\lm-\rm+\w,\t+\h+0.4*2*\h -0.4*\t)--(\lm +0.4*\rrm -0.4*\lm-\rm+\w,\t+\h+0.7*2*\h -0.7*\t)--(\lm +0.9*\rrm -0.9*\lm-\rm+\w,\t+\h+0.7*2*\h -0.7*\t)--(\lm +0.9*\rrm -0.9*\lm-\rm+\w,\t+\h+0.65*2*\h -0.65*\t)--(\rrm-\rm+\w,\t+\h+0.65*2*\h -0.65*\t)--(\rrm-\rm+\w,\t+\h+0.5*2*\h -0.5*\t)--(\rm-\rm+\w,\t+\h+0.5*2*\h -0.5*\t)--(\rm-\rm+\w,\t+\h)--(\w,\t+\h)--(\w,3*\h);
				
				\draw[fill = white!85!black, fill opacity = 0.7] (\lm+0.15*\rm -0.15*\lm,\rttt) -- (\lm+0.3*\rm -0.3*\lm,\rttt) -- (\lm+0.3*\rm -0.3*\lm,\h)--(\lm +0.9*\rrm -0.9*\lm,\h)--(\lm +0.9*\rrm -0.9*\lm,\t - 0.3*\t + 0.3*\h) --(\lm +0.4*\rrm -0.4*\lm,\t - 0.3*\t + 0.3*\h)--(\lm +0.4*\rrm -0.4*\lm,\rtt)--(\lm+0.15*\rm -0.15*\lm,\rtt);
				
				\begin{scope}[xshift = 1.4*\w cm]
				\pgfmathsetmacro{\h}{2}
				\pgfmathsetmacro{\w}{4}
				\pgfmathsetmacro{\t}{\h + 0.3*\h}
				\pgfmathsetmacro{\lm}{0.25 *\w}
				\pgfmathsetmacro{\rm}{0.85*\w}
				\pgfmathsetmacro{\rrm}{0.75*\w}
				\pgfmathsetmacro{\i}{\h + 0.8*\h}
				
				\pgfmathsetmacro{\rt}{\t - 0.2 *\t +0.2* \h}
				\pgfmathsetmacro{\rtt}{\t - 0.5 *\t +0.5* \h}
				\pgfmathsetmacro{\rttt}{\t - 0.7 *\t +0.7* \h}

				\draw   (0,0) rectangle (\w,3*\h);
				\node[below] at (\w/2,-2ex) {Case 1.2};

				\draw (\rrm,\h)  --(\rrm,\t- 0.7*\t + 0.7*\h) -- (\rm,\t- 0.7*\t + 0.7*\h)-- (\rm,\i)-- (0.9*\w,\i) -- (0.9*\w,\i + 1.0 *\h - 0.5*\i)-- (\w,\i + 1.0 *\h - 0.5*\i);
				
				\foreach \x/\xx/\y in {
					0.0/0.1/0.1,
					0.1/0.2/0.2,
					0.2/0.4/0.3,
					0.4/0.5/0.5,
					0.5/0.65/0.6,
					0.65/0.75/0.75,
					0.75/0.8/0.75,
					0.8/0.9/0.85,
					0.9/1.0/0.9
				}
				{
					\drawTallItem{\xx*\lm}{\y*\h}{\x *\lm}{2*\h};	
				}
				\drawTallItem{\lm}{\rt}{\lm+0.1*\rrm -0.1*\lm}{0.2};
				\drawTallItem{\lm+0.1*\rrm -0.1*\lm}{\rtt}{\lm+0.15*\rrm -0.15*\lm}{0.3};
				\drawTallItem{\lm+0.15*\rrm -0.15*\lm}{\rttt}{\lm+0.3*\rrm -0.3*\lm}{0};

				
				\begin{scope}[gray]
				\draw (\lm +0.3*\rrm -0.3*\lm,\t+\h - 0.3*\t + 0.3*\h) rectangle node[midway]{$i$} (\lm +0.9*\rrm -0.9*\lm,\i +\h);
				\end{scope}
				
				\drawVerticalItem{\lm}{\rt + \h}{\lm+0.1*\rrm -0.1*\lm}{\t+\h+0.15*2*\h -0.15*\t};
				\drawVerticalItem{\lm +0.1*\rrm -0.1*\lm}{\rtt +\h}{\lm +0.3*\rrm -0.3*\lm}{\t+\h+0.4*2*\h -0.4*\t};
				\drawVerticalItem{\lm +0.9*\rrm -0.9*\lm}{2*\h}{\rrm}{\t+\h+0.2*2*\h -0.2*\t};
				\drawVerticalItem{\rrm}{\t+\h- 0.7*\t + 0.7*\h}{\rm}{\t+\h+0.5*2*\h -0.5*\t};

				\begin{scope}[yshift = -\h cm]
				
				\drawVerticalItem[$i$]{\lm +0.3*\rrm -0.3*\lm}{\t+\h - 0.3*\t + 0.3*\h}{\lm +0.9*\rrm -0.9*\lm}{\i +\h};
				\end{scope}
				
				\draw[->] (\lm +0.6*\rrm -0.6*\lm,\t+\h - 0.3*\t + 0.3*\h) -- (\lm +0.6*\rrm -0.6*\lm,\i);

				\draw[fill = white!85!black, fill opacity = 0.7](\lm+0.3*\rrm -0.3*\lm,0)-- (\lm+0.3*\rrm -0.3*\lm,\h) --(\rrm,\h)  --(\rrm,\t- 0.7*\t + 0.7*\h) -- (\rm,\t- 0.7*\t + 0.7*\h)-- (\rm,\i)-- (0.9*\w,\i) -- (0.9*\w,\i + 1.0 *\h - 0.5*\i)-- (\w,\i + 1.0 *\h - 0.5*\i) -- (\w,0);
				
				\draw[fill = white!85!black, fill opacity = 0.7] 
				(0.0,0.0) -- 
				(0.0,0.1*\h) -- 
				(0.1*\lm,0.1*\h) --
				(0.1*\lm,0.2*\h) -- 
				(0.2*\lm,0.2*\h) -- 
				(0.2*\lm,0.3*\h) -- 
				(0.4*\lm,0.3* \h) -- 
				(0.4*\lm,0.5* \h) -- 
				(0.5*\lm,0.5* \h) -- 
				(0.5*\lm,0.6* \h) --
				(0.65*\lm,0.6*\h) -- 
				(0.65*\lm,0.75*\h) -- 
				(0.8*\lm,0.75*\h) -- 
				(0.8*\lm,0.85*\h) -- 
				(0.9*\lm,0.85* \h) -- 
				(0.9*\lm,0.9* \h) -- 
				(\lm,0.9* \h) -- 
				(\lm,0.2) --
				(\lm+0.1*\rrm -0.1*\lm,0.2) --
				(\lm+0.1*\rrm -0.1*\lm,0.3) -- 
				(\lm+0.15*\rrm -0.15*\lm,0.3) --
				(\lm+0.15*\rrm -0.15*\lm,0.0);

				\draw[fill = white!85!black, fill opacity = 0.7](0.0,3*\h)--(0.0,\t +\h)--(\lm,\t+\h)--(\lm,\t+\h+0.15*2*\h -0.15*\t)--(\lm+0.1*\rrm -0.1*\lm,\t+\h+0.15*2*\h -0.15*\t) -- (\lm+0.1*\rrm -0.1*\lm,\t+\h+0.4*2*\h -0.4*\t)--(\lm +0.3*\rrm -0.3*\lm,\t+\h+0.4*2*\h -0.4*\t)--(\lm +0.3*\rrm -0.3*\lm,\t+\h+0.7*2*\h -0.7*\t)--(\lm +0.9*\rrm -0.9*\lm,\t+\h+0.7*2*\h -0.7*\t)--(\lm +0.9*\rrm -0.9*\lm,\t+\h+0.2*2*\h -0.2*\t)--(\rrm,\t+\h+0.2*2*\h -0.2*\t)--(\rrm,\t+\h+0.5*2*\h -0.5*\t)--(\rm,\t+\h+0.5*2*\h -0.5*\t)--(\rm,\t+\h)--(\w,\t+\h)--(\w,3*\h);
				
				\draw[fill = white!85!black, fill opacity = 0.7]
				(\lm+0.15*\rrm -0.15*\lm,\rttt) -- 
				(\lm +0.3*\rrm -0.3*\lm,\rttt) -- 
				(\lm +0.3*\rrm -0.3*\lm,\h)--
				(\lm +0.9*\rrm -0.9*\lm,\h)--
				(\lm +0.9*\rrm -0.9*\lm,\t - 0.3*\t + 0.3*\h) --
				(\lm +0.3*\rrm -0.3*\lm,\t - 0.3*\t + 0.3*\h)--
				(\lm +0.3*\rrm -0.3*\lm,\rtt)--
				(\lm+0.15*\rrm -0.15*\lm,\rtt) -- 
				(\lm+0.15*\rrm -0.15*\lm,\rttt);
				
				\draw[dashed] (-0.5,\h) -- (\w +0.5,\h) node[right] {$T'/2$};
				\draw[dashed] (-0.5,2*\h) -- (\w +0.5,2*\h) node[right] {$T'$};
				\draw[dashed] (-0.5,3*\h) -- (\w +0.5,3*\h) node[right] {$T' + \addT$};
				
				\draw[dashed] (-0.5,\t) -- (\w+0.5,\t)  node[right] {$\tau$};
				\draw[dashed] (-0.5,\t+\h) -- (\w+0.5,\t+\h) node[right] {$\tau + \addT$};
				\draw[dashed](-0.5,\i +\h) -- (\w+0.5,\i+\h) node[right] {$\iota + \addT$};
				\draw[dashed](-0.5,\i) -- (\w+0.5,\i) node[right] {$\iota$};
				\end{scope}
				\end{tikzpicture} 
			}
				\caption{Examples for the two Cases 1.1. and 1.2. In Case 1.1 the gap is positioned between $\tau$ and $\tau + \addT$.
				In Case 1.2 the jobs in the set $\jobsLTS{\tau}^{\iota}$ are shifted back down. At each point between $\iota$ and $\iota + \addT$ there are at least $\gamma R$ unused resources.}
				\label{fig:SRCS:ShiftCase11}
			\end{figure}
			
			\item[$\jR{\jobsLTS{\tau}} > R- \gamma R$.]
			In this case, there is a point $t \in [\tau, T']$ such that after this point there are less than $\gamma R$ free resources. 
			Therefore, we need another position to place the fractionally scheduled jobs. 
			We partition the set $\jobsLTS{\tau}$ into at most $|\startPoints|/2$ sets $\jobsLTS{\tau}^{\iota}$ by the their original finishing points $\iota \in \startPoints_{> \tau}$, i.e.,  each  job in $\jobsLTS{\tau}^{\iota}$ finishes at $\iota$ in the non shifted rounded optimal schedule. 
			
			\begin{claim}
				One of the sets  $\jobsLTS{\tau}^{\iota}$,  $\iota \in \startPoints_{> \tau}$, has a resource requirement of at least $\gamma R$.
			\end{claim}
			\begin{proofClaim}
				The jobs in $\jobsLTS{\tau}$ use more than $R- \gamma R$ resource in total. 
				Since $\gamma \leq 1/(3|\startPoints|) \leq 1/(|\startPoints|/2 -1)$, it holds that
				\[\frac{R- \gamma R}{|\startPoints|/2} \geq \frac{(1- 1/(|\startPoints|/2 -1))R}{|\startPoints|/2} = R/(|\startPoints|/2 -1)  \geq \gamma R.\]
				Hence, by the pigeon principle, one of the sets, say $\jobsLTS{\tau}^{\iota}$, must have a resource requirement of at least $\gamma R$.
			\end{proofClaim} 
			Let $\jobsLTS{\tau}^{\iota}$ be this set.
			To generate a gap, we shift down all jobs in $\jobsLTS{\tau}^{\iota}$ back to their primary start position, see Figure \ref{fig:SRCS:ShiftCase11}.
			\begin{claim}
				As a result of this shift, there are at least $\gamma R$ free resources at each point between $\iota$ and $\iota + \addT$. 
			\end{claim}
			\begin{proofClaim}
				At each point between $\iota$ and $T'$ there were $\gamma R$ unused resources before. 
				Each job which starts between $T'$ and $\tau + \addT$ is an element of $\jobsLTS{\tau}$ and was therefore scheduled in parallel to the jobs in $\jobsLTS{\tau}^{\iota}$. 
				Therefore, at each point between $T'$ and $\tau + \addT$ at least $\gamma R$ resources are unused. 
				From $\tau + \addT$ to $\iota + \addT$ the jobs $\jobsLTS{\tau}^{\iota}$ were scheduled, so there are at least $\gamma R$ free resources. 
			\end{proofClaim}
			
			We now have to differentiate if there are at least $k$ machines unused between $\tau + \addT$ and $\iota + \addT$, see Figure \ref{fig:SRCS:ShiftCase11}.
			Let $\rho \in \{s| \tau \leq s \leq T, s \in S\}$ be the first point in the schedule where at most $k$ jobs from $\jobsT{\tau}$ are scheduled in the given optimal schedule (not the shifted one), i.e., $\rho$ is the first point in time where $|\jobsTInter{\rho}| \leq k$.
			Note that as a consequence $|\jobsT{\rho}| \geq k$ since otherwise there would have been a point in time before $\rho$, where at most $k$ machines are used by jobs starting before $T'/2$.
			We know that between $T'$ and $\rho + \addT$ there always will be $k$ machines unused since before the first shift they were blocked by jobs in $\jobsT{\tau}$.
			
			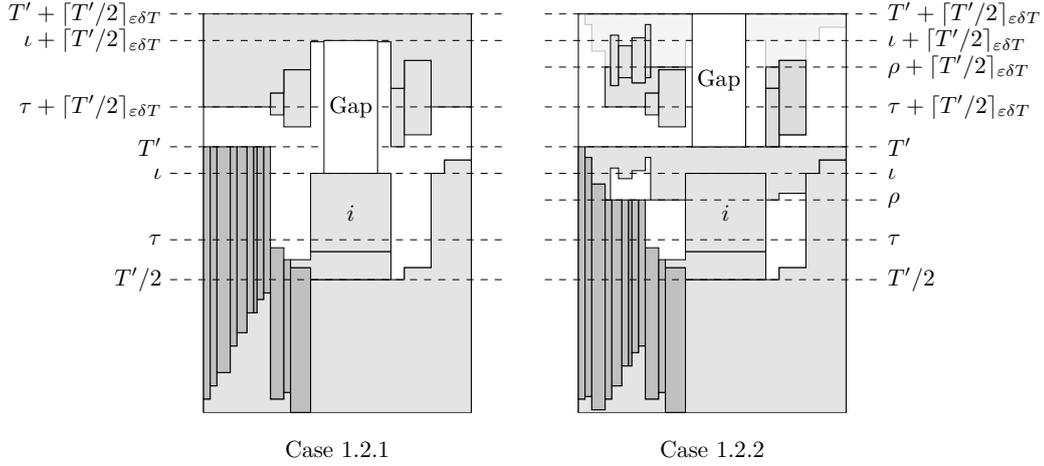
\begin{figure}[ht]
				\centering
				\resizebox{\textwidth}{!}{
				\begin{tikzpicture}				
				\pgfmathsetmacro{\h}{2}
				\pgfmathsetmacro{\w}{4}
				\pgfmathsetmacro{\t}{\h + 0.3*\h}
				\pgfmathsetmacro{\lm}{0.25 *\w}
				\pgfmathsetmacro{\rm}{0.85*\w}
				\pgfmathsetmacro{\rrm}{0.75*\w}
				\pgfmathsetmacro{\i}{\h + 0.8*\h}
				\pgfmathsetmacro{\r}{\h + 0.9*\h}				
				\pgfmathsetmacro{\rt}{\t - 0.2 *\t +0.2* \h}
				\pgfmathsetmacro{\rtt}{\t - 0.5 *\t +0.5* \h}
				\pgfmathsetmacro{\rttt}{\t - 0.7 *\t +0.7* \h}
				\draw   (0,0) rectangle (\w,3*\h);
				\node[below] at (\w/2,-2ex) {Case 1.2.1};
				\draw (\rrm,\h)  --(\rrm,\t- 0.7*\t + 0.7*\h) -- (\rm,\t- 0.7*\t + 0.7*\h)-- (\rm,\i)-- (0.9*\w,\i) -- (0.9*\w,\i + 1.0 *\h - 0.5*\i)-- (\w,\i + 1.0 *\h - 0.5*\i);
				
				\foreach \x/\xx/\y in {
					0.0/0.1/0.1,
					0.1/0.2/0.2,
					0.2/0.4/0.3,
					0.4/0.5/0.5,
					0.5/0.65/0.6,
					0.65/0.75/0.75,
					0.75/0.8/0.75,
					0.8/0.9/0.85,
					0.9/1.0/0.9
				}
				{
					\drawTallItem{\xx*\lm}{\y*\h}{\x *\lm}{2*\h};	
				}
				\drawTallItem{\lm}{\rt}{\lm+0.1*\rrm -0.1*\lm}{0.2};
				\drawTallItem{\lm+0.1*\rrm -0.1*\lm}{\rtt}{\lm+0.15*\rrm -0.15*\lm}{0.3};
				\drawTallItem{\lm+0.15*\rrm -0.15*\lm}{\rttt}{\lm+0.3*\rrm -0.3*\lm}{0};
				
				
				\drawVerticalItem{\lm}{\rt + \h}{\lm+0.1*\rrm -0.1*\lm}{\t+\h+0.15*2*\h -0.15*\t};
				\drawVerticalItem{\lm +0.1*\rrm -0.1*\lm}{\rtt +\h}{\lm +0.3*\rrm -0.3*\lm}{\t+\h+0.4*2*\h -0.4*\t};
				\drawVerticalItem{\lm +0.9*\rrm -0.9*\lm}{2*\h}{\rrm}{\t+\h+0.2*2*\h -0.2*\t};
				\drawVerticalItem{\rrm}{\t+\h- 0.7*\t + 0.7*\h}{\rm}{\t+\h+0.5*2*\h -0.5*\t};
				
				\begin{scope}[yshift = -\h cm]
				\drawVerticalItem[$i$]{\lm +0.3*\rrm -0.3*\lm}{\t+\h - 0.3*\t + 0.3*\h}{\lm +0.9*\rrm -0.9*\lm}{\i +\h};
				\end{scope}
				
				
				\draw[fill = white!85!black, fill opacity = 0.7](\lm+0.3*\rrm -0.3*\lm,0)-- (\lm+0.3*\rrm -0.3*\lm,\h) --(\rrm,\h)  --(\rrm,\t- 0.7*\t + 0.7*\h) -- (\rm,\t- 0.7*\t + 0.7*\h)-- (\rm,\i)-- (0.9*\w,\i) -- (0.9*\w,\i + 1.0 *\h - 0.5*\i)-- (\w,\i + 1.0 *\h - 0.5*\i) -- (\w,0);
				
				\draw[fill = white!85!black, fill opacity = 0.7] (0.0,0.0) -- (0.0,0.1*\h) -- (0.1*\lm,0.1*\h)--(0.1*\lm,0.2*\h) -- (0.2*\lm,0.2*\h)--(0.2*\lm,0.3*\h)--(0.4*\lm,0.3* \h)--(0.4*\lm,0.5* \h) -- (0.5*\lm,0.5* \h)-- (0.5*\lm,0.6* \h) --(0.65*\lm,0.6*\h)--(0.65*\lm,0.75*\h)--(0.8*\lm,0.75*\h)--(0.8*\lm,0.85*\h)--(0.9*\lm,0.85* \h)--(0.9*\lm,0.9* \h)--(\lm,0.9* \h) -- (\lm,0.2) --(\lm+0.1*\rrm -0.1*\lm,0.2) --(\lm+0.1*\rrm -0.1*\lm,0.3)--(\lm+0.15*\rrm -0.15*\lm,0.3) --(\lm+0.15*\rrm -0.15*\lm,0.0);

				\draw[fill = white!85!black, fill opacity = 0.7](0.0,3*\h)--(0.0,\t +\h)--(\lm,\t+\h)--(\lm,\t+\h+0.15*2*\h -0.15*\t)--(\lm+0.1*\rrm -0.1*\lm,\t+\h+0.15*2*\h -0.15*\t) -- (\lm+0.1*\rrm -0.1*\lm,\t+\h+0.4*2*\h -0.4*\t)--(\lm +0.3*\rrm -0.3*\lm,\t+\h+0.4*2*\h -0.4*\t)--(\lm +0.3*\rrm -0.3*\lm,\t+\h+0.7*2*\h -0.7*\t)--(\lm +0.9*\rrm -0.9*\lm,\t+\h+0.7*2*\h -0.7*\t)--(\lm +0.9*\rrm -0.9*\lm,\t+\h+0.2*2*\h -0.2*\t)--(\rrm,\t+\h+0.2*2*\h -0.2*\t)--(\rrm,\t+\h+0.5*2*\h -0.5*\t)--(\rm,\t+\h+0.5*2*\h -0.5*\t)--(\rm,\t+\h)--(\w,\t+\h)--(\w,3*\h);
				
				\draw[fill = white!85!black, fill opacity = 0.7]
				(\lm+0.15*\rrm -0.15*\lm,\rttt) -- 
				(\lm +0.3*\rrm -0.3*\lm,\rttt) -- 
				(\lm +0.3*\rrm -0.3*\lm,\h)--
				(\lm +0.9*\rrm -0.9*\lm,\h)--
				(\lm +0.9*\rrm -0.9*\lm,\t - 0.3*\t + 0.3*\h) --
				(\lm +0.3*\rrm -0.3*\lm,\t - 0.3*\t + 0.3*\h)--
				(\lm +0.3*\rrm -0.3*\lm,\rtt)--
				(\lm+0.15*\rrm -0.15*\lm,\rtt) -- 
				(\lm+0.15*\rrm -0.15*\lm,\rttt);
				
				\draw[dashed] (-0.5,\h) node[left] {$T'/2$} -- (\w +0.5,\h);
				\draw[dashed] (-0.5,2*\h) node[left] {$T'$} -- (\w +0.5,2*\h);
				\draw[dashed] (-0.5,3*\h) node[left] {$T' + \addT$} -- (\w +0.5,3*\h);
				
				\draw[dashed] (-0.5,\t) node[left] {$\tau$} -- (\w+0.5,\t) node[right] {};
				\draw[dashed] (-0.5,\t+\h) node[left] {$\tau + \addT$} -- (\w+0.5,\t+\h);
				\draw[dashed](-0.5,\i +\h) node[left] {$\iota +\addT$} -- (\w+0.5,\i+\h);
				\draw[dashed](-0.5,\i) node[left] {$\iota$} -- (\w+0.5,\i);

				\draw[fill = white] (\lm +0.4*\rrm -0.4*\lm,\i) rectangle node[midway]{Gap} (\lm +0.8*\rrm -0.8*\lm,\i+\h);
				
				\begin{scope}[xshift = 1.4*\w cm]
				\pgfmathsetmacro{\h}{2}
				\pgfmathsetmacro{\w}{4}
				\pgfmathsetmacro{\t}{\h + 0.3*\h}
				\pgfmathsetmacro{\lm}{0.25 *\w}
				\pgfmathsetmacro{\rm}{0.85*\w}
				\pgfmathsetmacro{\rrm}{0.75*\w}
				\pgfmathsetmacro{\i}{\h + 0.8*\h}
				\pgfmathsetmacro{\rt}{\t - 0.2 *\t +0.2* \h}
				\pgfmathsetmacro{\rtt}{\t - 0.5 *\t +0.5* \h}
				\pgfmathsetmacro{\rttt}{\t - 0.7 *\t +0.7* \h}
				\pgfmathsetmacro{\je}{\t+1.6}
				\pgfmathsetmacro{\tt}{\h + 0.6*\h}
				
				\pgfmathsetmacro{\tth}{\tt +\h}
				\pgfmathsetmacro{\tthx}{2*\h -\tt}
				\pgfmathsetmacro{\r}{\tt}
				\pgfmathsetmacro{\ryshift}{2*\h -\r}

				\begin{scope}[lightgray]
				\draw [fill = white!95!black, fill opacity = 0.7] (0.1 *\lm,3*\h) -- (0.1 *\lm,3*\h-0.2*2*\h + 0.2*\tt) -- (0.2 *\lm,3*\h-0.2*2*\h + 0.2*\tt) -- (0.2*\lm,3*\h-0.7*2*\h + 0.7*\tt)  -- (0.4*\lm,3*\h-0.7*2*\h + 0.7*\tt) -- (0.4*\lm,\r + \h) -- (0.12 *\w,\r + \h) -- (0.12 *\w,3*\h + 0.4*\r - 2*0.4*\h) -- (0.15*\w,3*\h + 0.4*\r - 2*0.4*\h) -- (0.15*\w,3*\h + 0.6*\r - 2*0.6*\h) -- (0.2*\w,3*\h + 0.6*\r - 2*0.6*\h) -- (0.2*\w,3*\h + 0.45*\r - 2*0.45*\h) -- (0.25*\w,3*\h + 0.45*\r - 2*0.45*\h) -- (0.25*\w,3*\h + 0.2*\r - 2*0.2*\h) -- (0.27*\w,3*\h + 0.2*\r - 2*0.2*\h) -- (0.27*\w,\r+\h) -- (\lm +0.3*\rrm -0.3*\lm,\r+\h) -- (\lm +0.3*\rrm -0.3*\lm,\i +\h) -- (\lm +0.9*\rrm -0.9*\lm,\i +\h) -- (\lm +0.9*\rrm -0.9*\lm,\r +\h) -- (\rrm,\r +\h) -- (\rrm,\t+\h+0.5*2*\h -0.5*\t) -- (\rm,\t+\h+0.5*2*\h -0.5*\t) --  (\rm,\i+\h)-- (0.9*\w,\i+\h) -- (0.9*\w,\i + 1.0 *\h - 0.5*\i+\h)-- (\w,\i + 1.0 *\h - 0.5*\i+\h) --(\w,3*\h) -- (0.1 *\lm,3*\h);
				\end{scope}
				
				\draw   (0,0) rectangle (\w,3*\h);
				\node[below] at (\w/2,-2ex) {Case 1.2.2};

				\draw (\rrm,\h)  --(\rrm,\t- 0.7*\t + 0.7*\h) -- (\rm,\t- 0.7*\t + 0.7*\h)-- (\rm,\i)-- (0.9*\w,\i) -- (0.9*\w,\i + 1.0 *\h - 0.5*\i)-- (\w,\i + 1.0 *\h - 0.5*\i);
				
				\begin{scope}[xscale = \lm/\w]
				\drawTallItem{0.0}{0.1*\h}{0.1 *\w}{2*\h};
				\begin{scope}[yshift = -0.2*2*\h cm + 0.2*\tt cm]
				\drawTallItem{0.1*\w}{0.2*\h}{0.2*\w}{2*\h};
				\end{scope}
				\begin{scope}[yshift = -0.7*2*\h cm + 0.7*\tt cm]
				\drawTallItem{0.2*\w}{0.3*\h}{0.4*\w}{2*\h};
				\end{scope}
				\end{scope}
				
				\begin{scope}[yshift = -\ryshift cm]
				\foreach \x/\xx/\y in {
					0.4/0.5/0.5,
					0.5/0.65/0.6,
					0.65/0.75/0.75,
					0.75/0.8/0.75,
					0.8/0.9/0.85,
					0.9/1.0/0.9
				}
				{
					\drawTallItem{\xx*\lm}{\y*\h}{\x *\lm}{2*\h};	
				}
				\end{scope}
				\drawTallItem{\lm}{\rt}{\lm+0.1*\rrm -0.1*\lm}{0.2};
				\drawTallItem{\lm+0.1*\rrm -0.1*\lm}{\rtt}{\lm+0.15*\rrm -0.15*\lm}{0.3};
				\drawTallItem{\lm+0.15*\rrm -0.15*\lm}{\rttt}{\lm+0.3*\rrm -0.3*\lm}{0};
				

				\drawVerticalItem{\lm}{\rt + \h}{\lm+0.1*\rrm -0.1*\lm}{\t+\h+0.15*2*\h -0.15*\t};
				\drawVerticalItem{\lm +0.1*\rrm -0.1*\lm}{\rtt +\h}{\lm +0.3*\rrm -0.3*\lm}{\t+\h+0.4*2*\h -0.4*\t};
				\drawVerticalItem{\lm +0.9*\rrm -0.9*\lm}{2*\h}{\rrm}{\t+\h+0.2*2*\h -0.2*\t};
				\drawVerticalItem{\rrm}{\t+\h- 0.7*\t + 0.7*\h}{\rm}{\t+\h+0.5*2*\h -0.5*\t};
				
				\begin{scope}[yshift = -\h cm]
				\drawVerticalItem[$i$]{\lm +0.3*\rrm -0.3*\lm}{\t+\h - 0.3*\t + 0.3*\h}{\lm +0.9*\rrm -0.9*\lm}{\i +\h};
				\end{scope}

				\drawVerticalItem{0.12 *\w}{\r +\h + 0.35*\r - 2*0.35*\h}{0.15 *\w}{3*\h + 0.4*\r - 2*0.4*\h};
				\drawVerticalItem{0.15*\w}{\r +\h + 0.2*\r - 2*0.2*\h}{0.2*\w}{3*\h + 0.6*\r - 2*0.6*\h};
				\drawVerticalItem{0.2*\w}{\r +\h+ 0.3*\r - 2*0.3*\h}{0.25*\w}{3*\h + 0.45*\r - 2*0.45*\h};
				\drawVerticalItem{0.25*\w}{\r +\h+ 0.2*\r - 2*0.2*\h}{0.27*\w}{3*\h + 0.2*\r - 2*0.2*\h};
				
				
				\begin{scope}[yshift = -\h cm]
				\draw [fill = white!85!black, fill opacity = 0.7] (0.1 *\lm,3*\h) -- (0.1 *\lm,3*\h-0.2*2*\h + 0.2*\tt) -- (0.2 *\lm,3*\h-0.2*2*\h + 0.2*\tt) -- (0.2*\lm,3*\h-0.7*2*\h + 0.7*\tt)  -- (0.4*\lm,3*\h-0.7*2*\h + 0.7*\tt) -- (0.4*\lm,\r + \h) -- (0.12 *\w,\r + \h) -- (0.12 *\w,3*\h + 0.4*\r - 2*0.4*\h) -- (0.15*\w,3*\h + 0.4*\r - 2*0.4*\h) -- (0.15*\w,3*\h + 0.6*\r - 2*0.6*\h) -- (0.2*\w,3*\h + 0.6*\r - 2*0.6*\h) -- (0.2*\w,3*\h + 0.45*\r - 2*0.45*\h) -- (0.25*\w,3*\h + 0.45*\r - 2*0.45*\h) -- (0.25*\w,3*\h + 0.2*\r - 2*0.2*\h) -- (0.27*\w,3*\h + 0.2*\r - 2*0.2*\h) -- (0.27*\w,\r+\h) -- (\lm +0.3*\rrm -0.3*\lm,\r+\h) -- (\lm +0.3*\rrm -0.3*\lm,\i +\h)  -- (\lm +0.9*\rrm -0.9*\lm,\i +\h) -- (\lm +0.9*\rrm -0.9*\lm,\r +\h) -- (\rrm,\r +\h) -- (\rrm,\t+\h+0.5*2*\h -0.5*\t) -- (\rm,\t+\h+0.5*2*\h -0.5*\t) --  (\rm,\i+\h)-- (0.9*\w,\i+\h) -- (0.9*\w,\i + 1.0 *\h - 0.5*\i+\h)-- (\w,\i + 1.0 *\h - 0.5*\i+\h) --(\w,3*\h) -- (0.1 *\lm,3*\h);
				\end{scope}

				\draw[fill = white!85!black, fill opacity = 0.7](\lm+0.3*\rrm -0.3*\lm,0)-- (\lm+0.3*\rrm -0.3*\lm,\h) --(\rrm,\h)  --(\rrm,\t- 0.7*\t + 0.7*\h) -- (\rm,\t- 0.7*\t + 0.7*\h)-- (\rm,\i)-- (0.9*\w,\i) -- (0.9*\w,\i + 1.0 *\h - 0.5*\i)-- (\w,\i + 1.0 *\h - 0.5*\i) -- (\w,0);
				
				\draw[fill = white!85!black, fill opacity = 0.7] 
				(0.0,0.0) -- 
				(0.0,0.1*\h) -- 
				(0.1*\lm,0.1*\h ) --
				(0.1*\lm,0.2*\h -0.2*2*\h + 0.2*\tt) -- 
				(0.2*\lm,0.2*\h-0.2*2*\h + 0.2*\tt) --
				(0.2*\lm,0.3*\h -0.7*2*\h  + 0.7*\tt) -- 
				(0.4*\lm,0.3* \h -0.7*2*\h  + 0.7*\tt) --
				(0.4*\lm,0.5* \h -\ryshift) -- 
				(0.5*\lm,0.5* \h -\ryshift) -- 
				(0.5*\lm,0.6* \h -\ryshift) --
				(0.65*\lm,0.6*\h -\ryshift) -- 
				(0.65*\lm,0.75*\h -\ryshift) --
				(0.8*\lm,0.75*\h -\ryshift) -- 
				(0.8*\lm,0.85*\h -\ryshift) -- 
				(0.9*\lm,0.85* \h -\ryshift) -- 
				(0.9*\lm,0.9* \h -\ryshift) --
				(\lm,0.9* \h -\ryshift) -- 
				(\lm,0.2) --
				(\lm+0.1*\rrm -0.1*\lm,0.2) -- 
				(\lm+0.1*\rrm -0.1*\lm,0.3) --
				(\lm+0.15*\rrm -0.15*\lm,0.3) --
				(\lm+0.15*\rrm -0.15*\lm,0.0);
				
				\draw[fill = white!85!black, fill opacity = 0.7]
				(\lm+0.15*\rrm -0.15*\lm,\rttt) -- 
				(\lm +0.3*\rrm -0.3*\lm,\rttt) -- 
				(\lm +0.3*\rrm -0.3*\lm,\h)--
				(\lm +0.9*\rrm -0.9*\lm,\h)--
				(\lm +0.9*\rrm -0.9*\lm,\t - 0.3*\t + 0.3*\h) --
				(\lm +0.3*\rrm -0.3*\lm,\t - 0.3*\t + 0.3*\h)--
				(\lm +0.3*\rrm -0.3*\lm,\rtt)--
				(\lm+0.15*\rrm -0.15*\lm,\rtt) -- 
				(\lm+0.15*\rrm -0.15*\lm,\rttt);
				
				\draw[fill = white!85!black, fill opacity = 0.7]
				(0.4*\lm,\r+\h) --
				(0.12 *\w,\r+\h) --
				(0.12 *\w,\r +\h + 0.35*\r - 2*0.35*\h) --
				(0.15 *\w,\r +\h + 0.35*\r - 2*0.35*\h) -- 
				(0.15 *\w,\r +\h + 0.2*\r - 2*0.2*\h) -- 
				(0.2*\w,\r +\h + 0.2*\r - 2*0.2*\h) -- 
				(0.2*\w,\r +\h+ 0.3*\r - 2*0.3*\h) --
				(0.25*\w,\r +\h+ 0.3*\r - 2*0.3*\h) --
				(0.25*\w,\r +\h+ 0.2*\r - 2*0.2*\h) --
				(0.27*\w,\r +\h+ 0.2*\r - 2*0.2*\h) --
				(0.27*\w,\r+\h) -- 
				(\lm +0.3*\rrm -0.3*\lm,\r+\h) -- 
				(\lm +0.3*\rrm -0.3*\lm,\t+\h+0.4*2*\h -0.4*\t) -- 
				(\lm +0.1*\rrm -0.1*\lm,\t+\h+0.4*2*\h -0.4*\t) --
				(\lm +0.1*\rrm -0.1*\lm,\t+\h+0.15*2*\h -0.15*\t) --
				(\lm,\t+\h+0.15*2*\h -0.15*\t) --
				(\lm,\t+\h) -- 
				(0.4*\lm,\t+\h) --
				(0.4*\lm,\r+\h);
				
				\draw[fill = white!85!black, fill opacity = 0.7]
				(\lm +0.9*\rrm -0.9*\lm,\t+\h+0.2*2*\h -0.2*\t) -- 
				(\lm +0.9*\rrm -0.9*\lm,\r+\h) -- 
				(\rrm,\r+\h) -- 
				(\rrm,\t+\h+0.2*2*\h -0.2*\t) --
				(\lm +0.9*\rrm -0.9*\lm,\t+\h+0.2*2*\h -0.2*\t);
				
				\drawVerticalItem{\lm +0.9*\rrm -0.9*\lm}{2*\h}{\rrm}{\t+\h+0.2*2*\h -0.2*\t};
				\drawVerticalItem{\rrm}{\t+\h- 0.7*\t + 0.7*\h}{\rm}{\t+\h+0.5*2*\h -0.5*\t};

				\draw[dashed] (-0.5,\h) -- (\w +0.5,\h) node[right] {$T'/2$};
				\draw[dashed] (-0.5,2*\h) -- (\w +0.5,2*\h) node[right] {$T'$};
				\draw[dashed] (-0.5,3*\h) -- (\w +0.5,3*\h) node[right] {$T' + \addT$};
				\draw[dashed] (-0.5,\t) -- (\w+0.5,\t)  node[right] {$\tau$};
				\draw[dashed] (-0.5,\t+\h) -- (\w+0.5,\t+\h)  node[right] {$\tau + \addT$};
				\draw[dashed] (-0.5,\r) -- (\w+0.5,\r) node[right] {$\rho$} ;
				\draw[dashed] (-0.5,\r+\h) -- (\w+0.5,\r+\h) node[right] {$\rho + \addT$};
				\draw[dashed](-0.5,\i +\h) -- (\w+0.5,\i+\h) node[right] {$\iota + \addT$};
				\draw[dashed](-0.5,\i) -- (\w+0.5,\i) node[right] {$\iota$};
				
				\draw[fill=white] (\lm +0.35*\rrm -0.35*\lm,2*\h) rectangle (\lm +0.75*\rrm -0.75*\lm,3*\h)node[midway] {Gap};
				\end{scope}
				\end{tikzpicture}  
			}
				\caption{Examples for the shifted schedule and the position of the gap in the Cases 1.2.1 and 1.2.2}
				\label{fig:SRCS:ShiftCase121}
			\end{figure}
			
			\begin{caseList}
				\item[$\rho \geq \iota$.]
				In this case, at each point between $\iota$ and $\iota + \addT$ there are $k$ machines unused. 
				Between $\iota$ and $T'$ there are $k$ free machines by the choice of $\tau$ and between $T'$ and $\rho +\addT$ there are $k$ free machines by the choice of $\rho$. 
				Therefore, there is a gap between $\iota$ and $\iota + \addT$, see Figure \ref{fig:SRCS:ShiftCase121}.
				
				Similar as in Case 1.1. the total number of guesses needed to place the huge jobs is bounded by $m^\Oh_{\eps}(1)$, although we have to add the guess for $\rho$.

				\item[$\rho < \iota$.] 
				Let $\jobs_{H,\rho} := \{j \in \jobs_{H}| s_j + p_j > \rho\}$ be the set of huge jobs, which are still scheduled after $\rho$. 
				It holds that $|\jobs_{H,\rho}| \leq k$. 
				As a consequence, it is possible to guess their starting positions in polynomial time. 
				Therefore, the algorithm will schedule each job in $\jobs_{H,\rho}$ as in the original simplified schedule $\OPT_{\mathrm{rounded}}$. 
				The other huge jobs, which end between $\tau$ and $\rho$, are scheduled such that they end at $\rho$, i.e., we define $\sched'(j) := \rho - \pT{j}$ for each of these huge jobs $j$. 
				Next, we shift the all the jobs $j$ with starting time $\sigma'(j) \geq \rho + \addT$ downwards such that they start as they had started before the first shift.
				As a result between $T'$ and $T' +\addT$, there are just jobs left which overlap the time from $\tau + \addT$ to $\rho + \addT$, see Figure \ref{fig:SRCS:ShiftCase121}.

				By the choice of $\rho$ and $\tau$ at each point between $\tau + \addT$ and $\rho + \addT$ there are at most $m-k$ jobs which use at most $R-\gamma R$ resource since the job $i$ was scheduled there before. 
				Since each job between $T'$ and $T' +\addT$ overlaps this area there are at least $k$ free machines and $\gamma R$ free resources in this area.
				Hence, we position the gap at $T'$.
				
				In the algorithm, we will guess $\tau$ and $\rho$ dependent on a given fractional solution for the large jobs and guess the at most $k$ jobs ending before $\tau$ and the $k$ jobs ending after $\rho$ in $\mathcal{O}(m^{2k-1})$. 
				For each of these jobs, we have to guess its starting time out of at most $|\startPoints|/2$ possibilities. 
				
			\end{caseList}
			
		\end{caseList}
		
		\item[$\jR{\jobsT{\tau}} > R-\gamma R$.]
		In this case, the gap has to start strictly after $\tau$ since at $\tau$ there is not enough free resource.
		Let $\jobs_{L,T'/2}$ be the set of large jobs intersecting the point in time $T'/2$.
		Remember that $\jobsT{\tau}$ contains huge and large jobs.
		Since $\jR{\jobsT{\tau}} > R- \gamma R$ at least one of these stets of jobs (huge or large) has to contribute a large resource requirement to $\jR{\jobsT{\tau}}$.
		In the following, we will find the gap, depending on which of both sets contributes a suitable large resource requirement.
	
		\begin{caseList}
			\item[$\jR{\jobs_{L,T'/2}} \geq 2 \gamma R$.]
			Let $\tau'\in \{s \in S| \tau \leq s \leq T'\}$ be the first point in time where $\jR{\jobs_{L,T'/2}} - \jR{\jobsLTNS{\tau'}} \geq \gamma R$. 
			Note that $\tau \leq \tau'$ since, otherwise, there would be $\gamma R$ free resources at $\tau$.
			\begin{claim}
				By this choice at each point between $\tau'$ and $\tau'+\addT$ there are at least $\gamma R$ free resources.
			\end{claim} 	
			Between $\tau'$ and $T'$ there are $\gamma R$ free resources since jobs form $\jobs_{L,T'/2}$ with a resource requirement of at least $\gamma R$ end before $\tau'$. 
			On the other hand, before the shift there was at least $\gamma R$ resource blocked by jobs from $\jobs_{L,T'/2}$ between $T'/2$ and $\tau'$ and hence after the shift there is at least $\gamma R$ free resource at any time between $T'$ and $\tau' +\addT$.
			
			Moreover, as in Case1.2, let $\rho \in \{s| \tau \leq s \leq T', s \in S\}$ be the first point in the schedule where $|\jobsTInter{\rho}| \leq k$, i.e., where at most $k$ jobs are scheduled that start before $T'/2$. 
			\begin{claim}
				By this choice at each point between $\tau$ and $\rho+ \addT$ there are at there are at least $k$ unused machines.
			\end{claim} 
			From $\tau$ to $T'$ there are $k$ unused machines, by the choice of $\tau$. 
			On the other hand, at each point in time between $T'$ and $\rho + \addT$ there where $k$ machines blocked by jobs from that started before $T'/2$ and these machines are now unused.
			
			Similar as in Cases 1.2.1 and 1.2.2, we will find the gap dependent of the relation between $\tau'$ and $\rho$.
			
			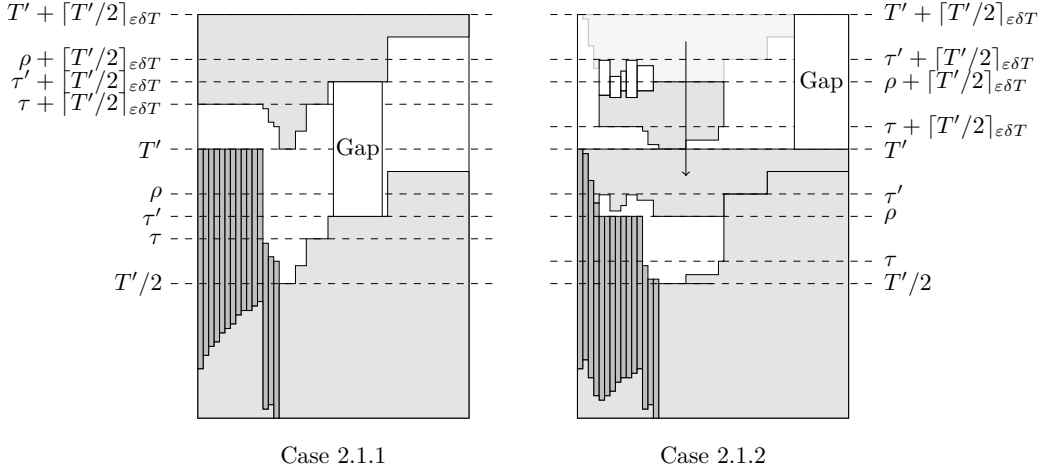
\begin{figure}[ht]
				\centering
				\resizebox{\textwidth}{!}{
				\begin{tikzpicture}		
				\pgfmathsetmacro{\h}{2}
				\pgfmathsetmacro{\w}{4}
				\pgfmathsetmacro{\t}{4*\h /3}
				\pgfmathsetmacro{\tt}{3*\h/2}
				\pgfmathsetmacro{\r}{5*\h/3}

				\draw   (0,0) rectangle (\w,3*\h);
				\node[below] at (\w/2,-2ex) {Case 2.1.1};
				\drawTallItem{0.0*\w/5}{\h -1.9*\h/3}{0.1*\w/5}{2*\h};
				\drawTallItem{0.1*\w/5}{\h -1.6*\h/3}{0.2*\w/5}{2*\h};
				\drawTallItem{0.2*\w/5}{\h -1.4*\h/3}{0.3*\w/5}{2*\h};
				\drawTallItem{0.3*\w/5}{\h -1.3*\h/3}{0.4*\w/5}{2*\h};
				\drawTallItem{0.4*\w/5}{\h -1.1*\h/3}{0.5*\w/5}{2*\h};
				\drawTallItem{0.5*\w/5}{\h -1.0*\h/3}{0.6*\w/5}{2*\h};
				\drawTallItem{0.6*\w/5}{\h -0.9*\h/3}{0.7*\w/5}{2*\h};
				\drawTallItem{0.7*\w/5}{\h -0.7*\h/3}{0.8*\w/5}{2*\h};
				\drawTallItem{0.8*\w/5}{\h -0.6*\h/3}{0.9*\w/5}{2*\h};
				\drawTallItem{0.9*\w/5}{\h -0.6*\h/3}{1.0*\w/5}{2*\h};
				\drawTallItem{1.0*\w/5}{\h -0.5*\h/3}{1.1*\w/5}{2*\h};
				\drawTallItem{1.1*\w/5}{\h -0.4*\h/3}{1.2*\w/5}{2*\h};
				\drawTallItem{1.3*\w/5}{\t -0.1*\h/3}{1.2*\w/5}{0.2*\h/3};
				\drawTallItem{1.4*\w/5}{\t -0.4*\h/3}{1.3*\w/5}{0.3*\h/3};
				\drawTallItem{1.5*\w/5}{\t -0.5*\h/3}{1.4*\w/5}{0*\h/3};
				
				\draw[fill = white!85!black, fill opacity = 0.7] 
				(1.8*\w/5,\h) --(1.8*\w/5,3.4*\h/3) -- 
				(2*\w/5,3.4*\h/3) -- (2*\w/5,\t) --
				(2.4*\w/5,\t) -- (2.4*\w/5,\tt)-- 
				(3.5*\w/5,\tt) -- (3.5*\w/5,5.5*\h/3)-- 
				(\w,5.5*\h/3) -- (\w,0) --
				(1.5*\w/5,0) -- 
				(1.5*\w/5,\h) --
				(1.8*\w/5,\h);
				
				\draw[fill = white!85!black, fill opacity = 0.7] 
				(0,0) -- (0,\h -1.9*\h/3) -- 
				(0.1*\w/5,\h -1.9*\h/3) -- (0.1*\w/5,\h -1.6*\h/3) -- 
				(0.2*\w/5,\h -1.6*\h/3) -- (0.2*\w/5,\h -1.4*\h/3) --
				(0.3*\w/5,\h -1.4*\h/3) -- (0.3*\w/5,\h -1.3*\h/3) --
				(0.4*\w/5,\h -1.3*\h/3) -- (0.4*\w/5,\h -1.1*\h/3) --
				(0.5*\w/5,\h -1.1*\h/3) -- (0.5*\w/5,\h -1.0*\h/3) --
				(0.6*\w/5,\h -1.0*\h/3) -- (0.6*\w/5,\h -0.9*\h/3) --
				(0.7*\w/5,\h -0.9*\h/3) -- (0.7*\w/5,\h -0.7*\h/3) --
				(0.8*\w/5,\h -0.7*\h/3) -- (0.8*\w/5,\h -0.6*\h/3) --
				(1.0*\w/5,\h -0.6*\h/3) -- (1.0*\w/5,\h -0.5*\h/3) --
				(1.1*\w/5,\h -0.5*\h/3) -- (1.1*\w/5,\h -0.4*\h/3) --
				(1.2*\w/5,\h -0.4*\h/3) -- (1.2*\w/5,0.2*\h/3) --
				(1.3*\w/5,0.2*\h/3)     -- (1.3*\w/5,0.3*\h/3) --
				(1.4*\w/5,0.3*\h/3)     -- (1.4*\w/5,0) --
				(0,0);
				
				\begin{scope}[yshift = \h cm]
				\draw[fill = white!85!black, fill opacity = 0.7] 
				(0,2*\h) -- (\w,2*\h) -- 
				(\w,5.5*\h/3) -- (3.5*\w/5,5.5*\h/3) --
				(3.5*\w/5,\tt) -- (2.4*\w/5,\tt) --
				(2.4*\w/5,\t) -- (2*\w/5,\t) --
				(2*\w/5,3.4*\h/3) -- (1.8*\w/5,3.4*\h/3) --
				(1.8*\w/5,\h) -- (1.5*\w/5,\h) --
				(1.5*\w/5,\t -0.5*\h/3) -- (1.4*\w/5,\t -0.5*\h/3) --
				(1.4*\w/5,\t -0.4*\h/3) -- (1.3*\w/5,\t -0.4*\h/3) --
				(1.3*\w/5,\t -0.1*\h/3) -- (1.2*\w/5,\t -0.1*\h/3) --
				(1.2*\w/5,\t) -- (0,\t) --
				(0,2*\h);
				\end{scope}
				
				\draw[dashed] (-0.5*\w/5,\h) node[left] {$T'/2$} -- (\w +0.5*\w/5,\h);
				\draw[dashed] (-0.5*\w/5,2*\h) node[left] {$T'$} -- (\w +0.5*\w/5,2*\h);
				\draw[dashed] (-0.5*\w/5,3*\h) node[left] {$T' + \addT$} -- (\w +0.5*\w/5,3*\h);

				\draw[dashed] (-0.5*\w/5,\tt) node[left] {$\tau'$} -- (\w+0.5*\w/5,\tt);
				\draw[dashed] (-0.5*\w/5,\tt+\h) node[left] {$\tau' + \addT$} -- (\w+0.5*\w/5,\tt+\h);
				
				\draw[dashed] (-0.5*\w/5,\t) node[left] {$\tau$} -- (\w+0.5*\w/5,\t);
				\draw[dashed] (-0.5*\w/5,\t+\h) node[left] {$\tau + \addT$} -- (\w+0.5*\w/5,\t+\h);
				
				\draw[dashed] (-0.5*\w/5,\r) node[left] {$\rho$} -- (\w+0.5*\w/5,\r);
				\draw[dashed] (-0.5*\w/5,\r+\h) node[left] {$\rho+ \addT$} -- (\w+0.5*\w/5,\r+\h);

				\draw[fill=white] (2.5*\w/5,\tt) rectangle (3.4*\w/5,\tt+\h)node[midway] {Gap};

				\begin{scope}[xshift = 1.4*\w cm]
					\pgfmathsetmacro{\h}{2}
				\pgfmathsetmacro{\w}{4}
				\pgfmathsetmacro{\tt}{5.0*\h/3}
				\pgfmathsetmacro{\r}{4.5*\h/3}
				\pgfmathsetmacro{\t}{3.5*\h/3}
				\pgfmathsetmacro{\ryshift}{2*\h -\r}
				\pgfmathsetmacro{\s}{1.4*\h/3}
				\pgfmathsetmacro{\ss}{0.8*\h/3}
				\pgfmathsetmacro{\sss}{0.3*\h/3}

				\pgfmathsetmacro{\rw}{0.4*\w/5}
				\pgfmathsetmacro{\rww}{0.6*\w/5}
				\pgfmathsetmacro{\rwww}{0.8*\w/5}
				\pgfmathsetmacro{\rwwww}{0.9*\w/5}
				\pgfmathsetmacro{\rwwwww}{1.1*\w/5}
				\pgfmathsetmacro{\rwwwwww}{1.4*\w/5}
				
				\pgfmathsetmacro{\rh}{\r +0.4*\w/5}
				\pgfmathsetmacro{\rhh}{\r +0.1*\w/5}
				\pgfmathsetmacro{\rhhh}{\r +0.2*\w/5}
				\pgfmathsetmacro{\rhhhh}{\r +0.4*\w/5}
				\pgfmathsetmacro{\rhhhhh}{\r +0.3*\w/5}
				
				\draw   (0,0) rectangle (\w,3*\h);
				\node[below] at (\w/2,-2ex) {Case 2.1.2};
				\begin{scope}[yshift = \h cm]
				\draw[lightgray, fill = white!95!black, fill opacity = 0.7] 
				(0.1*\w/5,2*\h) -- (0.1*\w/5, \r +\s) -- 
				(0.2*\w/5, \r +\s) -- (0.2*\w/5, \r +\ss) -- 
				(0.3*\w/5, \r +\ss) -- (0.3*\w/5, \r +\sss) -- 
				(\rw,\r +\sss) -- (\rw,\rh) -- 
				(\rww,\rh) -- (\rww,\rhh) -- 
				(\rwww,\rhh)  -- (\rwww,\rhhh)  -- 
				(\rwwww,\rhhh) -- (\rwwww,\rhhhh) -- 
				(\rwwwww,\rhhhh) -- (\rwwwww,\rhhhhh) -- 
				(\rwwwwww,\rhhhhh) -- (\rwwwwww,\r) -- 
				(2.7*\w/5,\r) -- (2.7*\w/5, \tt) -- 
				(3.5*\w/5,\tt) -- (3.5*\w/5,5.5*\h/3) -- 
				(\w,5.5*\h/3) -- (\w,2*\h) -- 
				(0.1*\w/5,2*\h);
				\end{scope}

				\drawTallItem{0.0*\w/5}{\h -1.9*\h/3}{0.1*\w/5}{2*\h};
				\drawTallItem{0.1*\w/5}{\h -1.6*\h/3 -\ryshift+\s}{0.2*\w/5}{2*\h-\ryshift+\s};
				\drawTallItem{0.2*\w/5}{\h -1.4*\h/3 -\ryshift + \ss}{0.3*\w/5}{2*\h-\ryshift + \ss};
				\drawTallItem{0.3*\w/5}{\h -1.3*\h/3-\ryshift+\sss}{0.4*\w/5}{2*\h-\ryshift+\sss};
				\drawTallItem{0.4*\w/5}{\h -1.1*\h/3-\ryshift}{0.5*\w/5}{2*\h-\ryshift};
				\drawTallItem{0.5*\w/5}{\h -1.0*\h/3-\ryshift}{0.6*\w/5}{2*\h-\ryshift};
				\drawTallItem{0.6*\w/5}{\h -0.9*\h/3-\ryshift}{0.7*\w/5}{2*\h-\ryshift};
				\drawTallItem{0.7*\w/5}{\h -0.7*\h/3-\ryshift}{0.8*\w/5}{2*\h-\ryshift};
				\drawTallItem{0.8*\w/5}{\h -0.6*\h/3-\ryshift}{0.9*\w/5}{2*\h-\ryshift};
				\drawTallItem{0.9*\w/5}{\h -0.6*\h/3-\ryshift}{1.0*\w/5}{2*\h-\ryshift};
				\drawTallItem{1.0*\w/5}{\h -0.5*\h/3-\ryshift}{1.1*\w/5}{2*\h-\ryshift};
				\drawTallItem{1.1*\w/5}{\h -0.4*\h/3-\ryshift}{1.2*\w/5}{2*\h-\ryshift};
				
				\drawTallItem{1.3*\w/5}{\t -0.1*\h/3}{1.2*\w/5}{0.2*\h/3};
				\drawTallItem{1.4*\w/5}{\t -0.4*\h/3}{1.3*\w/5}{0.3*\h/3};
				\drawTallItem{1.5*\w/5}{\t -0.4*\h/3}{1.4*\w/5}{0*\h/3};	
				
				\draw[fill = white!85!black, fill opacity = 0.7] 
				(2.0*\w/5,\h) --
				(2.0*\w/5,3.2*\h/3) -- 
				(2.6*\w/5,3.2*\h/3) -- 
				(2.6*\w/5,\t) --
				(2.7*\w/5,\t) -- 
				(2.7*\w/5,\tt) -- 
				(3.5*\w/5,\tt) -- 
				(3.5*\w/5,5.5*\h/3) -- 
				(\w,5.5*\h/3) -- 
				(\w,0) -- 
				(1.5*\w/5,0*\h/3) -- 
				(1.5*\w/5,\h) -- 
				(2,\h);	
				
				\draw[fill = white!85!black, fill opacity = 0.7] 
				(0.0*\w/5,0*\h/3) -- 
				(0.0*\w/5,\h -1.9*\h/3) --
				(0.1*\w/5,\h -1.9*\h/3) -- 
				(0.1*\w/5,\h -1.6*\h/3-\ryshift+\s) -- 
				(0.2*\w/5,\h -1.6*\h/3-\ryshift+\s) -- 
				(0.2*\w/5,\h -1.4*\h/3-\ryshift+\ss) -- 
				(0.3*\w/5,\h -1.4*\h/3-\ryshift+\ss) -- 
				(0.3*\w/5,\h -1.3*\h/3-\ryshift+\sss) -- 
				(0.4*\w/5,\h -1.3*\h/3-\ryshift+\sss) -- 
				(0.4*\w/5,\h -1.1*\h/3-\ryshift) -- 
				(0.5*\w/5,\h -1.1*\h/3-\ryshift) -- 
				(0.5*\w/5,\h -1.0*\h/3-\ryshift) -- 
				(0.6*\w/5,\h -1.0*\h/3-\ryshift) -- 
				(0.6*\w/5,\h -0.9*\h/3-\ryshift) -- 
				(0.7*\w/5,\h -0.9*\h/3-\ryshift) -- 
				(0.7*\w/5,\h -0.7*\h/3-\ryshift) -- 
				(0.8*\w/5,\h -0.7*\h/3-\ryshift) -- 
				(0.8*\w/5,\h -0.6*\h/3-\ryshift) -- 
				(1.0*\w/5,\h -0.6*\h/3-\ryshift) -- 
				(1.0*\w/5,\h -0.5*\h/3-\ryshift) -- 
				(1.1*\w/5,\h -0.5*\h/3-\ryshift) -- 
				(1.1*\w/5,\h -0.4*\h/3-\ryshift) -- 
				(1.2*\w/5,\h -0.4*\h/3-\ryshift) -- 
				(1.2*\w/5,0.2*\h/3) -- 
				(1.3*\w/5,0.2*\h/3) -- 
				(1.3*\w/5,0.1*\h) --
				(1.4*\w/5,0.1*\h) -- 
				(1.4*\w/5,0*\h/3) -- 
				(0,0);
				
				\draw[fill = white!85!black, fill opacity = 0.7] 
				(1.2*\w/5,\t -0.1*\h/3 +\h) --
				(1.3*\w/5,\t -0.1*\h/3+\h) --
				(1.3*\w/5,\t -0.4*\h/3+\h) --
				(1.5*\w/5,\t -0.4*\h/3+\h) --
				(1.5*\w/5,\h+\h) --
				(2.0*\w/5,\h+\h) --
				(2.0*\w/5,3.2*\h/3+\h) --
				(2.6*\w/5,3.2*\h/3+\h) --
				(2.6*\w/5,\t+\h) --
				(2.7*\w/5,\t+\h) -- 
				(2.7*\w/5,\r+\h) --
				(\rwwwwww,\r+\h) --
				(\rwwwwww,\r +\h-0.2*\h/3) --
				(\rwwwww,\r +\h-0.2*\h/3) --
				(\rwwwww,\r +\h-0.35*\h/3) --
				(\rwwww,\r +\h-0.35*\h/3) --
				(\rwwww,\r +\h-0.2*\h/3) --
				(\rwww,\r +\h-0.2*\h/3) -- 
				(\rwww,\r +\h-0.35*\h/3) --
				(\rww,\r +\h-0.35*\h/3) --
				(\rww,\r +\h-0.3*\h/3) --
				(\rw,\r +\h-0.3*\h/3) --
				(\rw,\t+\h) --
				(1.2*\w/5,\t+\h) --
				(1.2*\w/5,\t -0.1*\h/3 +\h);
				
				\draw (\rw,\r +\h-0.3*\h/3) rectangle (\rww,\rh + \h);
				\draw (\rww,\r +\h-0.35*\h/3) rectangle (\rwww,\rhh + \h);
				\draw (\rwww,\r +\h-0.2*\h/3) rectangle (\rwwww,\rhhh + \h);
				\draw (\rwwww,\r +\h-0.35*\h/3) rectangle (\rwwwww,\rhhhh + \h);
				\draw (\rwwwww,\r +\h-0.2*\h/3) rectangle (\rwwwwww,\rhhhhh + \h);

				\draw[fill = white!85!black, fill opacity = 0.7] 
				(0.1*\w/5,2*\h) -- (0.1*\w/5, \r +\s) -- 
				(0.2*\w/5, \r +\s) -- (0.2*\w/5, \r +\ss) -- 
				(0.3*\w/5, \r +\ss) -- (0.3*\w/5, \r +\sss) -- 
				(\rw,\r +\sss) -- (\rw,\rh) -- 
				(\rww,\rh) -- (\rww,\rhh) -- 
				(\rwww,\rhh)  -- (\rwww,\rhhh)  -- 
				(\rwwww,\rhhh) -- (\rwwww,\rhhhh) -- 
				(\rwwwww,\rhhhh) -- (\rwwwww,\rhhhhh) -- 
				(\rwwwwww,\rhhhhh) -- (\rwwwwww,\r) -- 
				(2.7*\w/5,\r) -- (2.7*\w/5, \tt) -- 
				(3.5*\w/5,\tt) -- (3.5*\w/5,5.5*\h/3) -- 
				(\w,5.5*\h/3) -- (\w,2*\h) -- (0.1*\w/5,2*\h);
				
				\draw[dashed] (-0.5*\w/5,1*\h) -- (\w +0.5*\w/5,\h) node[right] {$T'/2$};
				\draw[dashed] (-0.5*\w/5,2*\h) -- (\w +0.5*\w/5,2*\h) node[right] {$T'$};
				\draw[dashed] (-0.5*\w/5,3*\h) -- (\w +0.5*\w/5,3*\h) node[right] {$T' +\addT$}; 		
				\draw[dashed] (-0.5*\w/5,\tt) -- (\w+0.5*\w/5,\tt) node[right] {$\tau'$};
				\draw[dashed] (-0.5*\w/5,\r) -- (\w+0.5*\w/5,\r) node[right] {$\rho$};
				\draw[dashed] (-0.5*\w/5,\t) -- (\w+0.5*\w/5,\t) node[right] {$\tau$};
				\draw[dashed] (-0.5*\w/5,\tt+\h) -- (\w+0.5*\w/5,\tt+\h) node[right] {$\tau' + \addT$};
				\draw[dashed] (-0.5*\w/5,\r+\h) -- (\w+0.5*\w/5,\r+\h) node[right] {$\rho + \addT$};
				\draw[dashed] (-0.5*\w/5,\t+\h) -- (\w+0.5*\w/5,\t+\h) node[right] {$\tau+\addT$};

				\draw[fill=white] (4.0*\w/5,2*\h) rectangle (\w,3*\h)node[midway] {Gap};
				
				\draw[->] (2.0*\w/5, 2.8 * \h) -- (2.0*\w/5, 1.8 * \h);
				\end{scope}
				
				\end{tikzpicture} 
			}
				
				\caption{Examples for the shifted schedules and the position of the gap in Cases 2.1.1 and Case 2.1.2}
				\label{fig:SRCS:ShiftCase211}
			\end{figure}
			
			\begin{caseList}
				\item[$\rho \geq \tau'$.]	
				In this case between $\tau'$ and $\tau' + \addT$ there are at least $k$ unused machines. 
				Therefore, we have a gap between these two points, which is large enough, see Figure \ref{fig:SRCS:ShiftCase211}. 
				In the algorithm, we have to guess the $k$ huge jobs, which end before $\tau$ and their start point, as well as the points $\tau$, $\tau'$ and $\rho$. All the guesses for this case can be iterated in polynomial time. 
				
				\item[$\rho < \tau'$.]
				In this case, we act like in Case 1.2.2 and shift all huge jobs, but the at most $k$ jobs ending after $\rho$, downwards such that they end at $\rho$, see Figure \ref{fig:SRCS:ShiftCase211}. 
				Furthermore, we shift all jobs starting after $\rho +\addT$ back downwards such that they again start at their primary start position. 
				Now after $T'$ there are just jobs having their start or end position between $\tau + \addT$ and $\rho + \addT$. 
				At each point between these two points there are at least $k$ unused machines and $\gamma R$ unused resource with the same arguments as in Case 1.2.2. 
				Hence, we have a gap with the right properties between $T'$ and $T' + \addT$. 
				All the guesses for this case can be iterated in polynomial time. 
				
			\end{caseList}
			\item[$\jR{\jobs_{L,T'/2}} < 2 \gamma R$.]
			Since we have $\jR{\jobsT{\tau}} > R- \gamma R$ (by Case 2.) and it holds that $(\jobs_H \cap \jobsT{\tau}) \cup (\jobs_{L,T'/2}\cap \jobsT{\tau})= \jobsT{\tau}$ we get that $\jR{\jobs_H \cap \jobsT{\tau}} \geq R-3 \gamma R$.
			
			Similar as before, let $\rho \in \{s| \tau \leq s \leq T , s \in S\}$ be the first point in the schedule where less than $k$ jobs are scheduled that start before $T'/2$. 
			By the same argument as in Case 2.1, we know that at every point between $\tau$ and $\rho + \addT$ there are at least $k$ unused machines in the shifted schedule.
			
			\begin{caseList}
				\item[$\jR{\jobsT{\rho}} \geq \gamma R$.] 
				In this case, we can construct a schedule in the same way as in case 1.2.2 or 2.1.2 by shifting down the jobs that start after $\rho +\addT$ and positioning the gap at $T'$, see Figure \ref{fig:SRCS:ShiftCase221}.
				This is possible because the jobs that are scheduled between $\tau + \addT$ and $\rho+\addT$ can use at most $R- \gamma R$ resources in this case since $\jR{\jobsT{\rho}} \geq \gamma R$ and hence at least $ \gamma R$ resources are blocked by the jobs in $\jobsT{\rho}$.
				All the guesses for this case can be iterated in polynomial time. 
				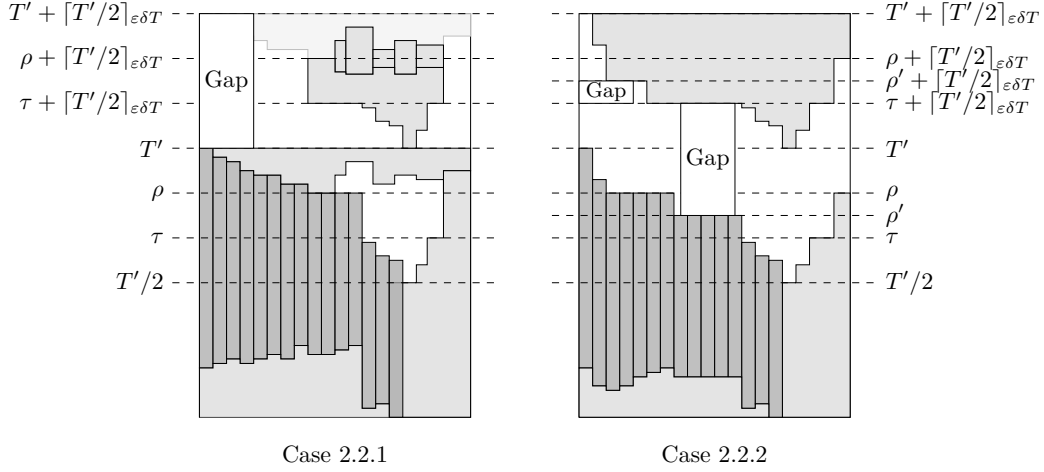
\begin{figure}[ht]
					\centering
					\resizebox{\textwidth}{!}{
					\begin{tikzpicture}
					
					\pgfmathsetmacro{\h}{2}
					\pgfmathsetmacro{\w}{4}
					\pgfmathsetmacro{\t}{4*\h/3}
					\pgfmathsetmacro{\r}{\t + 1.0*\h/3}
					\pgfmathsetmacro{\ryshift}{-2*\h +\r}

					\pgfmathsetmacro{\rw}{2.5*\w/5}
					\pgfmathsetmacro{\rww}{2.7*\w/5}
					\pgfmathsetmacro{\rwww}{3.2*\w/5}
					\pgfmathsetmacro{\rwwww}{3.6*\w/5}
					\pgfmathsetmacro{\rwwwww}{4.0*\w/5}
					\pgfmathsetmacro{\rwwwwww}{4.5*\w/5}
					
					\pgfmathsetmacro{\rh}{\r +0.4*\h/3}
					\pgfmathsetmacro{\rhh}{\r +0.7*\h/3}
					\pgfmathsetmacro{\rhhh}{\r +0.2*\h/3}
					\pgfmathsetmacro{\rhhhh}{\r +0.4*\h/3}
					\pgfmathsetmacro{\rhhhhh}{\r +0.3*\h/3}
					
					\draw(0,0) rectangle (\w,3*\h);
					\node[below] at (\w/2,-2ex) {Case 2.2.1};
					\begin{scope}[yshift = \h cm]
					\draw[lightgray, fill = white!95!black, fill opacity = 0.7]
					(0.05*\w, 2*\h) -- 
					(0.05*\w, 5.8*\h/3) -- 
					(0.10*\w, 5.8*\h/3)  -- 
					(0.10*\w, 5.7*\h/3)  -- 
					(0.15*\w, 5.7*\h/3) -- 
					(0.15*\w, 5.5*\h/3)  -- 
					(0.20*\w, 5.5*\h/3) -- 
					(0.20*\w, 5.4*\h/3) -- 
					(0.25*\w, 5.4*\h/3) -- 
					(0.25*\w, 5.2*\h/3) -- 
					(0.40*\w, 5.2*\h/3) -- 
					(0.40*\w,\r) -- 
					(\rw,\r) -- 
					(\rw,\r) -- 
					(\rw,\rh) -- 
					(\rww,\rh) -- 
					(\rww,\rhh) -- 
					(\rwww,\rhh)  -- 
					(\rwww,\rhhh)  -- 
					(\rwwww,\rhhh) -- 
					(\rwwww,\rhhhh) -- 
					(\rwwwww,\rhhhh) -- 
					(\rwwwww,\rhhhhh) -- 
					(\rwwwwww,\rhhhhh) -- 
					(\rwwwwww,\r) --
					(0.9*\w,\r) --
					(0.9*\w,5.5*\h/3) -- 
					(\w,5.5*\h/3) -- 
					(\w,2*\h) -- 
					(0.1*\w, 2*\h);
					\end{scope}
					
					\draw[fill = white!85!black, fill opacity = 0.7] 
					(0.8*\w,\h)  --
					(0.8*\w,3.4*\h/3) -- 
					(4.2*\w/5,3.4*\h/3) -- 
					(4.2*\w/5,\t) --
					(0.9*\w,\t) -- 
					(0.9*\w,5.5*\h/3)-- 
					(\w,5.5*\h/3) --
					(\w,0) -- 
					(0.70*\w,0) --
					(0.70*\w,\h) --
					(0.8*\w,\h);

					\drawTallItem{0.00*\w}{1.1*\h/3}{0.05*\w}{2*\h};
					\drawTallItem{0.05*\w}{1.2*\h/3}{0.10*\w}{5.8*\h/3};
					\drawTallItem{0.10*\w}{1.3*\h/3}{0.15*\w}{5.7*\h/3};
					\drawTallItem{0.15*\w}{1.2*\h/3}{0.20*\w}{5.5*\h/3};
					\drawTallItem{0.20*\w}{1.3*\h/3}{0.25*\w}{5.4*\h/3};
					\drawTallItem{0.25*\w}{1.4*\h/3}{0.30*\w}{5.4*\h/3};
					\drawTallItem{0.30*\w}{1.3*\h/3}{0.35*\w}{5.2*\h/3};
					\drawTallItem{0.35*\w}{1.6*\h/3}{0.40*\w}{5.2*\h/3};
					\drawTallItem{0.40*\w}{2.4*\h/3+\ryshift}{0.45*\w}{2*\h+\ryshift};
					\drawTallItem{0.45*\w}{2.4*\h/3+\ryshift}{0.50*\w}{2*\h+\ryshift};
					\drawTallItem{0.50*\w}{2.5*\h/3+\ryshift}{0.55*\w}{2*\h+\ryshift};
					\drawTallItem{0.55*\w}{2.6*\h/3+\ryshift}{0.60*\w}{2*\h+\ryshift};
					\drawTallItem{0.65*\w}{\t -0.1*\h/3}{0.60*\w}{0.2*\h/3};
					\drawTallItem{0.70*\w}{\t -0.4*\h/3}{0.65*\w}{0.1*\h};
					\drawTallItem{0.75*\w}{\t -0.5*\h/3}{0.70*\w}{0*\h/3};

					\draw[fill = white!85!black, fill opacity = 0.7]
					(0.00*\w,0.0*\h/3) --
					(0.00*\w,1.1*\h/3) --
					(0.05*\w,1.1*\h/3) --
					(0.05*\w,1.2*\h/3) --
					(0.10*\w,1.2*\h/3) --
					(0.10*\w,1.3*\h/3) --
					(0.15*\w,1.3*\h/3) --
					(0.15*\w,1.2*\h/3) --
					(0.20*\w,1.2*\h/3) --
					(0.20*\w,1.3*\h/3) --
					(0.25*\w,1.3*\h/3) --
					(0.25*\w,1.4*\h/3) --
					(0.30*\w,1.4*\h/3) --
					(0.30*\w,1.3*\h/3) --
					(0.35*\w,1.3*\h/3) --
					(0.35*\w,1.6*\h/3) --
					(0.40*\w,1.6*\h/3) --
					(0.40*\w,2.4*\h/3+\ryshift) --
					(0.45*\w,2.4*\h/3+\ryshift) --
					(0.45*\w,2.4*\h/3+\ryshift) --
					(0.50*\w,2.4*\h/3+\ryshift) --
					(0.50*\w,2.5*\h/3+\ryshift) --
					(0.55*\w,2.5*\h/3+\ryshift) --
					(0.55*\w,2.6*\h/3+\ryshift) --
					(0.60*\w,2.6*\h/3+\ryshift) --
					(0.60*\w,0.2*\h/3) --
					(0.65*\w,0.2*\h/3) --
					(0.65*\w,0.1*\h) --
					(0.70*\w,0.1*\h) --
					(0.70*\w,0);	
					
					\draw[fill = white!85!black, fill opacity = 0.7]
					(0.05*\w, 2*\h) -- 
					(0.05*\w, 5.8*\h/3) -- 
					(0.10*\w, 5.8*\h/3)  -- 
					(0.10*\w, 5.7*\h/3)  -- 
					(0.15*\w, 5.7*\h/3) -- 
					(0.15*\w, 5.5*\h/3)  -- 
					(0.20*\w, 5.5*\h/3) -- 
					(0.20*\w, 5.4*\h/3) -- 
					(0.30*\w, 5.4*\h/3) -- 
					(0.30*\w, 5.2*\h/3) -- 
					(0.40*\w, 5.2*\h/3) -- 
					(0.40*\w,\r) -- 
					(0.50*\w,\r) --
					(\rw,\r) -- 
					(\rw,\rh) -- 
					(\rww,\rh) -- 
					(\rww,\rhh) -- 
					(\rwww,\rhh) -- 
					(\rwww,\rhhh) -- 
					(\rwwww,\rhhh) -- 
					(\rwwww,\rhhhh) -- 
					(\rwwwww,\rhhhh) -- 
					(\rwwwww,\rhhhhh) -- 
					(\rwwwwww,\rhhhhh) -- 
					(\rwwwwww,\r) -- 
					(0.9*\w,\r) --
					(0.9*\w,5.5*\h/3) -- 
					(\w,5.5*\h/3) -- 
					(\w,2*\h) -- 
					(0.05*\w, 2*\h);
					
					\drawVerticalItem{\rw}{\r +\h-0.1*\h}{\rww}{\rh + \h};
					\drawVerticalItem{\rww}{\r +\h-0.35*\h/3}{\rwww}{\rhh + \h};
					\drawVerticalItem{\rwww}{\r +\h-0.2*\h/3}{\rwwww}{\rhhh + \h};
					\drawVerticalItem{\rwwww}{\r +\h-0.35*\h/3}{\rwwwww}{\rhhhh + \h};
					\drawVerticalItem{\rwwwww}{\r +\h-0.2*\h/3}{\rwwwwww}{\rhhhhh + \h};
					
					\draw[fill = white!85!black, fill opacity = 0.7]
					(0.60*\w,\t -0.1*\h/3 + \h) --
					(0.65*\w,\t -0.1*\h/3 + \h) --
					(0.65*\w,\t -0.4*\h/3 + \h) --
					(0.70*\w,\t -0.4*\h/3 + \h) --
					(0.70*\w,\t -0.5*\h/3 + \h) --
					(0.75*\w,\t -0.5*\h/3 + \h) --
					(0.75*\w,\h + \h) --
					(0.8*\w,\h + \h) --
					(0.8*\w,3.4*\h/3 + \h) --
					(4.2*\w/5,3.4*\h/3 + \h) --
					(4.2*\w/5,\t + \h) --
					(0.9*\w,\t + \h) --
					(0.9*\w,\r +2.8*\h/3) --
					(\rwwwww,\r +2.8*\h/3) --
					(\rwwwww,\r +2.65*\h/3) --
					(\rwwww,\r +2.65*\h/3) --
					(\rwwww,\r +2.8*\h/3) --
					(\rwww,\r +2.8*\h/3) --
					(\rwww,\r +2.65*\h/3) --
					(\rww,\r +2.65*\h/3) --
					(\rww,\r +0.9*\h) --
					(\rw,\r +0.9*\h) --
					(\rw,\r+\h) --
					(0.40*\w,\r+\h) --
					(0.40*\w,\t+\h) --
					(0.60*\w,\t+\h) --
					(0.60*\w,\t -0.1*\h/3 + \h);

					\draw[dashed] (-0.1*\w,\h) node[left] {$T'/2$} -- (1.1*\w,\h);
					\draw[dashed] (-0.1*\w,2*\h) node[left] {$T'$} -- (1.1*\w,2*\h);
					\draw[dashed] (-0.1*\w,3*\h) node[left] {$T' +\addT$} -- (1.1*\w,3*\h);
					
					\draw[dashed] (-0.1*\w,\t) node[left] {$\tau$} -- (1.1*\w,\t);
					\draw[dashed] (-0.1*\w,\t+\h) node[left] {$\tau +\addT$} -- (1.1*\w,\t+\h);
					
					\draw[dashed] (-0.1*\w,\r) node[left] {$\rho$} -- (1.1*\w,\r);
					\draw[dashed] (-0.1*\w,\r+\h) node[left] {$\rho +\addT$} -- (1.1*\w,\r+\h);
					
					\draw[fill=white] (0.0,2*\h) rectangle (\w/5,3*\h) node[midway] {Gap};
					
					\begin{scope}[xshift = 1.4*\w cm]
						
					\pgfmathsetmacro{\h}{2}
					\pgfmathsetmacro{\w}{4}
					\pgfmathsetmacro{\t}{4*\h/3}
					\pgfmathsetmacro{\r}{\t + 1.0*\h/3}
					\pgfmathsetmacro{\ryshift}{-2*\h +\r}
					\pgfmathsetmacro{\g}{(\t +\r)/2}
					
					\draw(0,0) rectangle (\w,3*\h);
					\node[below] at (\w/2,-2ex) {Case 2.2.2};
					\draw[fill = white!85!black, fill opacity = 0.7] 
					(0.8*\w,\h) --
					(0.8*\w,3.4*\h/3) -- 
					(0.85*\w,3.4*\h/3) -- 
					(0.85*\w,\t) --
					(0.94*\w,\t) -- 
					(0.94*\w,\r) -- 
					(\w,\r) --
					(\w,0) --
					(0.75*\w,0) --
					(0.75*\w,\h) --
					(0.8*\w,\h);
					
					\drawTallItem{0.00*\w}{1.1*\h/3}{0.05*\w}{2*\h};
					\drawTallItem{0.05*\w}{0.7*\h/3}{0.10*\w}{5.3*\h/3};
					\drawTallItem{0.10*\w}{1.6*\h/3+\ryshift}{0.15*\w}{2*\h+\ryshift};
					\drawTallItem{0.15*\w}{1.7*\h/3+\ryshift}{0.20*\w}{2*\h+\ryshift};
					\drawTallItem{0.20*\w}{1.9*\h/3+\ryshift}{0.25*\w}{2*\h+\ryshift};
					\drawTallItem{0.25*\w}{2.0*\h/3+\ryshift}{0.30*\w}{2*\h+\ryshift};
					\drawTallItem{0.30*\w}{2.1*\h/3+\ryshift}{0.35*\w}{2*\h+\ryshift};
					\drawTallItem{0.35*\w}{\g -3.6*\h/3}{0.40*\w}{\g};
					\drawTallItem{0.40*\w}{\g -3.6*\h/3}{0.45*\w}{\g};
					\drawTallItem{0.45*\w}{\g -3.6*\h/3}{0.50*\w}{\g};
					\drawTallItem{0.50*\w}{\g -3.6*\h/3}{0.55*\w}{\g};
					\drawTallItem{0.55*\w}{\g -3.6*\h/3}{0.60*\w}{\g};
					\drawTallItem{0.60*\w}{\t -0.1*\h/3}{0.65*\w}{0.2*\h/3};
					\drawTallItem{0.65*\w}{\t -0.4*\h/3}{0.70*\w}{0.3*\h/3};
					\drawTallItem{0.70*\w}{\t -0.5*\h/3}{0.75*\w}{0};
					
					\draw[fill = white!85!black, fill opacity = 0.7] 
					(0.75*\w,2*\h) --
					(0.8*\w,2*\h) --
					(0.8*\w,3.4*\h/3+\h) -- 
					(0.85*\w,3.4*\h/3+\h) -- 
					(0.85*\w,\t+\h) --
					(0.94*\w,\t+\h) -- 
					(0.94*\w,\r+\h) -- 
					(\w,\r+\h) --
					(\w,3*\h) --
					(0.05*\w,3*\h) --
					(0.05*\w,5.3*\h/3+\h) --
					(0.10*\w,5.3*\h/3+\h) --
					(0.10*\w,\g+\h) --
					(0.25*\w,\g+\h) --
					(0.25*\w,\t+\h) --
					(0.60*\w,\t+\h) --
					(0.60*\w,\t -0.1*\h/3+\h) --
					(0.65*\w,\t -0.1*\h/3+\h) --
					(0.65*\w,\t -0.4*\h/3+\h) --
					(0.70*\w,\t -0.4*\h/3+\h) --
					(0.70*\w,\t -0.5*\h/3+\h) --
					(0.75*\w,\t -0.5*\h/3+\h) --
					(0.75*\w,\h+\h);
					
					\draw[fill = white!85!black, fill opacity = 0.7] 
					(0,0) --
					(0,1.1*\h/3) --
					(0.05*\w,1.1*\h/3) --
					(0.05*\w,0.7*\h/3) --
					(0.10*\w,0.7*\h/3) --
					(0.10*\w,1.6*\h/3+\ryshift) --
					(0.15*\w,1.6*\h/3+\ryshift) --
					(0.15*\w,1.7*\h/3+\ryshift) --
					(0.20*\w,1.7*\h/3+\ryshift) --
					(0.20*\w,1.9*\h/3+\ryshift) --
					(0.25*\w,1.9*\h/3+\ryshift) --
					(0.25*\w,2.0*\h/3+\ryshift) --
					(0.30*\w,2.0*\h/3+\ryshift) --
					(0.30*\w,2.1*\h/3+\ryshift) --
					(0.35*\w,2.1*\h/3+\ryshift) --
					(0.35*\w,\g -3.6*\h/3) --
					(0.60*\w,\g -3.6*\h/3) --
					(0.60*\w,0.2*\h/3) --
					(0.65*\w,0.2*\h/3) --
					(0.65*\w,0.3*\h/3) --
					(0.70*\w,0.3*\h/3) --
					(0.70*\w, 0) --
					(0,0);

					\draw[dashed] (-0.1*\w,\h) -- (1.1*\w,\h) node[right] {$T'/2$};
					\draw[dashed] (-0.1*\w,2*\h) -- (1.1*\w ,2*\h) node[right] {$T'$};
					\draw[dashed] (-0.1*\w,3*\h) -- (1.1*\w,3*\h) node[right] {$T' + \addT$};
					
					\draw[dashed] (-0.1*\w,\t) -- (1.1*\w,\t) node[right] {$\tau$};
					\draw[dashed] (-0.1*\w,\t+\h) -- (1.1*\w,\t+\h) node[right] {$\tau + \addT$};
					
					\draw[dashed] (-0.1*\w,\r) -- (1.1*\w,\r) node[right] {$\rho$};
					\draw[dashed] (-0.1*\w,\r+\h) -- (1.1*\w,\r+\h) node[right] {$\rho + \addT$};
					
					\draw[dashed] (-0.1*\w,\g) -- (1.1*\w,\g) node[right] {$\rho'$};
					\draw[dashed] (-0.1*\w,\g+\h) -- (1.1*\w,\g+\h) node[right] {$\rho' + \addT$};

					\draw[fill=white] (0.75*\w/2,\g) rectangle (1.15*\w/2,\h + \t)node[midway] {Gap};
					\draw[fill=white] (0.0*\w/2,\h+\t) rectangle (0.4*\w/2,\h + \g)node[midway] {\small{Gap}};
					\end{scope}
					\end{tikzpicture}  
				}
					\caption{The shifted schedule and the position of the gap in Cases 2.2.1 and 2.2.2. 
					Note that in Case 2.2.2 the gap is not displayed continuously. However, by swapping the used resource, we can make it continuous. 
					We only need the fact that at each point in time there are enough free resources and machines.
				}
					\label{fig:SRCS:ShiftCase221}
				\end{figure}
				
				\item[$\jR{\jobsT{\rho}} < \gamma R$.]
				
				Let $\rho' \in \{i\delta^2| \tau/\delta^2 \leq i \leq \rho/\delta^2 , i \in \mathbb{N}\}$ be the smallest value, where $\jR{\jobsT{\rho'}} \leq \gamma R$. 
				Remember, we had $\jR{\jobs_H \cap \jobsT{\tau}} \geq R-3 \gamma R$ so huge jobs with summed resource requirement of at least $R-4 \gamma R$ are finished till $\rho'$.
				We partition the huge jobs that finish between $\tau$ and $\rho$ by their processing time. 
				Since each job has a processing time of at least $\addT$, we get at most $\Oh(1/\eps\delta) \leq |\startPoints|/2$ sets. 
				As seen in Section \ref{sec:SRCS:APTAS:largeJobs}, we have to discard at most $k \leq 3|\startPoints|$ large jobs, which have to be placed later on. 
				
				\begin{claim}
					There exists a set in the partition, which uses at least $3 \gamma R$ resource total.
				\end{claim}
				\begin{proofClaim}
				Since $\gamma \leq 1/(2|\startPoints|) \leq 1/(3|\startPoints|/2 +4)$ it holds that
				\[\frac{R-4 \gamma R}{|\startPoints|/2} \geq \frac{(1-4/(3|\startPoints|/2+4))R}{|\startPoints|/2} = 3 R/(3|\startPoints|/2 +4) \geq 3\gamma R.\]
				Therefore, by the pigeon principle, there must be one set in the partition, which has summed resource requirement of at least $3 \gamma R$. 
				\end{proofClaim}
			
				We sort the jobs in this partition by non increasing order of resource requirement. 
				We greedily take jobs from this set, till they have a summed resource requirement of at least $\gamma R$ and schedule them such that they end before $\rho'$. 
				If there was a job with more than $\gamma R$ resource requirement, it had to be finished before $\rho'$ since the resource requirement of huge jobs finishing after $\rho'$ is smaller than $\gamma R$ and we only chose it. 
				Otherwise, the greedily chosen jobs have summed resource requirement of at $2 \gamma R$. 
				
				Since the considered set has a summed resource requirement of at least $3 \gamma R$, jobs of this set with summed resource requirement at least $2 \gamma R$ end before $\rho'$. 
				Therefore, we do not violate any constraint by shifting down these jobs such that they end at $\rho'$, see Figure \ref{fig:SRCS:ShiftCase221}.
	
			\end{caseList}
		\end{caseList}
	\end{caseList}
	
	Concerning property three, note that since we only use the free area (machines and resources) to schedule the $k$ large jobs inside the gap, there is a layer $s'$ in the shifted schedule, for each layer $s \in \startPoints$ that has at least as many machines and resources not used by large and huge jobs as the layer $s$.
\end{proof}	

\subsection{Algorithm}
\label{sec:SRCS:Absolute:Summary}
The algorithm works similar to the algorithm in Section \ref{sec:SRCS:APTAS:Summary}.
The only step that differs is Step 4.
In the following, we describe the altered step. 

Given a value $T' := i \eps' T$, we determine the set $\startPoints$ and call the algorithm from Lemma \ref{lma:schedulingLargeJobs} with $\gamma = 1/(3|\startPoints|+4)$ to generate the set of schedules for the large jobs.
One of these schedules uses in each layer at most as many machines and resources for large jobs, as the rounded optimal schedule, or the value $T'$ is to small.
Furthermore, the set of not scheduled large jobs $\jobs'$ has a total machine requirement of at most $3|\startPoints|$, a total resource requirement of at most $\gamma R$, and each job has a processing time of at most $T'/2$. 

For each of these schedules, the algorithm iterates all values for $\tau$ and $\rho$ and all possibilities for the at most $2k$ huge jobs ending before or after these values and their starting positions. 
Then, we identify the case and the other variables dependent on the guesses and the solution schedule for the large jobs.
By this we generate a new set of schedules, for which it will try to place the small jobs.

In detail, the algorithm performs the following steps:
The set $S_{H,L}$ of schedules for huge and large jobs is initialized as the empty set.
For each of the schedules for large jobs, it iterates all the possible values for $\tau$ as well as all choices for the at most $k-1$ huge jobs ending before $ \tau$ and their starting positions ($\Oh((m/\eps^2)^{k})$ possibilities).
For each of these guesses, the algorithm identifies the applying case:
\begin{enumerate}[align=left]
	\item[Case 1.1] It schedules the removed jobs at $\tau$, shift the jobs as described, and adds the schedule to $S_{H,L}$.
	\item[Case 1.2] It finds $i$ and for each possible value of $\rho$ it identifies the resulting case:
	\begin{enumerate}[align=left]
		\item[Case 1.2.1.] It shifts the jobs as described, schedules the set $\jobs'$ at $\iota$, and adds the schedule to $S_{H,L}$.
		\item[Case 1.2.2.] For each possible choice of at most $k$ huge jobs ending after $\rho$ and their starting positions, it tries to shift the jobs as described, schedule the removed jobs at $T'$, and, if successful, adds the schedule to $S_{H,L}$.
	\end{enumerate}
	\item[Case 2.] It iterates all possible values for $\rho$ and all choices for the at most $k$ huge jobs ending after $\rho$ as well as their starting positions. For each of these guesses, it identifies the corresponding case:
	\begin{enumerate}[align=left]
		\item[Case 2.1.] It searches $\tau'$, schedules the jobs as described in Case 2.1.1. or Case 2.1.2. depending on the relation of $\tau'$  and $\rho$, and adds the schedule to $S_{H,L}$.
		\item[Case 2.2.1.] It schedules the jobs as described and adds the schedule to $S_{H,L}$.
		\item[Case 2.2.2.] It iterates all possible values for $\rho'$, tries to schedule the jobs as described, and, if successful, adds the schedule to $S_{H,L}$.
	\end{enumerate}		
\end{enumerate}
Remember that one, say $S_L$, of the generated schedules for large jobs has a profile of machines and resources used by large jobs that is dominated by the profile of machines and resources used by large jobs in a rounded optimal solution, i.e., the considered rounded optimal solution uses in each layer at least as many resources and machines by large jobs as the schedule $S_L$.
For this schedule $S_L$, the algorithm iterates all the possible cases to schedule the huge jobs as described in the proof of Lemma \ref{lma:hugeJobs}.
Hence there exists an injection from the layers of the rounded optimal solution to the layers of one of the generated schedules for huge and large jobs, such that in each hit layer the total number of machines and resources used by huge and large jobs is smaller than the one of the mapped layer.
Therefore by Lemma \ref{lma:schedulingSmallJobs}, if $T'$ is at least as large as the makespan of a rounded optimal schedule, we can find a schedule of the small jobs inside the $(1+\eps)$-scaled layers and an extra box with processing time $\Oh(\eps)T$ for at least one of the generated schedules of the large and huge jobs.

For each of the resulting schedules $S_{H,L}$ of huge and large jobs, the algorithm computes the residual machines and resources per layer and calls the algorithm from Lemma \ref{lma:schedulingSmallJobs}.
If the algorithm can place all the small jobs, the schedule is saved and the next smaller value for $T'$ is considered.
If for non of the generated schedules $S_{H,L}$ a schedule of the small jobs is found by the algorithm, it considers the next larger value for $T'$.

To bound the total number of guesses that we add by this procedure, note that we have to guess $\tau$, $\rho$ and $\rho'$ from at most $\Oh(|\startPoints|)$ possibilities. 
Further, we for each of these guesses, the algorithm guesses at most $2k$ huge jobs and their starting positions. 
The total number of these guesses is bounded by $(m/\eps^2)^{\Oh(k)}$, since the huge jobs start at multiples of $\eps^2 T$.
Therefore, the total number of guesses for the large jobs is bounded by $(m/\eps^2)^{\Oh(k)} \cdot \Oh(|\startPoints|^3)$.
Since $k \leq 3 |\startPoints|$, this guess for the huge jobs lengthens the running time by a factor of at most $(m/\eps)^{1/\eps^{\Oh(1/\eps^2)}}$.
This concludes the proof of Theorem \ref{thm:SRCS:absoluteApprox}.

\section{Proof of Lemma \ref{lma:schedulingSmallJobs} -- an algorithm to schedule the small jobs}
\label{sec:smallJobs}

For the sake of completeness, this section describes the algorithm to schedule the small jobs.
\schedulingSmallJobs*

The first step is to round the resource requirements of the small jobs.
Afterward, we schedule them fractionally with the objective to generate an integral schedule in a later step. 
When scheduling the jobs fractionally, we allow to stop the processing of a job $j$. 
Immediately the resources and the machine used by $j$ are deallocated. 
The rest of $j$ can be processed at any time, requiring again $r(j)$ resource and one machine. 
Especially we allow two or more parts of the job to be processed at the same time. 
A job is (fully) scheduled if all the processing times of its (interrupted) parts add up to $p(j)$. 
A fractional schedule is feasible, if all jobs are scheduled and neither the resource or machine condition is violated.

\paragraph*{Rounding the Resource Requirements}
To reduce the running time, we round the resource requirements of the small jobs using a standard technique called geometric grouping.
The idea of rounding item sizes by grouping was first introduced by  which was first introduced by Fernandez de la Vega and
Lueker~\cite{VegaL81} and refined for strip packing similar to the here presented form by Bougeret et al.~\cite{BougeretDJOT09} and Sviridenko~\cite{Sviridenko12}.

\begin{lemma}
	\label{lma:roundingOfSmallJobs}
	For every $\eps > 0$ with $1/\eps \in \mathbb{N}$ and every instance $I$ with set of small wide jobs $\jobs_{SW,\eps}$,
	it is possible to round the resource requirements of the small jobs to at most $\mathcal{O}(\log(m)/\eps)$ different sizes in at most $\Oh(n\log(1/\eps))$ operations such that 
	\begin{itemize}
		\item the number of distinct resource requirements larger than $\eps R$ is bounded by $\Oh(\log(1/\eps)/\eps)$,
		\item given any schedule for all jobs in $I$, we can integrate the rounded small jobs fractionally into the schedule, by replacing the original small jobs with fractions of the rounded small jobs and adding an extra schedule with a makespan bounded by  $2\eps\cdot \areaI{\jobs(I)}/R + \eps \pT{\jobs(I)}/m \leq 3\eps\OPTpre(I,\eps)$.
	\end{itemize} 
\end{lemma}
\begin{proof}
	First, we partition the set of small jobs $\jobs_S$ into $\lceil \log(m) \rceil +1$ sets.
	We construct the first $\lceil \log(m) \rceil$ sets such that
	the $i$th set $\jobs_{S,i}$ contains the jobs with resource requirements in the interval $(R/2^i,R/2^{i-1}]$. 
	The last of these sets contains jobs with resource requirements in the interval $(R/2^{\lceil \log(m) \rceil},R/2^{\lceil \log(m) \rceil -1}]$.
	The residual jobs build the last set $\jobs_{S,\bot}$. 
	All the jobs in this set have a resource requirement of at most $R/2^{\lceil \log(m) \rceil} \leq R/m$, and hence, we can schedule $m$ of them at the same time without violating any constraint.
	
	To round the jobs, we consider all the sets $\jobs_{S,i}$, $i \in \{1,\dots,\lceil\log(m)\rceil, \bot \}$. 
	The jobs in a set $\jobs_{S,i}$ are sorted in increasing order of their resource requirement, and are stacked such that the job with the smallest resource requirement is on the bottom, see Figure \ref{fig:SRCS:RoundingOfResourceRequirements}.
	The processing time of the stack is given by $\pT{\jobs_{i}} := \sum_{j \in \jobs_{i}} \pT{j}$.
	We will partition this stack into $1/\eps$ segments.
	We draw a horizontal line at each integral multiple of $\eps \pT{\jobs_{i}}$, starting at $0$. 
	If there is no job intersected by this line, we leave the horizontal line at its position.
	Otherwise, we shift this line up such that it lies at the end of the cut job. 
	We define $\jobs_{i,\mathrm{split}}$ as the set of jobs that where previously intersected by these lines, and call the set of jobs which is contained between the $l$th and $(l+1)$st horizontal line group $l$ and denote it with $\jobs_{i,l}$.
	We define the rounded resource requirement of the jobs in $\jobs_{i,l}$ as $r_{i,l} := \max_{j \in \jobs_{i,l}} \jR{j}$.
	Since each of the $\Oh(\log(m))$ stacks is divided into $\Oh(1/\eps)$ groups, the total number of rounded resource requirements is bounded by $\mathcal{O}(\log(m)/\eps)$.

	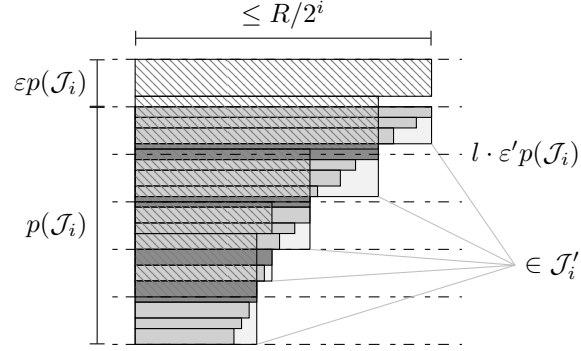
\begin{figure}
		\centering
		\begin{tikzpicture}
		\pgfmathsetmacro{\h}{-0.7}
		\pgfmathsetmacro{\w}{1}

		\foreach \i\x\y\xx\yy in {
			0.1/0/0.0/3.9/0.7,
			0.1/0/0.7/3.2/1.7,
			0.1/0/1.7/2.3/2.7,
			0.1/0/2.7/1.8/3.3,
			0.1/0/3.3/1.6/4.5,
			0.3/0/0.0/3.9/0.2,
			0.3/0/0.2/3.7/0.4,
			0.3/0/0.4/3.4/0.7,
			0.9/0/0.7/3.2/1.0,
			0.3/0/1.0/2.9/1.2,
			0.3/0/1.2/2.7/1.5,
			0.3/0/1.5/2.4/1.7,
			0.9/0/1.7/2.3/1.9,
			0.3/0/1.9/2.3/2.2,
			0.3/0/2.2/2.1/2.4,
			0.3/0/2.4/1.9/2.7,
			0.9/0/2.7/1.8/3.0,
			0.3/0/3.0/1.7/3.3,
			0.9/0/3.3/1.6/3.7,
			0.3/0/3.7/1.5/4.0,
			0.3/0/4.0/1.4/4.2,
			0.3/0/4.2/1.3/4.5
			%
		}{
			\draw[fill=gray,fill opacity=\i] (\x*\w,\h*\y) rectangle (\xx*\w,\yy*\h);
		}
		
		\foreach \i\x\y\xx\yy in {
			0.4/0/-0.9/3.9/-0.2,
			0.4/0/-0.2/3.2/0.8,
			0.4/0/0.8/2.3/1.8,
			0.4/0/1.8/1.8/2.4,
			0.4/0/2.4/1.6/3.6
		}{
			\draw[pattern = north west lines, pattern color = gray] (\x*\w,\h*\y) rectangle (\xx*\w,\yy*\h);
		}
		
		\foreach \i in {-1,0,2,3,4,5}
		{ 
			\draw  [dash pattern=on 2pt off 3pt on 4pt off 4pt] (-0.3*\w,\i*0.9*\h) -- 	(4.3*\w,0.9*\i*\h);	
		}
		\foreach \i/\x in {0.7/3.9,1.7/3.2,2.7/2.3,3.3/1.8,4.5/1.6}
		{ 
			\draw [lightgray] (\x*\w,\i*\h) -- (5*\w,3*\h);	
		}	
		\node[right] at (5*\w,3*\h) {$ \in \jobs_{i}'$};
		
		\draw  [dash pattern=on 2pt off 3pt on 4pt off 4pt] (-0.3*\w,0.9*\h) -- (4.3*\w,0.9*\h) node[right] {$l \cdot \epss \pT{\jobs_{i}}$};
		
		\draw[|-|] (-0.5*\w,0*\h) -- node[midway, left]{$\pT{\jobs_{i}}$} (-0.5*\w, 4.5*\h);
		\draw[|-|] (-0.5*\w,-0.9*\h) -- node[midway, left]{$\eps\pT{\jobs_{i}}$} (-0.5*\w, 0*\h);
		\draw[|-|] (-0.0*\w,-1.3*\h) -- node[midway, above]{$\leq R/2^{i}$} (3.9*\w, -1.3*\h);
		\end{tikzpicture} 
		\caption{Rounding of the jobs in $\jobs_i$. The hatched rectangles represent the rounded jobs that are shifted upwards by $\epss \pT{\jobs_{i}}$.}
		\label{fig:SRCS:RoundingOfResourceRequirements}
	\end{figure}
	
	\begin{claim}
		\label{lma:Segments2}
		Given any schedule for the jobs $\jobs$, we can find a schedule for the rounded small jobs by replacing the small jobs with fractions of the rounded jobs and adding an extra schedule with makespan bounded by $2\eps\cdot \areaI{\jobs(I)}/R + \eps\pT{\jobs(I)}/m \leq 3\eps\OPTpre(I,\eps)$. 
	\end{claim}
	\begin{proofClaim}
		Consider the stack of the jobs in $\jobs_{S,i}$ and the stack of rounded jobs for the jobs in $\jobs_{S,i}$. 
		Consider a horizontal line at $\tau \geq \eps \pT{\jobs_{S,i}}$ through the first stack and the corresponding horizontal line at $\tau - \eps \pT{\jobs_{S,i}}$ trough the second stack. 
		
		We claim that the resource requirement of the job $j_\tau$ intersected by $\tau$ is at least as large as the resource requirement of the rounded job $j_{\tau}'$ at $\tau - \eps\pT{\jobs_{S,i}}$.
		
		Let $j_\tau \in \jobs_{i,l}$ for some $l \in \{0,\dots,1/\eps-1\}$.
		There are two options for $j_{\tau}'$.
		Either it is in $\jobs_{i,l}$ as well or it is in $\jobs_{S,i',l}$, with $i' < i$.
		In the second case the claim follows trivially since, by construction, all the jobs in $\jobs_{S,i,l}$ have a resource requirement that is at least as large as the rounded resource requirements of the jobs in $\jobs_{S,i',l}$.
		
		On the other hand, if $j_{\tau}'$ is in $\jobs_{S,i,l}$, the job $j_\tau$ has to be a job in $\jobs_{S,i,\mathrm{split}}$ because the sum of processing times of jobs in $\jobs_{S,i,l}$ that are not in $\jobs_{S,i,\mathrm{split}}$ is bounded by $\eps\pT{\jobs_{S,i}}$.
		Since the jobs in $\jobs_{S,i,\mathrm{split}}$ are the ones with the largest resource requirement of their group the resource requirement of $j_{\tau}$ has to be at least as large as $j_{\tau}'$.

		Hence, when shifting the second stack such that it starts at $\eps \pT{\jobs_{S,i}}$, we schedule the jobs in $\jobs_{S,i}'$ instead of the jobs that are positioned at the same processing time in the stack. 
		What remains to be scheduled is a last part of processing time $\eps \pT{\jobs_{S,i}}$ that cannot be scheduled instead of any job due to the shifting.
		These jobs will be scheduled at the end of the schedule in the next step.
		We call these parts to be scheduled on the top of the schedule segment $S_i$ for the set $\jobs_{S,i}$.
		
		First consider segment $S_\bot$.
		Note that $\pT{\jobs_{S,\bot}} \leq m \OPTpre(I,\eps)$ because at most $m$ jobs can be scheduled in parallel. 
		Hence, the job parts in $S_\bot$ have a total processing time of $\eps \pT{\jobs_{S,\bot}} \leq \eps m\OPTpre(I,\eps)$.
		We partition this segment into sub segments of processing time $\eps \pT{\jobs_{S,\bot}}/m$ (slicing some jobs horizontally in the progress) and schedule each on one machine.
		This schedule does not violate the resource condition because $\jR{j} \leq R/m$ for each $j \in \jobs_{S,\bot}$.
		Hence, it is possible to schedule the jobs in this segment fractionally with  a makespan of at most $\eps \pT{\jobs_{S,\bot}}/m \leq \eps \pT{\jobs(I)}/m \leq \eps\OPTpre(I,\eps)$.
		
		Next, we consider the residual segments $S_i$.
		The jobs in $\jobs_{i}$ have a resource requirement of at  most $R/2^{i-1}$.
		Hence, we can schedule $2^{i-1} \leq m$ jobs from the segment $S_i$ in parallel without violating the resource constraint. 
		Therefore, we need a processing time of $\eps \pT{\jobs_{S,i}}/2^{i-1}$ to schedule the jobs in this segment fractionally.  
		On the other hand, each job in the set $\jobs_{i}$ has a resource requirement of at least $R/2^{i}$. 
		Therefore, it holds that $\sum_{i = 1}^{\lceil \log(m) \rceil +1} \pT{\jobs_{S,i}}/2^{i}\leq  \areaI{\jobs(I)}/R \leq \OPTpre(I,\eps)$.
		
		Hence, the total processing time added to schedule all the jobs in the segments $S_i$, for $i \in \{1,\dots,\lceil\log(m)\rceil, \bot \}$, is bounded by 
		\begin{align*}
		&\sum_{i = 1}^{\lceil \log(m) \rceil +1} \eps \pT{\jobs_{S,i}}/2^{i-1} + \eps \pT{\jobs_{S,\bot}}/m\\ 
		&= 2\eps \sum_{i = 1}^{\lceil \log(m) \rceil +1} \pT{\jobs_{S,i}}/2^{i} + \eps \pT{\jobs_{S,\bot}}/m\\
		&\leq  2\eps\cdot \areaI{\jobs(I)}/R + \eps \pT{\jobs(I)}/m \\
		&\leq  3\eps \OPTpre(I,\eps)\mathrm{,}
		\end{align*}
		which concludes the proof of the claim.
	\end{proofClaim}
	
	To find the partition into the groups $\jobs_{S,i,j}$ we only need to find all the size defining jobs for the groups.
	This can be done in $\Oh(n\log(1/\eps))$ using a modified median algorithm.
	We first partition the jobs into the $\log(m)$ sets and, in each of these sets, we find the $1/\eps$ size defining jobs.
	The partition can be done in $\Oh(n)$, provided, we can calculate the logarithm of $\lfloor R/\jR{j} \rfloor \leq m$ in $\Oh(1)$, while the search for the size defining jobs can be done in $\Oh(n\log(1/\eps))$ for all of the sets. 
	
	Note that the assumption that we can find $\log(\lfloor R/\jR{j} \rfloor)$ in $\Oh(1)$ is reasonable for all $\jR{j} \geq R/m$ since, in $O(m)$, we can provide a table of size $m$ which contains as entries the size of the corresponding logarithm.
	This table needs to be constructed only once and afterward the logarithm can be found in $\Oh(1)$ by a simple lookup.
	
	Now consider the number of resource requirements larger than $\eps R$. 
	These jobs are partitioned into at most $\Oh(\log(1/\eps))$ sets and, for each of these sets, we generated at most $1/\eps$ different sizes. 
	Hence the total number of different resource requirements larger than $\eps R$ is bounded by $\Oh(\log(1/\eps)/\eps)$.
\end{proof}

We denote by $\jobs_S'$ set that contains for each rounded resource requirement exactly one job that has the summed processing time of all original jobs mapped to this resource requirement.
	
\paragraph*{Solving a Configuration LP}
We will schedule the jobs in $\jobs_S'$ using a configuration LP that allows for each layer a certain set of configurations. 
In the following, we will describe the configurations with more detail.

A configuration of jobs in $\jobs_S'$ for a layer $s$ is a multiset $C := \sset{ a_j:j}{j \in \jobs_S'}$ such that $\jR{C}  := \sum_{j \in \jobs_S}a_j\jR{j} \leq \REntryS{s}$ and $m(C) := \sum_{j \in \jobs} a_j\leq \mEntryS{s}$; i.e., a configuration defines a multiset of jobs in $\jobs_S'$ which can be scheduled at the same time, without violating the resource or machine constraint. 
We define $\conf_s$ as the set of configurations for layer $s$.
We introduce a new layer $\LLS$ with processing time $3\eps\greedyMakespan$, to cover the extra makespan of at most $3\eps\greedyMakespan$ needed due to the rounding of the resource requirements.
We define $\startPoints_{\LLS} := \startPoints \cup \{\LLS\}$ as the set of all considered layers and $\mathcal{C}_s$ as the set of configurations for layer $s \in \startPoints_{\LLS}$.
Note that $\mathcal{C}_{\LLS}$ contains all configurations $C$ with $\jR{C} \leq R$ and $m(C) \leq m$ as this layer contains no large job.  

Consider the following linear program $LP_{\startPoints, \jobs_{S}}$:
\begin{align}
\sum_{c \in \mathcal{C}_{\LLS}} x_{\LLS,C} &= 3\eps \greedyMakespan \label{eq:APTAS:SRCS:smallWide}\\
\sum_{c \in \mathcal{C}_s} x_{s,C} &= \eps\delta\greedyMakespan &\forall s\in \startPoints\label{eq:APTAS:SRCS:smallLayer}\\
\sum_{s \in \startPoints_{\LLS}}\sum_{C \in \mathcal{C}_s} C_jx_{s,C}&= p(j) & \forall j \in {\jobs}_{S}'\label{eq:APTAS:SRCS:jobs}\\
x_{s,C} &\geq 0 &\forall s\in S_{\LLS}, C \in \mathcal{C}_s \mathrm{.}\\
\end{align}
The variable $x_{s,C}$ denotes the total processing time of the configuration $C \in \mathcal{C}_s$ in layer $s$. 
The first two conditions ensure that the total processing time of a layer is not exceeded by the total processing time of the configurations in this layer.
The third condition ensures that each job is scheduled completely.
Note that a rounded optimal solution can be transformed into a solution of this linear program, by determining the stets of rounded small jobs that are processed in parallel and setting the corresponding variable to the total time in which these jobs are processed in parallel.
Since this linear program has up to $(\log(m)/\eps)^{m}\cdot 1/\eps\delta$ variables, we can not solve it directly. 
Instead, we search for a solution to the relaxed version of this linear program, where we enlarge the right hand side of the equations (\ref{eq:APTAS:SRCS:smallWide}) and (\ref{eq:APTAS:SRCS:smallLayer}) by the factor $(1+\eps)$ and hence stretch the resulting schedule by a factor of $(1+\eps)$.

\begin{lemma}
	\label{lma:SRCS:APTAS:FirstLPSolutionSmallJobs}
	If there is a solution to the linear program $LP_{\startPoints, \jobs_{S}}$, we can find a solution $x$ to the relaxed version of the linear program in at most $\log(m)^{\Oh(1)} \cdot  (1/\eps)^{\Oh(1/\eps)}$ that uses at most $\Oh(|\startPoints| + |\jobs_S'|)$ non-zero components.
\end{lemma}
\begin{proof}
	This linear program has $|\startPoints| +1 +|{\jobs}_{S}'| $ constraints and at most $(\log(m)/\eps)^{m}\cdot 1/\eps\delta$ variables. 
	How to solve a similar problem (that does not contain constraint (\ref{eq:APTAS:SRCS:smallWide}) or constraint (\ref{eq:APTAS:SRCS:smallLayer})) has been described to construct an AFPTAS for this problem by Jansen et al.~\cite{journals/talg/JansenMR19}.
	We will not repeat the lengthily description here and instead give a high level overview.
	The main idea is to transform the above LP to a max-min-resource-sharing problem.
	In this problem, we are given a non-empty convex compact set $P$ as well set of $M$ non-negative continuous concave functions $f_i : B \rightarrow \mathbb{R}$, $i \in [M]$.
	The problem asks to find $\lambda^* := \max\{\lambda \in \mathbb{R} | \exists  x \in B: f_i(x) \geq \lambda \forall i \in [M]\}$. 	
	Note that $B$ can be of the form $B_1 \times \dots \times B_K$, where each of the sets $B_j$ is a non-empty convex compact set. 
	These independent sets are called blocks.
	It can be solved approximately using an algorithm for the max-min-sharing-problem by Grigoriadis et al.~\cite{GrigoriadisKPV01}. Combined with an algorithm that finds a basic solution given any feasible solution by Ke et al.~\cite{ke2008fast} the basic solution for the relaxed version of the LP can be found in $\log(m)^{\Oh(1)} \cdot  (1/\eps)^{\Oh(1/\eps)}$ operations. 
	The solver for max-min-resource-sharing described in \cite{GrigoriadisKPV01} has a time complexity of $\Oh(M(\eps^{-2} + \ln(M))(B\cdot \mathcal{ABS} + M\ln(M)))$ and generates $\Oh(M(\eps^{-2} + \ln(M)))$ non-zero components, where $M$ is the number of functions, $K$ is the number of blocks and $\mathcal{ABS}$ is the time complexity of a so-called block-solver.
	On the other hand, the algorithm to find the basic solution described in \cite{ke2008fast} has a time complexity of $\Oh(m^{1.5356}n)$ where $m$ is the number of constraints and $n$ is the number of variables.

	The above LP is transformed to a max-min-resource-sharing problem such that the constraints (\ref{eq:APTAS:SRCS:jobs}) become the functions $f_1, \dots, f_{|\jobs_S'|}$, while the constraints (\ref{eq:APTAS:SRCS:smallWide}) and (\ref{eq:APTAS:SRCS:smallLayer}) each define a set $B_j$, $j \in [|\startPoints|+1]$.
	As a consequence, we get that $M = |\jobs_S'|$ and $B = |\startPoints|+1$.
	In the block problem, we search for a configuration that maximizes a given function.
	This problem resembles the problem \acl{UKKP}, and therefore can be solved approximately in $|\jobs_S'|^{O(1)}\cdot (1/\eps)^{\Oh(1)}$, see e.g. \cite{MastrolilliH06}.
	
	Note that the solution generated for the max-min-resource-sharing problem is only an approximate solution, because the block-solver can only be solved approximately, i.e. we find a $\lambda \geq (1-\rho)\lambda^*$ for any chosen $\rho \in (0,1)$.
	When transforming this solution to a solution to the above LP, we have to relax the right hand side of the equations (\ref{eq:APTAS:SRCS:smallWide}) and (\ref{eq:APTAS:SRCS:smallLayer}) by the factor $(1+\eps)$ to guarantee the complete schedule of all the jobs. 
	In the last step, we transform the solution for the relaxed LP to a basic solution using the algorithm described in~\cite{ke2008fast}.  
	
\end{proof}

\paragraph*{Reducing the Number of Configurations}
In the next step, we reduce the number of non-zero components some further since scheduling the large jobs inside these configurations would add up to $\Oh(\mu T \log(m)/\eps^2\delta)$ to the makespan, which is too large.
First, we partition the set of rounded small jobs $\jobs_S'$ into wide and narrow jobs.
We say a job $j \in \jobs_S'$ is wide, if $\jR{j} \geq \eps R$  and narrow otherwise.
Let $\jobs_{SW}'$ be the set of wide and $\jobs_{SN}'$ be the set of narrow jobs.

To reduce the number of non-zero components, we use the same techniques as in the AFPTAS described in~\cite{journals/talg/JansenMR19}, i.e., we introduce windows and generalized configurations. 
Let $x_{\pre}$ be the solution for the linear program generated by the algorithm from Lemma \ref{lma:SRCS:APTAS:FirstLPSolutionSmallJobs}.
Let $\conf_{\LP}$ be the set of configurations that have a non-zero component in the considered solution and let $\conf_{\LP, s}$ be the set of non-zero component configurations for layer $s$.
For a configuration $C \in \conf_{\LP, s}$, we denote by $\pTs{C} := x_{s,C}$ the total processing time it is processed inside the layer $s$.
For a configuration $C \in \conf_{\LP}$, we define $C{|_{\jobs_{SW}'}}$ as the configuration where we removed all the narrow jobs.
Furthermore, we denote by $\conf_{\LP,s,W}$ the set of all configurations in $\conf_{\LP,s}$ that are reduced to their wide jobs.
A \textit{window} $w=(w_r,w_m)$ is a pair consisting of a resource requirement $\jR{w} = w_r$ and a number of machines $m(w) = w_m$. 
As for a configuration, the total time a window is processed inside a layer $s$ is denoted as $\pTs{w}$ and is called its processing time.
At each point of time in a given window $w$, there can be processed $m(w)$ jobs in parallel with a summed up resource requirement $\jR{w}$. 
For windows $w_1$, $w_2$, we write $w_1 \leq w_2$ if and only if $\jR{w_1} \leq \jR{w_2}$ and $m(w_1) \leq m(w_2)$. 
A \textit{generalized configuration} $(C,w)$ for a layer $s$ is a pair consisting of a configuration $C \in \conf_{\LP,s,W}$ and a window $w$ such that $m(w) \leq m_s - m(C)$ and $\jR{w} \leq R_s - \jR{C}$. 
For a configuration $C \in \conf_{\LP,s,W}$, we define by $w_s(C) := (R_s-\jR{C},m_s-m(C))$ the \textit{main window} for $C$. 
We define $\wind_s$ as the set of all main windows for the configurations in $\conf_{\LP, s}$, and $\wind$ be the set of all the generated widows.

Consider the following linear program:
$\LP_{W}(\startPoints_{\LLS}, \jobs_S'):$
\begin{align}
\sum_{s \in \startPoints_{\LLS}}\sum_{C \in \conf_{\LP,s,W}}\sum_{\substack{w \in \wind_s \\  w \leq w_s(C)}} C(j)x_{C,w,s} 
&= \pT{j}
&\forall j \in \jobs_{SW}' 
\label{eq:APTAS:SRCS:wideJobs}\\
\sum_{s \in \startPoints_{\LLS}}\sum_{w \in \wind_s} y_{j,w,s}  
&= \pT{j}
&\forall j \in \jobs_{SN}'  
\label{eq:APTAS:SRCS:narrowJobs}\\
m(w)\sum_{\substack{C \in \conf_{\LP,s,W}\\w_s(C) \geq w}} x_{C,w,s}  
& \geq \sum_{j \in \jobs_{SN}} y_{j,w,s}
&\forall s \in \startPoints_{\LLS}, w \in \wind_s
\label{eq:APTAS:SRCS:windowProcessingTime}\\
\jR{w}\sum_{\substack{C \in  \conf_{\LP,s,W} \\  w_s(C) \geq w}} x_{C,w,s}  
& \geq \sum_{j \in \jobs_{SN}} \jR{j}y_{j,w,s}
&\forall s \in \startPoints_{\LLS}, w \in \wind_s
\label{eq:APTAS:SRCS:windowResourceRequirement}
\end{align}
\begin{align}
x_{C,w,s}  &\geq 0&  \forall s \in \startPoints_{\LLS},C \in \conf_{\LP,s,W},  w \in \wind\\
y_{j,w,s}  &\geq 0&  \forall s \in \startPoints_{\LLS},w \in \wind_s,  j \in \jobs_{SN}
\end{align}

The variable $x_{C,w,s}$ denotes the processing time of the generalized configuration $(C,w)$ in the layer $s$ and
the value $y_{j,w,s}$ indicates which amount of job $j$ is processed in window $w$ in the layer $s$. 
Inequalities (\ref{eq:APTAS:SRCS:wideJobs}) and (\ref{eq:APTAS:SRCS:narrowJobs}) ensure that for each job there is enough processing time reserved, while equalities (\ref{eq:APTAS:SRCS:windowProcessingTime}) and (\ref{eq:APTAS:SRCS:windowResourceRequirement}) ensure that in each window there is enough space to schedule the contained jobs. 

Given a solution $(x,y)$ to $\LP_{W}$, we define 
\[\pTs{x} := \sum_{C \in \conf_{\LP,s,W}}\sum_{\substack{w \in \wind_s \\  w \leq w_s(C)}} x_{C,w,s},\] 
which is the processing time of $(x,y)$ in the layer $s$, and 
\[\pTs{w,x} := \sum_{\substack{C \in \conf_{\LP,s,W}\\ C(w) \geq w}}x_{C,w,s}\] 
which is the summed up processing time of a window $w \in \wind$ in $x$ in layer $s$.

\begin{lemma}
	\label{lma:SRCS:APTAS:firstWindowSolution}
	Given a solution $x_{\pre}$ to the relaxed version of $LP_{\startPoints}$, we can find a solution $(\tilde{x},\tilde{y})$ to the linear program $\LP_{W}$, which fulfills
	\begin{align}
	\pTs{x_{\LLS,C}} & \leq (1+\eps) 3\eps \greedyMakespan \label{eq:SRCS:APTAS:processingtime1}\\
	\pTs{x_{s',C}} &\leq (1+\eps)\eps\delta\greedyMakespan &\forall s' \in S\label{eq:SRCS:APTAS:processingtime2}
	\end{align}
\end{lemma}

\begin{proof} 
	To generate this solution, we look at each layer $s$ and each configuration $C \in \conf_{\LP,s,W}$ and sum up the processing time of each configuration $C' \in \conf_{\LP,s}$, which is reduced to $C$, i.e., $C'|_{\jobs_{SW}'} = C$. 
	Building a generalized configuration, we combine $C$ with its main window $w(C) \in \wind_{\pre}$. 
	More precisely for each $C \in \conf_{\LP,s,W}$, we define
	\[\tilde{x}_{(C,w(C))} := \sum_{\substack{C' \in \conf_{\LP,s}\\ C'{|_{\jobs_{SW}'}} = C}} x_{\pre_{C',s}}.\]
	Equations (\ref{eq:SRCS:APTAS:processingtime1}) and (\ref{eq:SRCS:APTAS:processingtime2}) hold for this choice for $\tilde{x}$ since the processing time of each configuration is added to exactly one generalized configuration and $x_{\pre}$ fulfills equations (\ref{eq:APTAS:SRCS:smallWide}) and (\ref{eq:APTAS:SRCS:smallLayer}).
	With a similar argument, one can see that inequality (\ref{eq:APTAS:SRCS:wideJobs}) holds since $x_{\pre}$ fulfills equation (\ref{eq:APTAS:SRCS:jobs}). 	
	
	On the other hand, we have to ensure that inequalities (\ref{eq:APTAS:SRCS:narrowJobs}) to (\ref{eq:APTAS:SRCS:windowResourceRequirement}) hold. 
	For this purpose, we look at each layer $s$ and each configuration $C \in \conf_{\LP,s}$ and consider the reduced configuration $C|_{\jobs_{SW}}$ and its main window $w := w(C|_{\jobs_{SW}})$. 
	For each job $j \in \jobs_N$, we add its processing time in $C$, which is given by $C(j)(x_\pre)_{C,s}$, to the window $w$. 
	More precisely for each layer $s$, each window $w \in \wind_s$, and each job $j \in \jobs_{SN}$, we define
	\[\tilde{y}_{j,w,s} :=  \sum_{\substack{C \in \conf_{\LP,s}\\ w({C|_{\jobs_{SW}'}}) = w}} C(j)x_{\pre_{C,s}}.\]
	Since the configuration $C$ was valid and equations (\ref{eq:APTAS:SRCS:smallWide}) (\ref{eq:APTAS:SRCS:smallLayer}), and (\ref{eq:APTAS:SRCS:jobs}) hold for $x_{\pre}$ the equations (\ref{eq:APTAS:SRCS:narrowJobs}) to (\ref{eq:APTAS:SRCS:windowResourceRequirement}) hold for $(\tilde{x},\tilde{y})$ as a direct consequence.
\end{proof}

Let $(\tilde{x},\tilde{y})$ be the solution to $\LP_{W}$ generated for $x_{\pre}$  by Lemma \ref{lma:SRCS:APTAS:firstWindowSolution}.
Note that the number of non-zero components in $(\tilde{x},\tilde{y})$ is bounded by $\Oh( \log(m)/\eps^2\delta)$ since in $x_{\pre}$ there are at most $\Oh( \log(m)/\eps^2\delta)$ non-zero components and, for each of these components, we generate at most two non-zero components in $(\tilde{x},\tilde{y})$.
Furthermore, since each configuration in $\conf_{\LP,W}$ contains at most $1/\eps$ jobs and there are at most $\Oh(\log(1/\eps)/\eps)$ different wide jobs, the number of configurations and their corresponding main windows in $(\tilde{x},\tilde{y})$ is bounded by $\Oh(|\startPoints_{\LLS}| \cdot (\log(1/\eps)/\eps)^{1/\eps}) = \Oh((1/\eps\delta) \cdot (\log(1/\eps)/\eps)^{1/\eps}) \leq  \Oh(1/\eps^{2/\eps+1}\delta)$.
As a consequence, a basic solution to $\LP_{w}$, where we add the equations (\ref{eq:SRCS:APTAS:processingtime1}) and (\ref{eq:SRCS:APTAS:processingtime2}) has at most $|\jobs_{S,W}'| + |\jobs_{S,N}'| + \Oh(1/\eps^{2/\eps+1}\delta)$ non-zero components and can be computed in $(1/\eps)^{1/\eps^{\Oh(1/\eps)}}$.
Since each job in $\jobs_{S,N}'$ uses at least one non zero component in $y$ and $|\jobs_{SW}'| \leq \Oh(1/\eps^2)$, this basic solution uses at most $\Oh(1/\eps^{2/\eps+1}\delta)$ generalized configurations or fractionally scheduled jobs from the set $\jobs_{SN}'$.

At this point, we could proceed to find an integral solution since the number of non-zero components for configurations and fractional scheduled jabs is bounded by $\Oh_{\eps}(1)$.
However, the $\Oh(1/\eps^{2/\eps+1}\delta)$ generalized configurations or fractionally scheduled jobs lead to a running time of the form $\Oh(n\log(1/\eps)) + (m\log(R)/\eps)^{1/\eps^{\Oh(1/\eps)}}$.
In the following steps, we will reduce the number of generalized configurations or fractionally scheduled jobs to be in $1/\delta\eps^{\Oh(1)}$ and, thus, reduce the running time to $\Oh(n\log(1/\eps)) + (m\log(R)/\eps)^{1/\eps^{\Oh(1/\eps)}}$.

\begin{lemma}
	Given a solution $(\tilde{x},\tilde{y})$ to $\LP_{W}(\startPoints_{\LLS}, \jobs_S')$, we can find a solution $(\bar{x},\bar{y})$ to $\LP_{W}(\startPoints_{\LLS}, \jobs_S')$ with $\pT{\bar{x}} \leq (1+\eps)\pT{\tilde{x}}$ and has at most $\Oh(1/\eps^3\delta) + |\jobs_{S,N}|$ non zero components.
	This solution can be found in at most $\Oh(\log(m)^{3})\cdot(1/\eps)^{\Oh(1/\eps)}$ operations.
\end{lemma}
\begin{proof}
	Consider a layer $s \in \startPoints$ and the set of windows $\wind_s$ occurring in this layer in the considered solution $(\tilde{x},\tilde{y})$. 
	At the moment there can be up to $\Oh((\log(1/\eps)/\eps)^{1/\eps})$ different windows in $\wind_s$.
	We will reduce this number to at most $\Oh(1/\eps^2)$. 
	
	We partition the windows by the size of $m(w)$ for each window $w \in \wind_s$.
	Since the generalized configuration can contain at most $1/\eps$ wide jobs, there are at most $1/\eps +1$ different values for $m(w)$ in $\wind_s$.
	However, there is only one window that has the largest value $m(w)$ since there is no wide job scheduled next to this window.
	For the residual $1/\eps$ sets, we stack the generalized configuration corresponding to the windows in sorted order such that the widest window is at the bottom and the most narrow one at the top.
	The partition and the corresponding stack of windows can be found in $\Oh((\log(1/\eps)/\eps)^{2/\eps})$.
	
	Let $P$ be the processing time of one of these stacks.
	We shift down the windows by $\eps P$ while the configuration part of the generalized configuration is not moved.
	We partition the stack into segments of the same processing time $\eps P$ and assign the jobs contained in a window from segment $i$ to the most narrow window in the segment below $i$ and change window corresponding to the generalized configuration accordingly.
	The jobs from the bottom most segment are assigned to the window $(R,m)$, which will be processed at the end of the schedule.
	This reassignment of jobs to windows can be done in $\Oh((\log(1/\eps)/\eps)^{2/\eps} \cdot |\jobs_{S,N}|)$.
	
	Since we remove windows with total processing time at most $\eps P$ from each stack, we remove windows with processing time at most $\eps\pT{\tilde{x}}$ total.
	After this rounding for each stack, we have at most $1/\eps$ windows left. Since we have at most $1/\eps$ stacks per layer and at most $\Oh(1/\eps\delta)$ layer, the total number of windows is reduced to $\Oh(1/\eps^3\delta)$.
	
	Given this solution, we transform it to a basic solution using the algorithm by Ke et al.~\cite{ke2008fast} in at most $\Oh((|\startPoints|\cdot |\wind| +|\jobs_{S}|)^{1.5356}|\startPoints|\cdot |\wind|\cdot |\jobs_{S}|/\eps) = \Oh(\log(m)^{3})\cdot(1/\eps)^{\Oh(1/\eps)}$.
\end{proof}

\paragraph*{Generating an Integral Schedule}
In the final step of the algorithm, after the binary search framework has found the correct value for $T'$, we use this relaxed solution to schedule the original jobs in $\jobs_S$.
Again, we use the same techniques as described in \cite{journals/talg/JansenMR19}.
More precisely we use the same algorithm as described in Lemma 2.5 of that paper:

\begin{lemma}[Lemma 2.5 of \cite{journals/talg/JansenMR19}]
	Let $(\bar{x}, \bar{y})$ be a basic solution to the linear program, such that the total number off fractional scheduled narrow jobs and used configurations is bounded by $K$. 
	There is a solution which places the jobs integrally and has a makespan of at most
	$(1+\Oh(\eps))P(\bar{x}) + (1+\eps + K)p_{\max}$-
\end{lemma}

Since the total number of configurations and fractionally scheduled small jobs is bound by $\Oh(1/\eps^3\delta)$ in the given solution to the linear program, 
and $p_{\max} \leq \mu T \leq \delta\eps^4T$, the schedule is extended by at most $\Oh(\eps)T$ when scheduling the jobs integrally.

\section{Conclusion}
In this paper, we presented an algorithm for \acl{resource} with absolute approximation ratio $(3/2 +\eps)$, which closes the gap between inapproximation result $3/2$ and the best algorithm.
Furthermore, we presented an AEPTAS for this problem.

It would be interesting to see if the techniques used in these algorithms can be extended to the moldable case of the problem; i.e., the case where the processing time of a job depends on the number of resources that are assigned to it.

\bibliography{lowerBound}

\end{document}